\newcommand{\version}{December 20, 2010}
\documentclass[11pt,a4paper,english,pdftex]{article}

\usepackage[bindingoffset=0.3cm,textheight=22.5cm,hdivide={2.7cm,*,2.7cm}, vdivide={*,22cm,*}]{geometry}
\usepackage[pdftex,bookmarksnumbered=true,breaklinks=true]{hyperref}

\usepackage{amsmath,amsfonts,amssymb,slashed}

\IfFileExists{empheq.sty}
  {\usepackage{empheq}
   \newenvironment{fboxalign}
   {\empheq[box=\fbox]{align}}
   {\endempheq}}
  {\typeout{Package empheq.sty was not found, using alternative macros.}
   \newenvironment{fboxalign}
   {\align}
   {\endalign}}

\usepackage[numbers,square,comma,sort&compress]{natbib}

\IfFileExists{dsfont.sty}
	{\usepackage{dsfont}
         \newcommand{\id}{\mathds{1}}}
	{\typeout{Package dsfont.sty was not found, using alternative macros.}
         \let\mathds=\mathbb
         \newcommand{\id}{\mbox{1 \kern-.59em {\rm l}}}}
\let\one=\id

\makeatletter

         \@addtoreset{equation}{section}
\makeatother

\let\startappendix=\appendix
\newcommand{\nocontentsline}[3]{}
\newcommand{\tocless}[3]{\bgroup\let\addcontentsline=\nocontentsline#1{#2}#3\egroup}

\newcommand{\Appendix}[1]{
  \refstepcounter{section}
  \section*{Appendix \thesection:\hspace*{1.5ex} #1}
  \addcontentsline{toc}{section}{Appendix \thesection}
}
\newcommand{\SubAppendix}[2]{\tocless\subsection{#1}{\label{#2}}}

\newtheorem{theorem}{Theorem}
\newtheorem{lemma}[theorem]{Lemma}

\newcommand{\qed}{\nobreak \ifvmode \relax \else
      \ifdim\lastskip<1.5em \hskip-\lastskip
      \hskip1.5em plus0em minus0.5em \fi \nobreak
      \vrule height0.75em width0.5em depth0.25em\fi}

\newenvironment{proof}[1][Proof]{\begin{trivlist}
\item[\hskip \labelsep {\bfseries #1}]}{\qed\end{trivlist}}

\newcommand{\be}{\begin{equation}}
\newcommand{\ee}{\end{equation}}
\newcommand{\eq}[1]{(\ref{#1})}

\def\nn{\nonumber}
\def\bea{\begin{eqnarray}}
\def\eea{\end{eqnarray}}

\def\beqa{\begin{eqnarray}} 
\def\eeqa{\end{eqnarray}} 
\def\beq{\begin{equation}} 
\def\eeq{\end{equation}}

\def\Tr{{\rm Tr}}

\def\a{\alpha}          
\def\b{\beta}           

\def\k{\kappa}
\def\l{\lambda} \def\L{\Lambda} 

\def\r{\rho}
\def\t{\tau}
\def\vp{\varphi}

\def\cA{{\cal A}}  
  \def\cF{{\cal F}}
 \def\cH{{\cal H}} 
  \def\cL{{\cal L}}
\def\cM{{\cal M}} \def\cN{{\cal N}} \def\cO{{\cal O}}
\def\cP{{\cal P}}

\newcommand{\sD}{\slashed{D}}

\newcommand{\Qt}{\tilde{Q}}

\newcommand{\lt}{\tilde{l}}
\newcommand{\kt}{\tilde{k}}

\def\qt{\tilde{q}}

\newcommand{\bG}{\bar{G}}

\def\L{{\bf l}}

\newcommand{\R}{\mathds{R}}

\newcommand{\Z}{\mathds{Z}}

\def\bit{\begin{itemize}}
\def\eit{\end{itemize}}

\def\({\left(}
\def\){\right)}
\def\diag{\mbox{diag}}

\def\pa{\partial} \def\del{\partial}

\newcommand{\tr}{\mbox{tr}}

\def\bcomment#1{}

\def\LNC{\L_{\textrm{NC}}}

\newcommand{\eqn}[1]{\begin{align}
                       #1
                       \end{align}
}

\newcommand{\Sin}[1]{\sin\!\left(\! \frac{ #1 }{2} \!\right)}

\newcommand{\Sinn}[2]{\sin^{#2}\!\left(\! \frac{ #1 }{2} \!\right)}

\newcommand{\SiN}[3]{\sin^{#3}\!\left(\! \tfrac{ #1 }{ #2 } \!\right)}

\newcommand{\Int}[1]{\int\!\!\frac{d^4 #1}{(2\pi)^4}}
\newcommand{\Intt}[1]{\int\!\!\frac{d^4 #1}{(2\pi\LNC^2)^2}}

\newcommand{\dott}[2]{#1 \cdot #2}

\newcommand{\moyal}{Groenewold-Moyal}

\newcommand{\nc}{non-com\-mu\-ta\-tive}

\newcommand{\naiv}{na\"ive}

\newcommand{\eqnref}[1]{Eqn.~(\ref{#1})}		
\newcommand{\secref}[1]{Section~\ref{#1}}		
\newcommand{\appref}[1]{Appendix~\ref{#1}}		
\newcommand{\inv}[1]{\frac{1}{#1}}				
\newcommand{\tinv}[1]{\tfrac{1}{#1}}
\newcommand{\co}[2]{[#1,#2]}						
\newcommand{\aco}[2]{[#1,#2]_+}						
\newcommand{\inttx}{\int\!\! \frac{d^4x}{(2\pi)^2}}	
\newcommand{\intg}{\int\!d^4x\sqrt{g}\,}				
\newcommand{\intG}{\int\!d^4x\sqrt{G}\,}				

\newcommand{\ba}{\bar \alpha}
\newcommand{\bt}{\bar t}

\renewcommand{\a}{\alpha}
\renewcommand{\b}{\beta}
\newcommand{\g}{\gamma}
\renewcommand{\d}{\delta}
\newcommand{\e}{\epsilon}
\newcommand{\vare}{\varepsilon}

\renewcommand{\th}{\theta}
\newcommand{\thb}{\bar\theta}

\renewcommand{\l}{\lambda}
\newcommand{\m}{\mu}
\newcommand{\n}{\nu}

\renewcommand{\r}{\rho}
\newcommand{\s}{\sigma}
\renewcommand{\t}{\tau}

\newcommand{\vph}{\varphi}

\newcommand{\G}{\Gamma}

\newcommand{\Th}{\Theta}

\renewcommand{\L}{\Lambda}
\newcommand{\Leff}{\Lambda_{\textrm{eff}}}

\renewcommand{\Xi}{\Xi}

\title{\begin{flushright}
       \small{UWThPh-2010-14}
       \end{flushright}
\vspace{3em}
Heat kernel expansion and induced  action\\[1ex]
 for the matrix model Dirac operator }
\author{Daniel N. Blaschke\footnote{daniel.blaschke@univie.ac.at}~, Harold Steinacker\footnote{harold.steinacker@univie.ac.at}~, Michael Wohlgenannt\footnote{miw@hep.itp.tuwien.ac.at}}

\date{\version}
\begin{document}

\maketitle

\begin{center}
\renewcommand{\thefootnote}{\fnsymbol{footnote}}
\textit{\footnotemark[1]\footnotemark[2]\footnotemark[3]University of Vienna, Faculty of Physics\\
Boltzmanngasse 5, A-1090 Vienna (Austria)
\\[0.3em]
\footnotemark[3]Vienna University of Technology, Institute for Theoretical Physics\\
Wiedner Hauptstrasse 8-10, A-1040 Vienna (Austria)}
\vspace{0.5cm}
\end{center}%
\begin{abstract}

We compute the quantum effective action induced by integrating out fermions in
Yang-Mills matrix models on a 4-dimensional background, expanded in powers of 
a gauge-invariant UV cutoff.
The resulting action is recast into the form of generalized matrix models, manifestly preserving the 
$SO(D)$ symmetry of the bare action. 
This provides {\nc} (NC) analogs of the Seeley-de Witt coefficients
for the emergent gravity which arises on NC branes, 
such as curvature terms. From the gauge theory point of view, this provides strong
evidence that the {\nc} $\cN=4$ SYM has a hidden $SO(10)$ symmetry even at the quantum level,
which is spontaneously broken by the space-time background. 
The geometrical view proves to be very powerful, and allows to predict 
non-trivial loop computations in the gauge theory.

\end{abstract}

\newpage
\tableofcontents

\section{Introduction}

The aim of this work is to study the quantum effective action induced by integrating out
fermions in Yang-Mills-type matrix models, using heat-kernel techniques.

The matrix models under consideration were introduced originally in the context of 
string theory, where they are viewed as non-perturbative definitions of superstrings
on flat $\R^{10}$. On the other hand, one can also consider {\nc} brane solutions 
(or generic brane configurations) in these matrix models. It is well-known that this leads to 
{\nc} gauge theory on the branes \cite{Aoki:1999vr}. Upon closer examination, 
it turns out that the $U(1)$ sector of this gauge theory can be understood in terms of geometry, 
and encodes an effective gravity theory
on the brane \cite{Steinacker:2007dq,Steinacker:2008ri,Steinacker:2010rh}. 
This gives a novel and direct gauge/gravity relation 
specific to the {\nc} setting.
However, to understand the dynamics of this emergent gravity, 
quantum effects must be taken into account. 
The reason is that quantization on non-trivial backgrounds leads to induced 
gravitational actions \cite{Sakharov:1967}. On a semi-classical level, this can be understood in terms 
of Seeley-de Witt coefficients. 

In order 
to properly understand the physical content of these models, one must study their
quantization at the level of matrix models, taking their {\nc} nature into account.
As an important step in this program, we study in this paper the quantum effective actions due to fermions,
and show that it can be recast in the form of generalized matrix models\footnote{While the quantization of 
Yang-Mills matrix models has been studied before
e.g. in \cite{Ishibashi:1996xs,Aoki:1999vr,Imai:2003jb,CastroVillarreal:2005uu,Steinacker:2003sd,Azuma:2005pm}, 
the results available so far are not very explicit, and not in the form of generalized matrix models.}. 
This is the language 
which is appropriate to understand the geometric aspects of these models. Moreover, it turns out to 
provide a powerful and predictive new tool for the description of {\nc} gauge theory, 
notably for the maximally supersymmetric $\cN=4$ gauge theory. 

We thus consider fermions coupled to a generic {\nc} background in the matrix model.
The matrices $X^a$ can be
considered as perturbations around the {\moyal} quantum plane $\R^4_\theta$. 
Hence, a scale of non-commutativity $\LNC$ is introduced, as will be explained in more detail in subsequent sections. 
We will then compute the quantum effective action induced by integrating out fermions. To this end we make a 
heat kernel expansion for $\sD^2$. For this expansion (more generally for the 
quantum effective action) to make sense, it is essential to consider
a UV cutoff $\L$ which satisfies the bound
\be
p^2 \L^2 \ll \LNC^4
\label{IR-condition}
\ee
for any momentum scale $p$ in the background. 
This condition \eq{IR-condition} guarantees that the UV/IR mixing is ``mild'',
so that the geometrical interpretation in terms of 
emergent gravity is expected to apply \cite{Steinacker:2008}.
Only in that case there is indeed a meaningful expansion of the fermionic effective action.
In contrast, the induced action appears to be pathological in the limit $\L \to \infty$. 
The physical motivation for such a 
``low'' cutoff comes from supersymmetric matrix models such as the IKKT model~\cite{Ishibashi:1996xs}
i.e. $\cN=4$ SYM, 
where such a cutoff could be provided by the SUSY breaking scale. Only such supersymmetric
models are expected to be well-behaved upon quantization.

Using this setup, we compute the heat kernel expansion of the 
Dirac operator with a generic 4-dimensional NC background, using a perturbative (Duhamel) expansion
and a covariant cutoff.
This allows to systematically obtain all terms in the induced action with 
any given operator dimension. Using the language of {\nc} gauge theory,
we compute in this way the complete induced action including all terms
of operator dimension 6 or less. Gauge invariance is recovered upon collecting 
the appropriate contributions. 
This results in an effective action for a {\nc} $U(1)$ gauge theory, incorporating
UV/IR mixing in a controlled way. 
In a second, crucial step, we show that this action can be recast in the form of an effective 
generalized matrix model. The resulting action has a manifest $SO(D)$ symmetry, which is completely hidden
in the gauge theory language. This is a remarkable and non-trivial result, which is very natural
from the geometric point of view of emergent gravity.

It is interesting to compare this with the results of \cite{Gayral:2006vd}, 
where the asymptotic expansion of a similar operator on the {\nc} torus
(note: the ``infinite-dimensional'' one, not the fuzzy torus) was studied. 
The result was found to depend crucially on number-theoretical properties of 
the non-commutativity parameter, and for a certain class of $\theta$ an asymptotic expansion 
of the heat-kernel was found to be quite similar to the commutative case.
However in these previous papers \cite{Gayral:2006vd,Vassilevich:2005vk}, 
the condition \eq{IR-condition} is not satisfied, so that their results are not related with ours. 
In contrast, the results of the present paper are robust and independent of any specific
properties of $\theta$. They are expected to apply also to 
the case of NC tori, assuming a similar cutoff.

The geometrical point of view of the matrix model 
and the $SO(D)$ symmetry turn out to be very powerful and predictive 
tools for the quantization. 
Indeed the complete ``vacuum energy'' term in the effective matrix model
can be inferred from a 2-line computation 
combined with geometrical insights. This completely predicts a series of highly 
non-trivial loop results in the gauge theory picture, in complete agreement
with our computations. The $SO(D)$ symmetry contains both space-time rotations 
as well as internal rotations, but 
--- most remarkably --- mixes the space-time and internal sector. 
More specifically, it relates the gauge fields with the scalar fields. 
It is somewhat related to extended supersymmetry, i.e. $\cN=4$ SUSY for $D=10$, 
since the internal symmetry $SO(6)$ coincides with the $SO(6)$ R-symmetry of $\cN=4$ SUSY.
However, the $SO(D)$ symmetry goes beyond SUSY, and leads to additional
restrictions on the effective action. This is manifest in the 
present computation of the heat kernel associated to the Dirac operator,
which amounts to a one-loop computation involving only fermions.
It underlines the fact that NC gauge theory is richer than commutative gauge theory.
We expect that this symmetry also extends to the non-Abelian $SU(N)$ sector, which
is less affected by UV/IR mixing~\cite{Armoni:2000xr} and should reduce to the commutative gauge theory
in a suitable limit. It is thus
natural to suspect some relation to the recent results on hidden symmetries of SUSY gauge theory 
\cite{Drummond:2009fd,Beisert:2010gn,Bern:2004bt,ArkaniHamed:2010kv}, 
although the specific relation is not clear at present.

Finally we should point out that although the induced action due to fermions is 
of course only one piece of the complete effective action, it is closely related to Connes spectral action
\cite{Chamseddine:1996zu}, which is most naturally computed using the heat kernel expansion. 
Our results thus provide the leading contributions to this spectral action for the matrix model  Dirac operator. 
The missing bosonic integral can be viewed as an integral over the backgrounds, 
thus providing a measure for the integration over geometries.

This paper is organized as follows. We first recall in \secref{sec:fermion-action} 
the basic definition of the Dirac operator and the geometrical interpretation of branes in
matrix models. The essential
steps and the novel aspects of the heat kernel expansion and the effective action
are explained in \secref{sec:strategy}.
The detailed and lengthy computation of the effective gauge theory action is 
given in Sections \ref{sec:heatkernel-details} and \ref{sec:effgaugethy}; 
a reader not interested in the details can jump to the main results
which are \eq{potential-collect}, \eq{2-field-collect-curv1} and 
\eq{2-field-collect-curv2}.
In \secref{sec:effMM} this effective gauge theory is rewritten as an effective 
matrix model, culminating in \eq{full-action-MM}
which is the main result of this paper.

\section{The fermionic action}
\label{sec:fermion-action}
Our starting point is the matrix model 
fermion action~\cite{Ishibashi:1996xs,Steinacker:2008a,Klammer:2009dj,Steinacker:2010rh} in Euclidean space 
\begin{align}
S_\Psi&= (2\pi)^2\Tr\Psi^\dagger\sD{\Psi}
= (2\pi)^2 \Tr\Psi^\dagger\g_a\co{X^a}{\Psi}
\,,
\label{Dirac-action}
\end{align}
where $X^a,\,\, a= 1,2,\ldots,D$ are Hermitian matrices, and
\be
\sD\Psi:=\g_a\co{X^a}{\Psi}
\,.
\label{Dirac-def}
\ee
Latin indices are pulled up/down with the flat embedding metric $g_{ab} = \d_{ab}$ of $\R^D$,  
and the $\g$-matrices form the usual Clifford algebra $\aco{\g_a}{\g_b}= 2\d_{ab}$.
The matrices can be interpreted as operators on a separable Hilbert space $\cH$.
This action is invariant under the following  symmetries:
\be
\begin{array}{lllll}
X^a \to U^{-1} X^a U\,,\quad  &\Psi \to U^{-1} \Psi U\,,\quad  &  U \in U(\cH)\,,\quad  & \mbox{gauge invariance,}  \\
X^a \to \L(g)^a_b X^b\,,\quad  &\Psi_\a \to \tilde \pi(g)_\a^\b \Psi_\b\,,\quad  & g \in \widetilde{SO}(D)\,,\; 
  & \mbox{rotational symmetry,}  \\
X^a \to X^a + c^a \one\,,\quad & \quad & c^a \in \R\,,\quad  & \mbox{translational symmetry,}
\end{array}
\label{transl-inv}
\ee
where the tilde indicates the corresponding spin group.
The induced effective action $\G[X]$ is defined as
\begin{align}
e^{-\G[X]}&=\int\! d\Psi d\Psi^\dagger e^{-S_\Psi}
=(\textrm{const.}) \exp\left(\inv{2}\Tr\log(\sD^2)\right)
\,,\nn\\
\sD^2\Psi&=\g_a\g_b\co{X^a}{\co{X^b}{\Psi}}
\,.
\end{align}
As a background,
we will consider matrices $X^a$ which define 4-dimensional NC spaces (branes) embedded in $\R^D$. 
The simplest example is the {\moyal} quantum plane $\R^4_\theta$, defined by $X^a = (\bar X^\mu,0)$   
where the $\bar X^\mu$  satisfy
\be
[\bar X^\mu,\bar X^\nu] = i \bar\theta^{\mu\nu}
= i\LNC^{-2} \left(\begin{array}{cccc}  0 & \a & 0 & 0 \\
                                 -\a & 0 & 0 & 0 \\
                                 0 & 0 & 0 & \pm\a^{-1}  \\
                                  0 & 0 & \mp\a^{-1} & 0 \end{array}\right)\, ;
\label{theta-standard-general-E}
\ee
we can assume this standard form of $\bar\theta^{\mu\nu}$ using a $SO(4)$
transformation if necessary. 
More generally, we assume that the matrices can be split as
\be
X^a = (X^\mu,\phi^i(X^\mu))
\label{eq:matrix-splitting}
\ee
where $X^\m$, $\m = 1,\ldots,4$ are considered as independent quantized coordinate ``functions''
satisfying generic commutation relations $[X^\mu,X^\nu] = i \theta^{\mu\nu}(X)$,
while the $\phi^i(X^\mu)$ are ``smooth'' functions of these coordinates.
As explained in \cite{Steinacker:2010rh}, there are two different interpretations of 
such a background.
First, the action \eq{Dirac-action} can be viewed as describing
fermions propagating on a (generally curved) brane $\cM^4\subset \R^D$, with effective metric
\begin{align}
G^{\mu\nu} &= \LNC^4\theta^{\mu\mu'}\theta^{\nu\nu'} g_{\mu\nu} \,,  &
g_{\mu\nu} &= \del_\mu x^a \del_\nu x^b g_{ab}
\,, \label{eq:def-Gges}
\end{align}
in the semi-classical limit where $X^a\sim x^a$ and $\theta^{\mu\nu}(X) \sim \theta^{\mu\nu}(x)$. 
Note that the fixed background metric $g_{ab} = \d_{ab}$ of $\R^D$
defines a scale, and all subsequent quantities 
are measured with respect to this scale. 
In that sense, $X^a$ has dimension length, 
and $\theta^{\mu\nu}$ encodes the non-commutativity scale
\be
\LNC^4 = \det{\th^{-1}_{\m\n}} \,,
\ee
which may depend on $x$; note also $\det{g_{\m\n}} = \det{G_{\m\n}}$.
Second, \eq{Dirac-action} can be viewed as describing fermions on
$\R^4_\theta$ coupled to {\nc} gauge fields and scalars, which arise through \eq{eq:general-X}. 
Hence the flat resp. free case
corresponds to the {\moyal} quantum plane $\R^4_\theta$.
To proceed, we need to regularize\footnote{Other regularizations might be easier to work with,
however we choose a covariant cutoff in order to ensure that both
gauge invariance as well as the $SO(D)$ symmetry are preserved. 
In general, we will be cavalier 
about precise mathematical definitions, aiming at physically meaningful and finite results.
In particular we will not worry about the ``trivial'' divergence
associated to the infinite volume of $\R^4$, which poses no problem 
in the Duhamel expansion.} 
the divergent functional determinant 
using a gauge-invariant cutoff as follows:
\begin{align}
\inv{2}\Tr\Big(\log\sD^2\Big)
\,\,\to\,\,-\inv{2}\Tr\int\limits_0^\infty\frac{d\a}{\a}
\,e^{-\a\sD^2} e^{-\inv{\a L^2}}
\,\,=:\,\, \Gamma_{L}[X] \, .
\label{TrLog-id}
\end{align}
Here $L$ is a cutoff of dimension `length', which 
essentially sets a lower limit $\a >\frac 1L$ for the $\a$ integral.
Although $X^a$ and $\sD$ have dimension `length' in the present setting, 
 this will amount to a UV cutoff
\be
\L := \LNC^2 L 
\label{eq:def-cutoff}
\ee
of dimension (length)$^{-1}$ in a NC background. This defines an effective 
(generalized) matrix model $\Gamma_{L}[X]$ in $X^a$,
which depends on the cutoff $L$ and satisfies the scaling relation
\be
\Gamma_{cL}[c X] = \Gamma_{L}[X] 
\,. \label{basic-scaling}
\ee
It corresponds to the induced action in NC gauge theory, resp. to the induced gravitational
action in emergent gravity.
Our goal is to compute this $\Gamma_{L}[X]$ explicitly, using a systematic approximation.

For completeness, we recall that 
the fermionic action \eq{Dirac-action} together with the bosonic action
\begin{align}
S_{YM} =-(2\pi)^2\Tr\left(\co{X^a}{X^b}\co{X_a}{X_b}\right)
\,
\label{YM-action}
\end{align}
defines the class of Yang-Mills matrix models. In particular, the
IKKT or IIB model \cite{Ishibashi:1996xs}  is 
obtained for $D=10$ imposing a Majorana-Weyl condition
on the fermions, and admits a maximal supersymmetry.

\section{Strategy of the heat kernel expansion}
\label{sec:strategy}

Before diving into the computations, we first explain the setup and the
essential ideas behind the complicated details.
We will take advantage of the two complementary points of view of the 
model, 1) as NC gauge theory and 2) as matrix model (for emergent gravity).
Gauge invariance is essential for both points of view. 
The second makes also the global $SO(D)$ symmetry manifest, 
which is hidden in the gauge theory point of view.
We will use the gauge theory point of view for the explicit (perturbative) computations,
and then recast the result in the matrix model language.

The form \eq{TrLog-id} of the induced action suggests to consider the associated 
heat kernel expansion 
\be
\Tr e^{-\a\sD^2} = \sum_{n}\, \a^{\frac{n-4}2}\, \Gamma^{(n)}[X]  \, .
\label{Seeley-deWitt}
\ee
The $\Gamma^{(n)}[X]$ are by construction gauge-invariant, 
and at least formally they are also invariant under
$SO(D)$. 
In the commutative case, 
the analog of \eq{Seeley-deWitt} involves a sum over positive $n$ only, 
and the Seeley-de Witt coefficients $\Gamma^{(n)}$ turn out to be
integrals of gauge-invariant densities over $\R^4$; cf.~\cite{Gilkey:1995mj,Vassilevich:2003xt}. 
This yields an effective action organized as
\be
\Gamma_{\L} \sim \L^4 \sum_{n\geq 0}\int d^4 x\sqrt{g}\, 
\cO\Big((\frac{p}{\L})^n\Big)\,,
\ee
where  $\cO\Big((\frac{p}{\L})^n\Big)$ stands
for some Lagrangian density involving $n$ powers of momentum
(due to background curvature, field strength, etc.), starting with 
the vacuum energy $\L^4\int d^4 x\sqrt{g}$.

In the NC case, things are much more subtle due to UV/IR mixing. 
In a perturbative expansion of NC gauge theory, one finds additional terms such as 
$e^{-p^2 \LNC^{-4}/\a}$ in \eq{Seeley-deWitt} originating from non-planar diagrams. 
These lead to a sum over arbitrary $n \in \Z$ in \eq{Seeley-deWitt}.
It is then not clear in general how to extract a meaningful limit 
or asymptotic series for $\L \to \infty$. 
This is the infamous UV/IR mixing problem (for a review see~\cite{Szabo:2001,Blaschke:2010kw} and references therein), 
which leads to various strange or pathological 
phenomena\footnote{For example, the heat kernel on the {\nc} torus was found 
to depend on number-theoretical properties of $\theta$ \cite{Gayral:2006vd}.}.

In the framework of emergent gravity \cite{Steinacker:2010rh}, this pathological UV/IR mixing 
is turned into a desirable feature, by restricting oneself to  
well-behaved models with sufficient supersymmetry, such as the $\cN=4$ Super-Yang Mills
model resp. the IKKT model~\cite{Ishibashi:1996xs}. That model is expected (and to some extent verified) 
to be UV finite just like its commutative counterpart \cite{Jack:2001cr},
and hence free of UV/IR mixing.
If the supersymmetry is (spontaneously or softly) broken at some scale $\L$, then a
mild form of UV/IR mixing arises from non-planar diagrams below this scale. This 
can be understood semi-classically in terms of induced gravity \cite{Steinacker:2008}, 
and this is what we want to 
compute in the present paper. We therefore consider the case of a {\em finite}
cutoff $\L$ resp. $L$, as expected in the full model due to 
SUSY breaking. We can then compute the induced action $\Gamma_L[X]$ \eq{TrLog-id}
by studying the heat kernel {\em not} in the limit $\L\to \infty$ but for finite 
$\L$, such that the external momenta $p$
(due to curvature, field strength, etc.) satisfy the condition
\be
 \epsilon_L(p):=\frac{p^2\L^2}{\LNC^4}  = p^2 L^2 \,\, \ll\,  1\, .
\label{IR-regime}
\ee
This characterizes the 
``semi-classical'' low-energy regime\footnote{Meaning that the phases in the loop integrals
are less than 1, i.e. UV/IR mixing is weak and within the semi-classical regime \cite{Steinacker:2008}.}. 
It certainly includes
the regime of interest for gravity, since both $\L$ and $\LNC$ are assumed to 
be physical high-energy scales, presumably related to the Planck scale.
Under this assumption, we can expand the UV/IR mixing terms such as 
\be
e^{-p^2\LNC^4/\a} = \sum_{m\geq 0} \frac 1{m!} (-p^2\LNC^4/\a)^m
\,\, \approx \,\, \sum_{m\geq 0} a_m \epsilon_L(p)^m
\ee 
replacing $\L^2 \geq \a^{-1}$ by  $\L^2$ 
as justified by the cutoff in \eq{TrLog-id}. 
Thus the heat-kernel expansion becomes an expansion of the form
\be
\Gamma_L \sim \L^4\sum_{n,l,k\geq 0} \int d^4 x\, \cO\left(\epsilon_L(p)^n (\frac{p^2}{\LNC^2})^l(\frac{p^2}{\L^2})^k \right)
\label{Seeley-deWitt-mod}
\ee
in powers
of three small parameters $\epsilon_L(p)$, $(p\th q) \sim \frac{p^2}{\LNC^2}$  and $\frac{p^2}{\L^2}$,
which makes sense provided
\eq{IR-regime} holds. This condition is very important, and our results no longer make sense in the limit
$\L \to \infty$ for fixed $\LNC$ as considered in  \cite{Gayral:2006vd},
because then $\epsilon_L(p)$ diverges.
A related observation on the appropriate definition of the heat kernel expansion
on fuzzy spaces was made in \cite{Sasakura:2004dq}.

A remark on the symmetries is in order.
By construction, the expansion of $\Gamma_L$ preserves gauge invariance 
at each order in $L^n$ resp. $\L^n$ as well as  the $SO(D)$ symmetry,  
at least formally. Although the background $\R^4_\th$ of course breaks both symmetries, 
these
symmetries are still realized in a non-linear way on the fluctuations, and should therefore
be respected in the quantum effective action. 
While this is quite evident for the gauge symmetry, the argument is not strictly valid for $SO(D)$,
since rotating the $X^\mu$ into the scalar fields leads to unbounded $\d\phi^i$
which might not be admissible. Nevertheless, it turns out that the effective 
actions obtained in this way do indeed respect this $SO(D)$ symmetry, as one 
would hope.

Finally, one may wonder about the relation with the commutative case.
The commutative limit should be approximated by imposing 
\be
\L \ll \LNC
\ee
(hence $\LNC L \ll 1$)
in addition to \eq{IR-regime}; in particular, we cannot send $\LNC \to \infty$ for fixed $L$.
Finally, non-Abelian gauge fields and scalar fields are naturally obtained within
the same model by replacing the ``single-brane'' background
by $n$ coinciding branes. However, then the following analysis needs to be refined, 
which should be worked out elsewhere.

\subsection{Perturbative expansion in NC gauge theory.}
\label{sec:perturb-expansion}

Our aim is to compute the leading (divergent) terms in the momentum expansion 
of $\Gamma_{L}[X]$ \eq{Seeley-deWitt-mod}.
For this purpose, we take the point of view of NC gauge theory on $\R^4_\th$,
and consider small fluctuations around $\R^4_\th$.
We accordingly split the matrices into background plus fluctuations,
\be
X^a = \begin{pmatrix}
 \bar X^\mu \\ 0
\end{pmatrix} \, + \, 
\begin{pmatrix}
 -\bar\th^{\mu\nu} A_\nu  \\   \phi^i
\end{pmatrix} \,,
\label{eq:general-X}
\ee
and treat $A_\mu = A_\mu(\bar X)$ and $\phi^i = \phi^i(\bar X)$ as gauge fields resp. scalar
fields on $\R^4_\th$.
In order to recover the standard dimensional assignment of quantum field theory, 
we define\footnote{Here and in most of the following, the scales
$\LNC$ and $\L$ will be defined by the Moyal-Weyl background.
Otherwise we will write $\L(x)$ etc.}
\be
\varphi^i:= \LNC^2 \phi^i 
\ee
which has dimension $\dim\varphi = L^{-1}$.
The corresponding field-theoretic normalization of the Dirac operator is 
given by $\LNC^{2} \sD$.
We can now perform a perturbative expansion of
\begin{align}
\inv{2}\Tr\left(\log\sD^2-\log\sD_0^2\right) &\to-\inv{2}\Tr\int\limits_0^\infty\frac{d\a}{\a}
\left(e^{-\a\sD^2}-e^{-\a\sD_0^2}\right) e^{-\inv{\a L^2}}   
 = \sum_{k>0} \cO(V^k) 
\nn\\
&= \L^4\sum_{n \geq 0} \int d^4 x\, \cO\left(\frac{(p,A,\varphi)^n}{(\L,\LNC)^n}\right)\,,
\label{engineering-expand}
\end{align}
where
\be
\sD_0\Psi:=\g_a\co{\bar X^a}{\Psi},  
\qquad \sD^2\Psi=  \sD_0^2\Psi + V .
\ee
Each term $\cO(V^k)$ contributes one or two fields, given by the
Duhamel expansion as explained in \appref{app:regularization}.  
However, each of these terms contains arbitrarily high powers of 
momenta due to UV/IR mixing, and their appropriate organization is not obvious a priori.
Moreover, gauge invariance is not respected by this expansion.
The most physical organization is according to their 
engineering dimension i.e. according to powers of fields and momenta,
as indicated in \eq{engineering-expand}.
This makes sense due to the IR condition \eq{IR-regime},
involving three small parameters 
$\e_L(p)$, $(p\th q) \sim \frac{p^2}{\LNC^2}$ and $\frac{p^2}{\L^2}$.
It leads to an effective action of the type
\be
\Gamma_L = \sum_{n,k,l \geq 0} \L^{2n-2k}\LNC^{-4n-2l} \int d^4 x\, \Gamma^{(n,k,l)}[A_\mu,\varphi^i] 
\ee
on $\R^4_\th$, where $\Gamma^{(n,k,l)}[A_\mu,\varphi^i]$
has engineering dimension $r = 4+2k+2l+2n$. 
Only finitely many terms in the Duhamel expansion contribute to each given $r$.
However, recovering gauge invariance requires a non-trivial 
re-organization of this expansion, since e.g. a commutator
of fields such as $[A_\mu,A_\nu]$ can be rewritten in momentum space
as an expansion in powers of $(p\theta q) \sim \frac{p^2}{\LNC^2}$.
This can be done iteratively, and one obtains
\be
\Gamma_L = \sum_{m > 0; n} \L^{n}\LNC^{-m} \int d^4 x\, \Gamma^{(n,m)}[A_\mu,\varphi^i] \,,
\label{eff-gauge-action-general}
\ee
where each $\Gamma^{(n,m)}[A_\mu,\varphi^i]$ is manifestly gauge-invariant,
interpreted as an effective higher-order NC gauge theory.
Note that the $\cO(V^k)$ contribution in 
the perturbative expansion \eq{engineering-expand} contains at least $k$ 
powers of fields $A$ resp. $\varphi$, so that any given term in \eq{eff-gauge-action-general}
is completely determined at some finite order of $k$.
These steps are carried out in \secref{sec:effgaugethy}, where 
the first terms in this expansion are found to be 
\eq{2-field-collect-b}, \eq{2-field-collect-curv1}, and \eq{2-field-collect-curv2}.

This result would be perfectly satisfactory from the point of view of NC gauge theory,
since each of the $\Gamma^{(n,m)}$ respects gauge invariance as well as the global
$SO(D-4)$ symmetry.
However, the effective action appears to violate $SO(4)$ invariance
due to terms such as $\theta^{\mu\nu} F_{\mu\nu}$, not to mention the original 
$SO(D)$ symmetry of the matrix model. 
These will be recovered in the next, non-trivial step, suggested by the gravity point of view.

\subsection{Re-assembling the effective matrix model.}

In the last step, we collect the gauge-invariant actions 
$\Gamma^{(n,m)}[A_\mu,\varphi^i]$ on $\R^4_\th$ and 
 translate the result into the form of a generalized matrix model,
\begin{align}
\Gamma_L = \sum_{n,m} \L^n \LNC^{-m} \int d^4 x\, \Gamma^{(n,m)}[A_\mu,\varphi^i] 
= \Tr \cL(X^a) \,,
\label{re-assemble}
\end{align}
which manifestly respects gauge invariance as well as the $SO(D)$ symmetry. 
This is a highly non-trivial (and at this point non-systematic)
step, which will be carried out explicitly in \secref{sec:effMM}. 
We obtain the following effective matrix model action \eq{full-action-MM}
\be
\Gamma_L = -\frac 14\Tr \frac{L^4}{\sqrt{\inv2 H^2-H^{ab}H_{ab} + \frac{c_1}{L^2}[X^c,H^{ab}][X_c, H_{ab}]
 + \frac{c_2}{L^2}H^{cd}[X_c,\Theta^{ab}] [X_d,\Theta_{ab}]  +\ldots }} 
\label{eq:Gamma-L}
\ee
which constitutes the main result of this paper.
This action has a manifest global $SO(D)$ symmetry,
which is hidden in the gauge theory point of view.
The dependence on the specific background $\R^4_\th$ has disappeared 
(except for its 4-dimensional nature),
which allows to translate the result into the geometric language of emergent gravity
where the $X^a$ are interpreted as embedding of $\cM^4 \subset \R^D$.
The higher-order terms then
correspond e.g. to curvature terms as shown in \cite{Blaschke:2010rg,Blaschke:2010qj}.
The ellipses in \eqnref{eq:Gamma-L} stand for terms which contain more derivatives resp. commutators,
corresponding to higher-order curvature terms, etc. 

The same type of effective action should be expected from the bosonic part of the 
matrix model, integrating out the $X^a$ at one loop or beyond. We plan to report on this 
elsewhere.
In this way, the quantization of NC gauge theory allows to obtain the quantization of 
(emergent) gravity.

In general this resulting expansion depends
on the background, in particular it is different e.g. for the non-Abelian 
modes. Thus it is not ``the complete'' effective matrix model, and there should be some universal
form which applies to all backgrounds. 
Note also that  the dependence on $\L$ characterizes some kind of 
``matrix'' renormalization group, which should be studied in detail.

\section{Details of the heat kernel expansion}
\label{sec:heatkernel-details}

\subsection{Square of the Dirac operator}

We choose coordinates such that 
\[g_{\mu\nu} = \del_\mu x^a \del_\nu x^b g_{ab} = \diag(1,1,1,1)\,,\] 
and using a suitable $SO(4)$ rotation we can assume that $\bar\theta^{\mu\nu}$
has the standard form \eq{theta-standard-general-E}.
According to \eqref{eq:def-Gges}, the effective metric on  $\R^4_\th$ is
\begin{align}
\bar G^{\mu\nu} &= \LNC^4\bar\theta^{\mu\mu'}\bar\theta^{\nu\nu'} g_{\mu'\nu'} \,.
\label{eq:def-Gges-moyal}
\end{align}
We consider the square of the Dirac operator on $\R^4_\theta$, and treat $A_\mu$ as well as $\phi^i$ as perturbations.
This gives
\begin{align}
\sD^2\Psi&=\g_a\g_b\co{X^a}{\co{X^b}{\Psi}} 
= (\d_{ab} -2i \Sigma_{ab})[X^a,[X^b,\Psi]] \nn\\
&= \d_{ab} [X^a,[X^b,\Psi]] +\Sigma_{ab}[\Theta^{ab},\Psi]\,,
\end{align}
where we define
\begin{align}
\Sigma_{ab} &= \frac i4 [\gamma_a,\gamma_b]  \,, &
\Theta^{ab} &= -i[X^a,X^b] \,.
\end{align}
Furthermore, we split the square of the Dirac operator according to
\begin{align}
\sD^2\Psi&=(H_0 + V) \Psi\,, \nn\\
 H_0\Psi &:= \d_{\mu\nu} [\bar X^\mu,[\bar X^\nu,\Psi]] 
  = -\LNC^{-4}\bG^{\mu\nu}\del_\mu\del_\nu \Psi \nn\\
 &= \bar\Box \Psi 
\,. 
\end{align} 
Using
\be
X^\mu = \bar X^\mu + \cA^\mu
= \bar X^\mu - \thb^{\m\n} A_\nu 
\,,
\ee
one has
\begin{align}
[X^a,[X_a,\Psi]] 
&= \bar\Box \Psi + \d_{\mu\nu}([\bar X^\mu,[\cA^\nu,\Psi]] +[\cA^\mu,[\bar X^\nu,\Psi]] + [\cA^\mu,[\cA^\nu,\Psi]]) 
 + \d_{ij}[\phi^i,[\phi^j,\Psi]] \nn\\ 
&= \bar\Box \Psi - \frac{i\bG^{\m\n}}{\LNC^4}
\left(2[A_{\mu},\del_{\nu}\Psi] + [\del_{\mu}A_{\nu},\Psi] + i[A_{\mu},[A_{\nu},\Psi]]\right)
+ \LNC^{-4}[\varphi^i,[\varphi_i,\Psi]] 
, \nn\\ 
2\Sigma_{ab}[X^a,[X^b,\Psi]] &=  \Sigma_{ab}\left([X^a,[X^b,\Psi]] - [X^b,[X^a,\Psi]]\right)
\nn\\
&=  i\Sigma_{ab}[\Theta^{ab},\Psi] 
\,.
\end{align}
Note that the scalar fields have been rescaled above according to $\varphi_i=\LNC^2\phi_i$. 
The components of $\Theta^{ab}$ are given by 
\begin{align}
[X^\mu,X^\nu] &= i\bar\theta^{\mu\nu} 
+ i\bar\theta^{\mu\rho}\partial_\rho \cA^{\nu}
- i\bar\theta^{\nu\rho}\partial_\rho \cA^{\mu} + [\cA^{\mu},\cA^{\nu}] \nn\\
&=  i\bar\theta^{\mu\nu} +  i\cF^{\mu\nu} 
\,, \nn\\
[X^\mu,\phi^i] &= i \bar\theta^{\mu\nu} D_\nu\phi^i
\,, \label{XX-gauge}
\end{align}
where
\begin{align}
\cF^{\m\n} &= -\thb^{\m\r} \thb^{\n\s}F_{\r\s} \nn\\
&= -\thb^{\m\r} \thb^{\n\s}\left(\pa_\r A_\s-\pa_\s A_\r+i\co{A_\r}{A_\s}\right) \,, \nn\\
D_\a \phi &= \del_\a\phi + i [A_\a,\phi] 
\,.
\end{align}
Then the above gives
\begin{align}
\LNC^4 V \Psi &= -i \bG^{\mu\nu}
\Big(2[A_{\m},\pa_{\n}\Psi] +\co{\pa_\m A_\n}{\Psi} +i[A_{\m},[A_{\n},\Psi]]\Big)+ \d_{ij}[\varphi^i,[\varphi^j,\Psi]] \nn\\
&\quad + \LNC^4\,\Sigma_{ab}[\Theta^{ab},\Psi]\, \nn\\
&= -i \bG^{\mu\nu}
\Big(2[A_{\m},\pa_{\n}\Psi] +\co{\pa_\m A_\n}{\Psi} +i[A_{\m},[A_{\n},\Psi]]\Big)+ \d_{ij}[\varphi^i,[\varphi^j,\Psi]] \nn\\
&\quad + \LNC^4\(\Sigma_{\mu\nu}[\cF^{\mu\nu},\Psi]
 + 2\Sigma_{\mu i}\bar\theta^{\mu\nu}[\del_\nu\phi^i +i [A_\nu,\phi^i],\Psi]
 -i \Sigma_{ij}[[\phi^i,\phi^j],\Psi]\)
\,.
\end{align}
Note that $V$ contains both linear and quadratic terms in the fluctuations $\varphi^i$ resp. $\cA^\mu$.
All these formulas are exact on $\R^4_\theta$.

\subsection{Setup for the trace computations}

First, we must specify the Hilbert space under consideration. 
The algebra $\cA$ of ({\nc}) functions\footnote{We do not consider functional-analytic details here.
One way to make this more rigorous would be to use the 
fuzzy torus, which is compact while having a very similar non-commutative structure.} 
on $\R^4_\th$ can be identified with the Heisenberg algebra, and 
therefore is represented as (infinite-dimensional) matrix algebra $\cA \subset {\rm End}(\cH)$
acting on a separable Hilbert space $\cH$. We suppress the spinor indices for simplicity
in this discussion. As usual $\cA$ (resp. a suitable subspace of $\cA$) can be equipped
with an inner product structure
\be
\langle \Psi_1,\Psi_2\rangle = \Tr_{\cH} \Psi_1^\dagger \Psi_2 = 
\LNC^4\inttx \sqrt{g}  \Psi_1^\dagger  \Psi_2
\label{inner-product}
\ee
for $\Psi_i \in \cA \cong \textrm{End}(\cH)$,  considered as a (pre-)Hilbert space
of wave-functions on $\R^4_\theta$.
Here the integral over $x$ is understood as integral over $\R^4_\th$.
We will always assume coordinates $x^\mu$
with $g_{\mu\nu} = \d_{\mu\nu}$ and often drop the $\sqrt{g}$.
Then $\sD \in \textrm{End}(\cA)$ resp. $\sD^2$ are Hermitian operators on $\cA$,
\be
\Tr_{\cH} \Psi_1^\dagger  \sD^2 \Psi_2 = \LNC^4\inttx \sqrt{g}  \Psi_1^\dagger \sD^2 \Psi_2
 = \LNC^4 \inttx\sqrt{g} (\sD^2 \Psi_1)^\dagger \Psi_2 
\,. \label{inner-product-D}
\ee
$\cA$ can be identified via  Fourier transform 
with square-integrable functions on $\R^4$, so that
$\cA \cong L^2(\R^4_\th) \cong L^2(\R^4)$.
This identification (the Weyl quantization map) is defined by mapping plane waves $e^{i p_\mu x^\mu}$ 
to the ``generalized eigenfunctions'' 
\be
|p\rangle = e^{i p_\mu \bar X^\mu} \in \bar \cA \supset \cA
\ee
of NC plane waves, which satisfy
$\bar P_\mu |p\rangle = i p_\mu |p\rangle$ for $\bar P_\mu = -i\bar\theta^{-1}_{\mu\nu}[\bar X^\nu,.]$;
note that $[\bar P^\mu,\bar P^\nu] = 0$. We can compute their inner product formally 
\begin{align}
\langle q|p\rangle &= \Tr(|p\rangle\langle q|) = 
 \Tr_\cH (e^{-i q_\mu \bar X^\mu} e^{i p_\mu \bar X^\mu})
= (2\pi\LNC^2)^2   \d^4(p-q) 
\,.
\end{align}
(Note that $\dim\cH = \Tr_\cH\one = \Tr(|0\rangle\langle 0|) = (2\pi\LNC^2)^2   \d^4(0)$
so that $\d^4(0) \sim {\rm Vol}$ corresponds to a divergent volume factor).
Hence the trace $\Tr$ over operators on $\cA$
(not to be confused with $\Tr_\cH$)
can be computed using the following relations
\begin{subequations}
\begin{align}
\Tr\, \cO &= \Intt{p} \langle p | \cO | p\rangle \,, \label{trace-p}\\
\one &= \Intt{p} |p\rangle\langle p |
\,. \label{one-p}
\end{align}
\end{subequations}
Notice the presence of the NC scale. 
Analogous formulas can be justified rigorously on (compact) fuzzy spaces such as $S^2_N$ or $T^2_\theta$.
In particular, a field $\Psi \in \cA$ 
is conveniently written in momentum basis as 
\begin{align}
|\Psi\rangle &= \Intt{p}\,|p\rangle\langle p|\Psi\rangle\,
 = \Intt{p}\, \psi(p)\,e^{i p_\mu \bar X^\mu}  \,,  \nn\\
\psi(p) &= \langle p | \Psi\rangle 
= \Tr_\cH(e^{-i p_\mu \bar X^\mu} \Psi) 
= \LNC^4\,\inttx\, \psi(x)\,e^{-i p_\mu x^\mu} \,, \nn\\
 \qquad \psi(x) &= \langle x|\Psi\rangle
= \Intt{p} e^{i p_\mu x^\mu} \psi(p)
 =  \Intt{p} \Tr_\cH(e^{-i p_\mu (x^\mu -\bar X^\mu)} \Psi)
\,. 
\end{align}
For example, $|p\rangle$ 
corresponds to $\psi(p') = (2\pi\LNC^2)^2\d^{(4)}(p'-p)$.
Similarly, we consider the Fourier representation of the external fields
\begin{align}
\phi^i &= \int\frac{d^4 p}{(2\pi\LNC^2)^2} \,\phi^i(p) \,e^{i p_\m \bar X^\m} \,,  \nn\\
A_\m &= \int\frac{d^4 p}{(2\pi\LNC^2)^2} \,A_\m(p) \,e^{i p_\m \bar X^\m} 
\,.
\end{align}
We can now  start with the formula \eq{general-eff-action} of \appref{app:regularization} which expresses the effective action 
as a trace of certain operators acting on the spinor field on $\R^4_{\th}$. 
In subsequent computations we will need: 
\be
H_0 |p\rangle = \LNC^{-4}\, p\cdot p\, |p\rangle\,,
\ee
where 
\be
p\cdot p:= \bar G^{\m\n}p_\m p_\n \,,
\ee 
and
\be
e^{i k \bar X} e^{i l \bar X} = e^{-\frac i2 k \bar\theta l}\, e^{i (k+l) \bar X} 
\ee
on $\R^4_{\theta}$. It therefore follows that 
\begin{align}
\co{e^{ik\bar X}}{e^{i l \bar X}}&= -2i\Sin{k\thb l}e^{i (k+l) \bar X} \,, \nn\\
e^{ip\bar X}\co{e^{ik\bar X}}{e^{i l \bar X}}&= -2ie^{\frac i2(k+l)\thb p}\Sin{k\thb l} e^{i (k+l+p) \bar X} \,, \nn\\
\co{\co{e^{ik\bar X}}{e^{i l \bar X}}}{e^{ip\bar X}}&= (-2i)^2\Sin{k\thb l}\Sin{(k+l)\thb p} e^{i (k+l+p) \bar X} \,, 
\end{align}
and similarly, for anti-commutators $-2i\sin()$ is replaced by $2\cos()$ in the above formulas. 
Furthermore, we have the following spinorial trace
\begin{align}
\tr (\Sigma_{ab}\Sigma_{cd}) &= \frac{\tr\one}{4}(g_{ac}g_{bd} - g_{ad}g_{bc}) 
\,.
\end{align}
Having collected all basic ingredients, we may proceed with the explicit computations presented 
in the next section.
The reader not interested in the details  may jump to section \ref{sec:effgaugethy},
where the effective gauge theory action is presented.

\subsection{Order by order computations}
\label{sec:order-by-order}

To start the computation, consider first the general matrix element
{\allowdisplaybreaks
\begin{align}
\langle\Psi_\b'|V|\Psi_\a\rangle &= \int\frac{d^4x}{(2\pi)^2}\, \Psi_\b'^\dagger 
\Big([\varphi^i,[\varphi_i,\Psi_\a]]  +\LNC^4\Sigma_{ab}[\Th^{ab},\Psi_\a]\nn\\*
&\quad -i \bG^{\mu\nu}
\left(2[A_{\mu},\del_{\nu}\Psi_\a] +\co{\pa_\m A_\n}{\Psi_\a} + i[A_{\mu},[A_{\nu},\Psi_\a]]\right)\Big) \nn\\
&=\LNC^{-4}\Intt{p}\Intt{q}\, \psi_\b'(p)^*
 \Bigg( 2i\bG^{\mu\nu} (p+q)_\n A_\m(p-q)\Sin{q\thb p} \nn\\*
&\quad\qquad -4\bG^{\mu\nu}\!\Intt{l} A_\m(p-q-l)A_\n(l)\Sin{q\bar\th l}\Sin{(l+q)\thb p} \nn\\*
&\quad\qquad -4\d_{ij}\Intt{l} \varphi^i(p-q-l)\varphi^j(l)\Sin{q\thb l}\Sin{(q+l)\thb p} \nn\\*
&\quad\qquad +2i\LNC^4\Sigma_{ab}\Th^{ab}(p-q)\Sin{q\thb p} \Bigg)\psi_\a(q) 
\,, \label{eq:basic-V-simple}
\end{align}
where
}
\be
\Theta^{ab}(x) = \Intt{k}\, \Theta^{ab}(k)\,e^{i k_\m x^\m}
\,.
\ee 
We note that this interaction $V$ vanishes when any of the external 
fields $A$ or $\varphi$ has zero momentum. This is clear because 
they arise only from commutators, which vanish in the commutative case.
Hence all the non-trivial results
are due to {\nc} effects resp. UV/IR mixing.

\paragraph{First order. }

From \eqref{eq:basic-V-simple} it follows that
\begin{align}
&\Tr\left(Ve^{-\a H_0}\right)=\Intt{p} \tr V_{p,p} e^{-\bar\a p\cdot p} \nn\\
&=\frac{4\tr\one}{\LNC^4}\Intt{l} \left(\bG^{\m\n} A_\m(-l)A_\n(l)
+\varphi^i(-l)\varphi_i(l)\right) 
\!\Intt{p} \Sinn{l\thb p}{2} e^{-\bar\a p\cdot p} \nn\\
&=\frac{\tr\one}{2}\frac{1}{\L_{NC}^8}\Intt{l} \sqrt{g}\left(\bG^{\m\n} A_\m(-l)A_\n(l)
+\varphi^i(-l)\varphi_i(l)\right)
\inv{\bar\a^2}\left(1-e^{-\frac{\tilde l \cdot \tilde l}{4\bar\a}}\right)
\,, \label{first-order}
\end{align}
where we define
\eqn{
\bar\a &= \LNC^{-4} \a \,, && \L^2 = \LNC^{4} L^2\,,  \nn\\*
\tilde p_\mu &:= \bG_{\mu\nu} \tilde p^\nu := \bG_{\mu\nu} \bar\theta^{\nu\rho} p_\rho \,,
}
and note that
\begin{align}
l \thb p &= l_\mu \thb^{\m\n}p_\nu = - \tilde l \cdot p \,, \nn\\
\tilde l \cdot \tilde l &= \bar G_{\mu\nu} \thb^{\nu\rho} l_\rho \thb^{\mu\rho'} l_{\rho'}
=  \LNC^{-4} l_\rho l_{\rho'} g^{\rho\rho'} =: \LNC^{-4}\,l^2 
\,.
\label{lcdotl-identity}
\end{align}
This gives the following contribution to the effective action 
\begin{align}
\Gamma^{(1)}&= \inv2\int\limits_0^{\infty}\!d\a\Tr\left(Ve^{-\a H_0}\right)e^{\frac{-1}{\a L^2}}
 = \inv2\LNC^{4}\int\limits_0^{\infty}\!d\bar\a\Tr\left(Ve^{-\bar\a p\cdot p}\right)
e^{\frac{-1}{\bar\a\L^2}}\nn\\
 &= \frac{\tr\one}{4}\LNC^{-4}\Intt{l} \int_0^\infty\! d\bar\a\, \inv{\bar\a^2}\left(1-e^{-\frac{\tilde l \cdot \tilde l}{4\bar\a}}\right) 
e^{- \inv{\bar\a\L^2}} \,\left(\bG^{\m\n} A_\m(-l)A_\n(l)
+\varphi^i(-l)\varphi_i(l)\right) \nn\\
 &= \frac{\tr\one}{4}\LNC^{-4}\Intt{l}\sqrt{g}  \left(\L^2 -  \Leff^2(l)\right)
\left(\bG^{\m\n} A_\m(-l)A_\n(l)+\varphi^i(-l)\varphi_i(l)\right)
\,, \label{Gamma-eff-1}
\end{align}
which involves the famous ``effective cutoff'' 
\be
\Leff^2(l) = \frac{\L^2}{1 + \inv4\frac{\L^2}{\LNC^4}\, l^2} 
\,.
\ee
The point is now that gravity is an infrared phenomenon, so that we are interested in the regime where
\be
\epsilon_L(l) := \frac{\L^2}{\L_{NC}^4}\, l^2  =  L^2 l^2 \ll 1 
\,, 
\ee
given some finite cutoff $\L$.
Hence we can expand\footnote{This is the essential difference to the previous work 
\cite{Gayral:2006vd,Vassilevich:2005vk}.} 
\be
\L^2 -  \Leff^2(l) = \frac 14 \epsilon_L(l) + \cO(\epsilon_L(l)^2)
\,. \label{Leff-expand}
\ee
Keeping only the 
leading non-trivial term $\L^2 -  \Leff^2(l) \to \frac{\L^4}{4\L_{NC}^4}\, l^2 $, 
the above action reduces to 
\begin{align}
\Gamma^{(1)} &= \frac{\tr\one}{16}\,\frac{\L^4}{\L_{NC}^8}\Intt{l} \sqrt{g} l^2
\left(\bG^{\m\n} A_\m(-l)A_\n(l)+\varphi^i(-l)\varphi_i(l)\right) \quad + \cO(l^4)  \nn\\
&= \frac{\tr\one}{16}\,\frac{\L^4}{\LNC^{4}}\int \frac{d^4x}{(2\pi)^2}  \sqrt{g}
 \left(\bG^{\m\n} g^{\a\b} \del_\a A_\m(x)\del_\b A_\n(x)
+g^{\a\b} \del_\a \varphi^i(x)\del_\b \varphi_i(x)\right) + \textrm{h.o.}
\label{first-order-result}
\end{align}
Notice that this is already a ``non-planar'' (UV/IR mixing) contribution, as it involves $\epsilon_L(l)$.
It will be identified as part of the  ``cosmological constant'' term $\L^4\Int{x}  \sqrt{g}$, 
which depends on the embedding metric $g_{\mu\nu}$, cf. \cite{Steinacker:2008}. 
This is consistent with the semi-classical result in
\cite{Klammer:2009dj}, but the present derivation is exact to all orders in $\theta$.
Higher-order terms in the expansion \eq{Leff-expand} will contribute in particular 
to the curvature action.

\paragraph{Second order. }
The second order in the heat kernel expansion yields
{\allowdisplaybreaks
\begin{align}
&\quad \Tr\left(Ve^{- t H_0}Ve^{-( \a- t)H_0}\right)
=\Intt{p}\Intt{q} \,\tr \left(V_{p,q}e^{-\bar tq\cdot q}V_{q,p} e^{-(\bar \a-\bar t) p\cdot p}\right) \nn\\
&= \frac{16\tr\one}{\LNC^{8}}\int\!\!\frac{d^4p\,d^4q}{(2\pi\LNC^2)^4}e^{-\bar tq\cdot q-(\bar \a-\bar t) p\cdot p}\!\Bigg(\!\bG^{\m\n}\bG^{\r\s}\frac{(p\!+\!q)_\n (p\!+\!q)_\s}{4} A_\m(p\!-\!q)A_\r(q\!-\!p)\Sinn{q\thb p}{2} \nn\\*
&\quad +\frac 14\LNC^8\left(g_{ac}g_{bd}-g_{ad}g_{bc}\right)\Th^{ab}(p-q)\Th^{cd}(q-p)\Sinn{q\thb p}{2} \nn\\
&\quad +\!\Intt{l}\Bigg[i\bG^{\m\n}\bG^{\r\s}\frac{(p\!+\!q)_\n}{2} A_\m(p\!-\!q)A_\r(q\!-\!p\!-\!l)A_\s(l)\SiN{q\thb p}{2}{\!}\!\SiN{l\thb p}{2}{\!}\!\SiN{(l\!+\!p)\thb q}{2}{\!} \nn\\*
&\quad\quad +i\bG^{\m\n}\bG^{\r\s}\frac{(p\!+\!q)_\s}{2} A_\r(q\!-\!p) A_\m(p\!-\!q\!-\!l)A_\n(l)\Sin{q\thb p}\!\Sin{q\thb l}\!\Sin{(l\!+\!q)\thb p} \nn\\*
&\quad\quad +i\d_{ij}\bG^{\m\n}\frac{(p\!+\!q)_\n}{2} A_\m(p\!-\!q)\varphi^i(q\!-\!p\!-\!l)\varphi^j(l)\Sin{q\thb p}\!\Sin{l\thb p}\!\Sin{(l\!+\!p)\thb q} \nn\\*
&\quad\quad +i\d_{ij}\bG^{\m\n}\frac{(p\!+\!q)_\n}{2} A_\m(q\!-\!p)\varphi^i(p\!-\!q\!-\!l)\varphi^j(l)\Sin{q\thb p}\!\Sin{q\thb l}\!\Sin{(l\!+\!q)\thb p} \Bigg] \nn\\
&\quad +\Intt{l}\Intt{k}\Bigg[\Sin{q\thb l}\!\Sin{p\thb k}\!\Sin{(l+q)\thb p}\!\Sin{(k+p)\thb q}\times \nn\\*
&\quad\quad\times\!\left(\bG^{\m\n}\bG^{\r\s}A_\m(p\!-\!q\!-\!l)A_\r(q\!-\!p\!-\!k)A_\n(l)A_\s(k)
+\varphi^i(p\!-\!q\!-\!l)\varphi_i(l)\varphi^j(q\!-\!p\!-\!k)\varphi_j(k)\right) \nn\\*
&\quad\quad +\d_{ij}\bG^{\m\n}\Big(A_\m(p-q-l)A_\n(l)\varphi^i(q-p-k)\varphi^j(k)\Sin{q\thb l}\!\Sin{p\thb k}\times\nn\\*
&\hspace{3cm}\times\Sin{(l+q)\thb p}\!\Sin{(p+k)\thb q}
     +\left\{\!\begin{array}{c}
       q\leftrightarrow p \\
       l\leftrightarrow k
       \end{array}\!\right\}\Big) \Bigg]\Bigg) 
\,, \label{eq:second-order-terms-compact}
\end{align}
involving two, three and four field contributions. 
We will ultimately do the computations of this section up to third order, and hence neglect the four field 
contributions --- a sample expression is nonetheless given in \appref{app:heatkernel-four-field}. 
In order to compute at least one of the momentum integrals, we need to make clever variable substitutions where the fields are 
independent of one of the new integration variables. For example, for the terms which only involve integrals over $p$ and $q$, 
the substitution $P=p+q$ and $Q=p-q$ is favourable and leads to integrals of the type
\begin{subequations}\label{eq:p-gauss}
\begin{align}
& \Intt{P} \sin^2\!\left(\!\frac{Q\thb P}{4}\!\right)\!e^{-\frac{\bt}{4}(Q-P)\cdot(Q-P)-\inv{4}(\ba-\bt)(P+Q)\cdot(P+Q)} \nn\\*
&= \frac{2\sqrt{\bG}}{\LNC^4\ba ^2} \left(e^{\frac{\dott{\Qt}{\Qt}}{4 \ba }}-1\right)
   e^{-\frac{4 \bt (\ba -\bt)Q\cdot Q+\tilde{Q}\cdot \tilde{Q}}{4 \ba
   }}
\,,\label{eq:p-gauss-a}\\
& \Intt{P}P_\n P_\s \sin^2\!\left(\!\frac{Q\thb P}{4}\!\right)\!e^{-\frac{\bt}{4}(Q-P)\cdot(Q-P)-\inv{4}(\ba-\bt)(P+Q)\cdot(P+Q)} \nn\\*
&=\frac{2\sqrt{\bG}}{\LNC^4\ba ^4} e^{-\frac{4 \bt (\ba -\bt)\dott{Q}{Q}+\dott{\Qt}{\Qt}}{4 \ba }} \left(\tilde{Q}_{\nu } \tilde{Q}_{\sigma
   }+ \left(Q_{\nu } Q_{\sigma } (\ba -2 \bt)^2+ 2\a\bG_{\n\s}\right) \left(e^{\frac{\dott{\Qt}{\Qt}}{4 \ba }}-1\right)\right)
\,,\label{eq:p-gauss-b}\\
& \Intt{P}P_\n \sin\!\left(\!\tfrac{Q\thb P}{4}\!\right)\!\sin\!\left(\!\tfrac{l\thb (Q+P}{4}\!\right)\!\sin\!\left(\!\tfrac{(2l+P+Q)\thb (Q-P)}{8}\!\right)\!e^{-\tfrac{\bt}{4}\dott{(Q-P)}{(Q-P)}-\inv{4}(\ba-\bt)\dott{(P+Q)}{(P+Q)}} \nn\\*
&=\frac{\sqrt{\bG}}{2\LNC^4\ba ^3} \exp \left(-\frac{2
   \tilde{l}\cdot\tilde{Q}+\tilde{l}^2+\tilde{Q}\cdot\tilde{Q}+2 i (2 \bt+\ba )
   (l \theta  Q)+4 Q\cdot Q \bt (\ba -\bt)}{4 \ba }\right)\times\nn\\*
&\quad\times   \Bigg(\tilde{l}_{\nu } \left(e^{\frac{2
   \tilde{l}\cdot\tilde{Q}+\tilde{Q}\cdot\tilde{Q}}{4 \ba }}-1\right)
   \left(e^{\frac{2 i \bt (l \theta  Q)}{\ba }}+e^{i (l \theta 
   Q)}\right)\nn\\*
&\quad\quad +e^{\frac{2 \tilde{l}\cdot\tilde{Q}+\tilde{l}\cdot\tilde{l}+4 i \bt (l
   \theta  Q)}{4 \ba }} \left(\tilde{Q}_{\nu } \left(1+e^{i (l
   \theta  Q)}\right)-i Q_{\nu } (\ba -2 \bt)
   \left(e^{\frac{\tilde{Q}\cdot\tilde{Q}}{4 \ba }}-1\right) \left(-1+e^{i
   (l \theta  Q)}\right)\right)\nn\\*
&\quad\quad +i \left(Q_{\nu } (\ba -2 \bt)
   \left(e^{\frac{2 \tilde{l}\cdot\tilde{Q}+\tilde{Q}\cdot\tilde{Q}}{4 \ba
   }}-1\right) \left(e^{i (l \theta  Q)}-e^{\frac{2 i \bt (l
   \theta  Q)}{\ba }}\right)+i \tilde{Q}_{\nu }
   \left(e^{\frac{2 i \bt (l \theta  Q)}{\ba }}+e^{i (l \theta 
   Q)}\right)\right)\Bigg)
\,,\label{eq:p-gauss-c} 
\end{align}
\end{subequations}
for the two and three field contributions. 
In order to make sense of the UV/IR mixing terms, we consider $\L$ as a {\em finite} cutoff, 
and assume that $\epsilon_L(p)= p^2\L^2/\L_{NC}^4\ll 1$. 
We can then expand the ``UV/IR mixing terms'' e.g. as 
$e^{\tilde p \cdot \tilde p/\a} = 1+ \sum \inv{\a^n} \frac{p^{2n}}{\LNC^{4n}}$,
which amounts to
an expansion in the NC parameter $\theta$ {\em after} performing the loop integral.
}

Adopting the expansion in $\epsilon_L(p)$ as justified above, we only need
{\allowdisplaybreaks
\begin{subequations}\label{eq:p-para-expand}
\begin{align}
\int\limits_0^{\infty}\!\!d\ba\!\int\limits_0^{\ba}\!\!d\bt e^{\frac{-1}{\L^2}}\eqref{eq:p-gauss-a}
&= 
\frac{\sqrt{\bG}}{2\LNC^4} \Bigg( \L^2 \tilde Q\cdot\tilde Q+ \frac{Q\cdot Q \tilde{Q}\cdot \tilde{Q}}{6} \left( \ln \!\left(\!\frac{Q\cdot Q}{\L^2}\!\right)+2 \gamma_E -\frac{8}{3}\right)\!\!\Bigg) +\cO\left(\tfrac{Q^6}{\L^6}\right)
\!,\label{eq:p-para-a-expand}\\
\int\limits_0^{\infty}\!d\ba\!\int\limits_0^{\ba}\!d\bt e^{\frac{-1}{\L^2}}\eqref{eq:p-gauss-b}
&= \frac{2\sqrt{\bG}}{\LNC^4}\Bigg(\frac 1{2} \bar G_{\nu\s} \L^4\tilde Q\cdot\tilde Q
+\L^4\tilde Q_\nu \tilde Q_\s -\frac{\L ^6}{2}
    \tilde{Q}\cdot \tilde{Q} \left(\tilde{Q}_{\nu } \tilde{Q}_{\sigma } + \inv{4}\tilde{Q}\cdot \tilde{Q}\bG_{\n\s} \right) \nn\\*
&\qquad - \frac {\L^2}{12}(\bar G_{\nu\s} Q\cdot Q - Q_\n Q_\s)\tilde Q\cdot\tilde Q 
 - \frac{\L^2}{6} Q\cdot Q \Qt_\nu \Qt_\s   \Bigg) 
 + \cO(Q^6/\L^6) 
\,, \label{eq:p-para-b-expand} \\
\int\limits_0^{\infty}\!d\ba\!\int\limits_0^{\ba}\!d\bt e^{\frac{-1}{\L^2}}\eqref{eq:p-gauss-c}
&= \frac{\sqrt{\bG}}{\LNC^4}\Bigg(\frac{\L^6 Q_{\nu }}{8 Q\cdot Q} (2 \lt \cdot
   \Qt +\Qt \cdot \Qt ) \left(2
   \lt \cdot (\lt +\Qt )+\Qt \cdot
   \Qt \right) \sin\! \left(\tfrac{(l \th  Q)}{2} \right) \nn\\*
&\qquad +\frac{\L^4}{4 Q\cdot Q} \bigg(Q\cdot Q \left(\Qt \cdot \Qt  \left(2 \Qt _{\nu } 
   \sin^2\!\left(\tfrac{(l \th  Q)}{4} \right)+\lt _{\nu }\right)+\lt \cdot \lt  \Qt _{\nu }\right) \nn\\*
&\quad\qquad +Q_{\nu } (2\lt+\Qt) \cdot \Qt  \left(l \th  Q-2 \sin \!\left(\tfrac{(l \th Q)}{2} \right)\right) 
    +2 \lt \cdot \Qt Q\cdot Q (\lt _{\nu }+\Qt _{\nu })\bigg) \nn\\*
&\qquad +\L^2 \Qt _{\nu } \left(\cos \!\left(\tfrac{(l \th Q)}{2} \right)-1\right) \nn\\*
&\qquad +\inv{9} Q\cdot Q \Qt _{\nu } \sin^2\!\left(\tfrac{(l \th  Q)}{4} \right) 
  \!\! \left(\!3 \ln \!\left(\!\tfrac{\L^2}{Q\cdot Q} \!\right)-6 \gamma_E +8\right)\!\!\!\Bigg)
+ \cO\big(\tfrac{Q^6}{\L^6}\big) 
. \label{eq:p-para-c-expand} 
\end{align}
\end{subequations}
Note that every power of $Q$ is suppressed by either $\inv{\L}$ or $\inv{\LNC}$, along possibly
with factors $\frac{\L}{\LNC}$ which we assume to be finite.
}

Collecting all 2-field contributions in \eqref{eq:second-order-terms-compact} we hence find with \eqref{eq:p-gauss-a}, 
\eqref{eq:p-gauss-b}, \eqref{eq:p-para-a-expand} and \eqref{eq:p-para-b-expand} and 
$d^4p\, d^4q = \frac 1{16} d^4P d^4 Q$:
\begin{align}
&\int\limits_0^\infty\! d\a \int\limits_0^\a\! dt\,\,
\Tr\left(V e^{-tH_0}Ve^{-(\a-t)H_0}\right)e^{- \frac 1{\a L^2}}\Big|_{\textrm{2-fields}} \nn\\*
&=\frac{\tr\one\sqrt{\bG}}{2\LNC^4}
\Intt{Q}\Bigg\{ \bG^{\m\n}\bG^{\r\s}A_\m(Q)A_\r(-Q)\Bigg(\inv{2} 
\bG_{\nu\s} \L^4\tilde Q\cdot\tilde Q +\L^4\tilde Q_\nu \tilde Q_\s  \nn\\*
&\quad\qquad  -\frac{\Lambda ^6}{2} \tilde{Q}\cdot \tilde{Q} \Big(\tilde{Q}_{\nu } \tilde{Q}_{\sigma } 
   + \inv{4}\tilde{Q}\cdot \tilde{Q}\bG_{\n\s} \Big) 
 - \frac {\L^2}{12}(\bar G_{\nu\s} Q\cdot Q - Q_\n Q_\s)\tilde Q\cdot\tilde Q 
 - \frac{\L^2}{6} Q\cdot Q \Qt_\nu \Qt_\s   \Bigg) \nn\\*
&\quad +\frac{\LNC^8}2 \Th^{ab}(Q)\Th_{ab}(-Q)\Bigg(\! \L^2 \tilde Q\cdot\tilde Q
+ \tfrac{Q\cdot Q}{6} \tilde{Q}\cdot \tilde{Q} \left(\! \ln \!\left(\!\tfrac{Q\cdot Q}{\L^2}\!\right)+2 \gamma_E -\tfrac{8}{3}\right)\!\!\Bigg)
 + \cO\!\left(\tfrac{Q^6}{\L^6}\right)\!\! \Bigg\}
\!. \label{2-field-collect}
\end{align}
Similarly, for the 3-field contributions in \eqref{eq:second-order-terms-compact} using \eqref{eq:p-gauss-c} and \eqref{eq:p-para-c-expand} we have:
\begin{align}
&\int\limits_0^\infty\! d\a \int\limits_0^\a\! dt\,\,
\Tr\left(V e^{-tH_0}Ve^{-(\a-t)H_0}\right)e^{- \frac 1{\a L^2}}\Big|_{\textrm{3-fields}} \nn\\*
&\approx \frac{\tr\one\sqrt{\bG}}{2\LNC^4}\int\!\!\frac{d^4Qd^4l}{(2\pi\LNC^2)^4}i\bG^{\m\n}\Bigg\{
\!\!\left(\bG^{\r\s}\!A_\m(Q)A_\r(-Q-l)A_\s(l)
 +A_\m(Q)\vph^i(-Q-l)\vph_i(l) \right) \! \times  \nn\\*
&\quad \times  \Bigg(\frac{\L^6 Q_{\nu }}{8 Q\cdot Q} (2 \lt \cdot
   \Qt +\Qt \cdot \Qt ) \left(2
   \lt \cdot (\lt +\Qt )+\Qt \cdot
   \Qt \right) \sin\! \left(\tfrac{(l \th  Q)}{2} \right) \nn\\*
&\qquad +\frac{\L^4}{4 Q\cdot Q} \bigg(Q\cdot Q \left(\Qt \cdot \Qt  \left(2 \Qt _{\nu } 
   \sin^2\!\left(\tfrac{(l \th  Q)}{4} \right)+\lt _{\nu }\right)+\lt \cdot \lt  \Qt _{\nu }\right) 
   +2 \lt \cdot \Qt Q\cdot Q (\lt _{\nu }+\Qt _{\nu }) \nn\\*
&\quad\qquad +Q_{\nu } (2\lt+\Qt) \cdot \Qt  \left(l \th  Q-2 \sin \!\left(\tfrac{(l \th Q)}{2} \right)\!\right) 
  \!\!\bigg) +\L^2 \Qt _{\nu } \left(\cos \!\left(\tfrac{(l \th Q)}{2} \right)-1\right) \nn\\*
&\qquad +\inv{9} Q\cdot Q \Qt _{\nu } \sin^2\!\left(\tfrac{(l \th  Q)}{4} \right) 
   \left(3 \ln \!\left(\!\tfrac{\L^2}{Q\cdot Q} \!\right)-6 \gamma_E +8\right)\!\!\Bigg) 
 + \Big(Q\to-Q\Big) + \cO\big(\tfrac{Q^6}{\L^6}\big) \Bigg\} 
. \label{3-field-collect}
\end{align}

\paragraph{Third order.} 
The third order in the heat kernel expansion yields for the three field contributions:
{\allowdisplaybreaks
\begin{align}
&\quad \Tr\left(Ve^{-r H_0}Ve^{-(t-r)H_0}Ve^{-(\a-t)H_0}\right) \Big|_{\textrm{3fields}} \nn\\*
&=\Intt{p}\Intt{q}\Intt{l}\,\tr\left(V_{p,q}e^{-\bar rq\cdot q}V_{q,l}e^{-(\bar t-\bar r)l\cdot l}V_{l,p} e^{-(\bar \a-\bar t) p\cdot p}\right)\Big|_{3\textrm{fields}} \nn\\
&= -\frac{16\tr\one}{\LNC^{12}}\Intt{p}\Intt{q}\Intt{l} \,e^{-\bar rq\cdot q -(\bar t-\bar r)l\cdot l -(\bar \a-\bar t) p\cdot p}\times \nn\\*
&\; \times \!\!\Bigg(\!\tfrac{i}{2}\bG^{\m\n}\bG^{\r\s}\bG^{\t\e}(p\!+\!q)_\n(q\!+\!l)_\s(l\!+\!p)_\e A_\m(p\!-\!q)A_\r(q\!-\!l)A_\t(l\!-\!p) \Sin{q\th p} \!\Sin{l\th q} \!\Sin{p\th l}\nn\\*
&\quad\; +\frac{i}{4}\LNC^8\bG^{\m\n}\Sin{q\th p} \!\Sin{l\th q} \!\Sin{p\th l} \Big((p+q)_\n A_\m(p-q)\Th^{ab}(q-l)\Th_{ab}(l-p) \nn\\*
&\quad\qquad +\Th^{ab}(p-q)\Th_{ab}(l-p)(q+l)_\n A_\m(q-l)+\Th^{ab}(p-q)\Th_{ab}(q-l)(l+p)_\n A_\m(l-p)\Big) 
\!\Bigg).
\end{align}
In order to be able to perform these integrals explicitly, we need to find an appropriate 3-dimensional coordinate transformation so that all fields become independent of one of these variables. One possibility would be the following:
}
\begin{align}
\left(\begin{array}{c}
P \\
Q \\
L 
\end{array}\right)=\left(\begin{array}{ccc}
1 & 0 & 1\\
1 & -1 & 0\\
0 & 1 & -1
\end{array}\right)\left(\begin{array}{c}
p \\
q \\
l 
\end{array}\right)
. 
\end{align}
Then e.g. $A_\m(p-q)A_\r(q-l)A_\t(l-p) = A_\m(Q)A_\r(L)A_\t(-Q-L)$ is independent of $P$, 
allowing explicit integration over that variable.
We hence find
\begin{align}
&\quad \Tr\left(Ve^{-rH_0}Ve^{-(t-r)H_0}Ve^{-(\a-t)H_0}\right) \Big|_{\textrm{3fields}} \nn\\*
&= \frac{-\tr\one}{\LNC^{12}}\int\!\frac{d^4Pd^4Qd^4L}{(2\pi\LNC^2)^{6}} \,e^{-\frac{\bar r}{4}(P-Q+L)\cdot(P-Q+L) -\frac{\bar t-\bar r}{4}(P-Q-L)\cdot(P-Q-L) -\frac{\bar \a-\bar t}{4} (P+Q+L)\cdot(P+Q+L)}\times \nn\\*
&\; \times \!\Bigg(\!\frac{i}{2}\bG^{\m\n}\bG^{\r\s}\bG^{\t\e}(P+L)_\n(P-Q)_\s P_\e A_\m(Q)A_\r(L)A_\t(-Q-L) \SiN{(P+L)\th Q}{4}{\!}\times \nn\\*
&\quad\quad\times\SiN{(P-Q)\th L}{4}{\!} \!\SiN{(Q+L)\th P}{4}{\!}\nn\\
&\quad\; +\frac{i}{4}\LNC^8\bG^{\m\n}\SiN{(P+L)\th Q}{4}{\!} \SiN{(P-Q)\th L}{4}{\!} \!\SiN{(Q+L)\th P}{4}{\!} 
\Big(\!(P\!+\!L)_\n A_\m(Q)\Th^{ab}(L)\Th_{ab}(-Q\!-\!L) \nn\\*
&\quad\quad +\Th^{ab}(Q)(P-Q)_\n A_\m(L)\Th_{ab}(-Q-L)+\Th^{ab}(Q)\Th_{ab}(L) P_\n A_\m(-Q-L)\Big) 
\!\Bigg).
\end{align}
Integration over $P$ yields a rather lengthy expression, and in order to continue with the parameter integrals (once more regularized by a cutoff $\L$), we consider the following approximations \emph{after} the $P$ integration:
\begin{enumerate}
\item We replace all phases such as $e^{i Q\th L}$ by $1$, which amounts to dropping 
higher-order terms in the $\theta$-expanded action resp. commutators in the NC action.
\item Since the divergent contributions are due to the parameter region $t<r<\a\approx0$, 
we keep only the leading terms in the exponent of type $\propto\inv\a$ 
(i.e. dropping contributions such as $rt/\a$,
which would correspond to higher-order terms in the action).
\item Finally, an expansion of type $e^{\frac{\tilde{Q}\cdot\tilde{Q}}{\a}}\approx 
\left(1+\frac{\tilde{Q}\cdot\tilde{Q}}{\a}+\ldots\right)$ is made, as explained in \secref{sec:strategy}.
\end{enumerate}
Therefore, we get the following leading order contributions
\begin{subequations}
\begin{align}
&\quad\int\limits_0^{\infty}d\bar\a\int\limits_0^{\bar \a}d\bar t\int\limits_0^{\bar t}d\bar r\int d^4P 
e^{-\frac{\bar r}{4}(P-Q+L)\cdot(P-Q+L) -\frac{\bar t-\bar r}{4}(P-Q-L)\cdot(P-Q-L) -\frac{\bar \a-\bar t}{4} (P+Q+L)\cdot(P+Q+L)}
 \times \nn\\*
&\quad\qquad \times e^{-\frac{1}{\bar\a\L^2}} \frac{\bar r}{\bar\a} (P+L)_\n(P-Q)_\s P_\e \SiN{(P+L)\th Q}{4}{\!}\SiN{(P-Q)\th L}{4}{\!} \!\SiN{(Q+L)\th P}{4}{\!} \nn\\*
&\approx -\frac{\pi ^2 \L^4}{3} \Bigg(\tilde{Q}\cdot \tilde{Q}
   \left(\tilde{L}_{\e } \bar G_{\nu  \s}
   +\tilde{L}_{\nu } \bar G_{\e  \s }+\tilde{L}_{\s} \bar G_{\e  \nu }\right)
  +\tilde{L}\cdot \tilde{L} \left(\tilde{Q}_{\e } \bar G_{\nu  \s }
   +\tilde{Q}_{\nu } \bar G_{\e  \s} 
  + \bar G_{\e  \nu } \tilde{Q}_{\s } \right)  \nn\\*
&\quad +2 \tilde{L}\cdot \tilde{Q}
   \left(\bar G_{\nu  \s } (\tilde{L}
   +\tilde{Q})_{\e }+\bar G_{\e  \s }
   (\tilde{L}+\tilde{Q})_{\nu }+\bar G_{\e
    \nu } (\tilde{L}+\tilde{Q})_{\s}\right) \nn\\*
&\quad + 2\tilde{Q}_{\s } \tilde{L}_{\nu }
   (\tilde{L}+\tilde{Q})_{\e }
  +2\tilde{L}_{\e } \tilde{Q}_{\nu }(\tilde{Q}_{\s } +\tilde{L}_{\s})
  +2 \tilde{L}_{\s } \tilde{Q}_{\e }(\tilde{L}_{\nu } +\tilde{Q}_{\nu })
\Bigg)
, \label{V3-approx-1}
\end{align}
\begin{align}
&\quad\int\limits_0^{\infty}d\bar\a\int\limits_0^{\bar \a}d\bar t\int\limits_0^{\bar t}d\bar r\int d^4P 
e^{-\frac{\bar r}{4}(P-Q+L)\cdot(P-Q+L) -\frac{\bar t-\bar r}{4}(P-Q-L)\cdot(P-Q-L) -\frac{\bar \a-\bar t}{4} (P+Q+L)\cdot(P+Q+L)}
 \times \nn\\*
&\quad\qquad \times e^{-\inv{\bar\a\L^2}} \frac{\bar r}{\bar\a} (P+\textrm{shift})_\n\SiN{(P+L)\th Q}{4}{\!}\SiN{(P-Q)\th L}{4}{\!} \!\SiN{(Q+L)\th P}{4}{\!} \nn\\*
&\approx -\frac{1}{6} \pi ^2 \L^2 \left(\tilde{L}_{\nu }
   \tilde{Q}\cdot \tilde{Q}+\tilde{L}\cdot \tilde{L}
   \tilde{Q}_{\nu }+2 \left(\tilde{L}_{\nu }+\tilde{Q}_{\nu
   }\right) \tilde{L}\cdot \tilde{Q}\right)
. \label{V3-approx-2}
\end{align}
\end{subequations}
Putting everything together, \eq{V3-approx-1} and \eq{V3-approx-2} finally lead to the 3-field contribution 
\begin{align}
&\quad\int\limits_0^{\infty}\frac{d\a}{\a}\int\limits_0^{\a}dt\int\limits_0^{t}dr r
 \Tr\left(Ve^{-r H_0}Ve^{-(t-r)H_0}Ve^{-(\a-t)H_0}\right) 
e^{-\inv{\a L^2}} \Big|_{\textrm{3 fields}} \nn\\*
&= \frac{i\tr\one}{6}   \int\!\frac{d^4Qd^4L}{(2\pi\LNC^2)^{4}} \Bigg\{
 \frac{\L^4}{4\LNC^4}\bG^{\n\n'}\bG^{\s\s'}\bG^{\e\e'} A_{\n'} (Q)A_{\s'}(L)A_{\e'}(-Q-L)\times \nn\\
& \quad \times\!\Bigg(\!\tilde{Q}\cdot \tilde{Q}
   \big(\tilde{L}_{\e } \bar G_{\nu  \s}
   +\tilde{L}_{\nu } \bar G_{\e  \s }+\tilde{L}_{\s} \bar G_{\e  \nu }\big)\!
  +\tilde{L}\cdot \tilde{L} \big(\tilde{Q}_{\e } \bar G_{\nu  \s }
   +\tilde{Q}_{\nu } \bar G_{\e \s}  + \bar G_{\e  \nu } \tilde{Q}_{\s } \big)\! 
 + 2\tilde{Q}_{\s } \tilde{L}_{\nu } (\tilde{L}+\tilde{Q})_{\e } \nn\\*
&\quad 
  +2 \tilde{L}_{\s } \tilde{Q}_{\e }(\tilde{L}_{\nu } +\tilde{Q}_{\nu })  
 +2 \tilde{L}\cdot \tilde{Q}
   \left(\bar G_{\nu  \s } (\tilde{L} +\tilde{Q})_{\e }+\bar G_{\e  \s } (\tilde{L}+\tilde{Q})_{\nu }
  +\bar G_{\e\nu } (\tilde{L}+\tilde{Q})_{\s}\right) \nn\\*
&\quad  +2\tilde{L}_{\e } \tilde{Q}_{\nu }(\tilde{Q}_{\s }+\tilde{L}_{\s})\! \Bigg) 
 + \frac{\L^2\LNC^4}{16}\bG^{\m\n} \left(\tilde{L}_{\nu }
   \tilde{Q}\cdot \tilde{Q}+\tilde{L}\cdot \tilde{L}
   \tilde{Q}_{\nu }+2 \left(\tilde{L}_{\nu }+\tilde{Q}_{\nu
   }\right) \tilde{L}\cdot \tilde{Q}\right) \times \nn\\
&\quad \times\! \left(A_\m(Q)\Th^{ab}(L)\Th_{ab}(-Q\!-\!L)+\Th^{ab}(Q)A_\m(L)\Th_{ab}(-Q\!-\!L)+\Th^{ab}(Q)\Th_{ab}(L)A_\m(-Q\!-\!L)\right) \nn\\
&\quad +\mbox{higher orders}\Bigg\}
. \label{V3-A3-approx}
\end{align}
Note that the replacement $\a = \LNC^4\bar \a$ etc. provides a factor $\LNC^{12}$.

\section{Effective NC gauge theory action}
\label{sec:effgaugethy}

Now we can recast the above results into an effective gauge theory action,
organized in terms of engineering dimension. 
As discussed in \secref{sec:strategy}, this  arises systematically due to the expansion in three small parameters 
$\epsilon_L(p)$, $(p\th q) \sim \frac{p^2}{\LNC^2}$ and $\frac{p^2}{\L^2}$, imposing
the IR condition \eq{IR-regime}.
We will systematically compute all terms of operator dimension $\leq 6$.

It is important to note that since $V$ is either linear or quadratic in the fields $(A,\vph)$,
there is only a finite number of terms in the perturbation expansion
which can produce gauge theory terms involving $n$ fields (i.e. of order $V^n$ up to $V^{2n}$). 
Therefore these are completely determined by the perturbative expansion.

We will first consider the  quadratic terms in $A_\mu$ resp. $\phi^i$, which
arise from the first and second order terms
in $V$. They may contain arbitrarily high powers of momenta resp. 
derivatives. Due to translational invariance \eq{transl-inv},
the leading term in this expansion is quadratic in momenta
and quadratic in the fields, of type $\L^4 \int \cO(p^2(A,\vph)^2)$.
This is the usual vacuum energy contribution in 
quantum field theory, which diverges as $\L^4$. In the present context, we will denote it as 
``potential'' term since it governs the vacuum structure of the NC brane solution
in the flat case, in particular $\theta^{\mu\nu}$ and the dilaton. 
Gauge invariance then requires the presence of certain cubic terms in the 
fields, which 
will be verified in detail\footnote{The quartic terms are not verified here in order
to keep the paper within reasonable bounds.}.

Next, we will analyze the dimension 6 operators 
with structure $\L^2\! \int\! \cO(p^4 (A,\varphi)^2)$ in a similar way,
leading to curvature-type terms. Again, gauge invariance will be verified
up to the cubic terms in the fields. However, there will also be terms proportional 
$\L^6\!\int\! \cO((A,\vph)^2 p^4)$ due to UV/IR mixing resp. factors of $\epsilon_L(p)$, which 
also correspond to curvature contributions. Higher-order terms 
of dimension 8 and higher will not be analyzed further in this paper.

To make contact with the commutative case, one should consider the case 
$\frac \L\LNC \ll 1$. Then UV/IR mixing terms would 
give an expansion in this small parameter.

\subsection{\texorpdfstring{$\L^4$}{Lambda**4} potential terms}

\paragraph{Two field contributions.}

We have contributions from the first-order term \eq{first-order-result}
as well as from the second-order term \eq{2-field-collect}, where
both the $AA$ terms as well as the $\Th\Th$ terms contain $D\phi D\phi$. 
The complete action with engineering dimension 4
is as expected proportional to $\L^4$, given by 
{\allowdisplaybreaks
\begin{align}
\Gamma_{\L^4}^{(\textrm{2})}((A,\varphi)^2,p^2) 
&=\frac{\tr\one}{16}\,\frac{\L^4}{\LNC^8}\Intt{l} \sqrt{g}
\Big( l^2\varphi^i(-l)\varphi_i(l)
- 2\LNC^4\bar\th^{\m\n}l_\nu A_\m(l)\bar\th^{\r\s} l_\s A_\r(-l)\Big) \nn\\
&= \frac{\tr\one}{16}\,\frac{\L^4}{\LNC^4}\int\frac{d^4x}{(2\pi)^2} \sqrt{g}
\Big(g^{\a\b}\del_\a\varphi^i\del_\b\varphi_i
- \frac 12\LNC^4\bar\th^{\m\n} F_{\nu\m}\bar\th^{\r\s} F_{\s\r} \quad + \cO(A^3)\Big)
 \label{2-field-collect-b}
\end{align}
using \eq{lcdotl-identity}.
The result essentially gauge-invariant up to $\cO(A^3)$ and $\cO(A^4)$ terms, which will be recovered 
from higher-order terms. This is  
consistent with previous results \cite{Steinacker:2008a}, where the fermionic one-loop action 
was computed on $\R^4_\th$. 
Note that there is no renormalization of the bare Yang-Mills action \eq{YM-action}. 
}

\paragraph{Three field contributions.}

The dimension 6 contributions from $\cO(V^2)$ due to \eq{3-field-collect} proportional to 
$\L^4$ are given by
\begin{align}
\Gamma^{(2)}((A,\varphi)^3,p^3) 
&= -i \frac{\L^4\tr\one\sqrt{\bG}}{16\LNC^8} \int\!\!\frac{d^4qd^4l}{(2\pi\LNC^2)^{4}}\bigg\{
\!\! \left( q^2 \tilde l^{\mu} + l^2\tilde q^{\mu} + 2 q l (\tilde l^{\mu}+ \tilde q^{\mu})\right)\!\times \nn\\*
&\quad \times\! \left(\bG^{\r\s}\!A_\m(q)A_\r(-q\!-\!l)A_\s(l)
 +A_\m(q)\vph^i(-q\!-\!l)\vph_i(l) \right)   + \cO(p^4 A^3) \!\bigg\}  
\label{3-field-collect-V2}
\end{align}
using \eq{lcdotl-identity} and 
dropping higher-order terms arising e.g. from $(l\th q)^2$. 
The contributions from $\cO(V^3)$ due to \eq{V3-A3-approx} are
\begin{align}
\Gamma^{(3)}((A,\varphi)^3,p^3)
&=  \frac{i\L^4\tr\one}{16 \LNC^8}  \int\!\frac{d^4qd^4l}{(2\pi\LNC^2)^{4}} 
  A_\nu(q)A_\s(l)A_\e(-q-l) \times \nn\\
& \qquad \times\bigg(\bar G^{\e \s}
 \big(q^2 \tilde l^{\nu } + l^2 \tilde q^{\nu } + 2 q l (\tilde l^{\nu } + \tilde q^{\nu })\big)
+ 2 \LNC^4 \,\tilde{l}^{\s } \tilde{q}^{\e}(\tilde{l}^{\nu }  + \tilde{q}^{\nu }) \bigg)
\nn
\end{align}
(using $q \to -l-q$ at some point).
This cancels precisely the term involving  $G^{\mu\nu}A_\mu A_\nu$ in \eq{3-field-collect-V2}, 
and the combined 3-field contribution up to dimension 6 operators
proportional to $\L^4$ is 
\begin{align}
& \Gamma_{\L^4}((A,\varphi)^3,p^3)  \nn\\
&= -i \frac{\L^4\tr\one\sqrt{\bG}}{16\LNC^8} \int\!\frac{d^4qd^4l}{(2\pi\LNC^2)^{4}}
\bigg(A_\m(q)\vph^i(-q-l)\vph_i(l) 
 \bigg(2 q l (\tilde l^{\mu} + \tilde q^{\mu})  
 + q^2 \tilde l^{\mu} + l^2 \tilde q^{\mu} \bigg)  \nn\\
& \qquad -2 \LNC^4 \,(\tilde{l}^{\s }A_\s(l)) A_\nu(q) A_\e(-q-l) \tilde{q}^{\e}(\tilde{l}^{\nu }  + \tilde{q}^{\nu }) \bigg) \nn\\
&= -\frac{\L^4\tr\one\sqrt{\bG}}{16\LNC^4} \int\!\frac{d^4x}{(2\pi)^{2}}
\bigg(\bar\th^{\mu\nu}g^{\a\b} \big( 2\del_\a A_\m\del_\nu\vph^i\del_\b\vph_i
 - \del_\a\del_\b A_\m\vph^i\del_\nu\vph_i 
 - \del_\nu A_\m\vph^i\del_\a\del_\b\vph_i \big)  \nn\\
& \qquad -2 \LNC^4 \,(\bar\th^{\s\s'}\del_{\s'}A_{\s}) \bar\th^{\e\e'}\del_{\e'}A_{\nu}
\bar\th^{\nu\nu'}\del_{\nu' }A_{\e} \bigg) 
. \label{3-field-collect2-V2}
\end{align}
Now the $AAA$ terms can be simplified using
\eqn{
-2(\bar\th^{\s\s'}\del_{\s'}A_{\s}) \bar\th^{\e\e'}\del_{\e'}A_{\nu}\bar\th^{\nu\nu'}\del_{\nu' }A_{\e} 
&= (\bar\th^{\s\s'}F_{\s\s'}) \bar\th^{\e\e'}(F_{\e'\nu} + \del_{\nu}A_{\e'}) \bar\th^{\nu\nu'}\del_{\nu' }A_{\e} \nn\\
&= (\bar\th^{\s\s'}F_{\s\s'}) \bar\th^{\e\e'}
(\frac 12 F_{\e'\nu}\bar\th^{\nu\nu'}F_{\nu'\e}  - i[A_{\e'},A_{\e}]_\th)
}
up to quartic terms,
replacing the commutator with a Poisson bracket 
to leading order in $\th$. Here we used the identity
\be
F_{\e'\nu}\bar\th^{\nu\nu'}\del_{\nu' }A_{\e} = \frac 12 F_{\e'\nu}\bar\th^{\nu\nu'}F_{\nu'\e} 
\label{S3-gaugeinv}
\ee
(up to cubic terms),
which can be seen by renaming $\e'\leftrightarrow\nu, \nu'\leftrightarrow\e$.
Similarly, 
the $A\varphi\varphi$ terms can be simplified using partial integration as follows
\eqn{
&\quad \int\!\frac{d^4x}{(2\pi)^{2}} \bar\th^{\mu\nu}g^{\a\b} \big( 2\del_\a A_\m\del_\nu\vph^i\del_\b\vph_i
 - \del_\a\del_\b A_\m\vph^i\del_\nu\vph_i 
-  \del_\nu A_\m\vph^i\del_\a\del_\b\vph_i \big) \nn\\
&= \int\!\frac{d^4x}{(2\pi)^{2}}\bar\th^{\mu\nu}g^{\a\b} \big(3 \del_\a A_\m\del_\nu\vph^i\del_\b\vph_i
 + \del_\b A_\m\vph^i\del_\a\del_\nu\vph_i 
 - \frac 12 F_{\nu\mu}\vph^i\del_\a\del_\b \vph_i \big) + \textrm{h.o.} \nn\\
&= \int\!\frac{d^4x}{(2\pi)^{2}}\bar\th^{\mu\nu}g^{\a\b} \big(
2 (F_{\a\m} + \del_\mu A_\a)\del_\nu\vph^i\del_\b\vph_i
 + \frac 12  F_{\nu\m} \del_\b\vph^i\del_\a\vph_i \big) + \textrm{h.o.}
\label{Avpvp-term}
}
Thus we get
\begin{align}
 \Gamma_{\L^4}((A,\varphi)^3,p^3)  
&=  -\frac{\L^4\tr\one\sqrt{g}}{16\LNC^4} \int\!\frac{d^4x}{(2\pi)^{2}}
\bigg(g^{\a\b} \big(
2\bar\th^{\mu\nu}F_{\a\m}\del_\nu\vph^i\del_\b\vph_i 
- 2 i [A_\a,\vph^i]_\th\del_\b\vph_i\nn\\*
& \quad
 - \frac 12 (\bar\th^{\mu\nu}F_{\mu\nu}) \del_\b\vph^i\del_\a\vph_i
 + \LNC^4 \,(\bar\th^{\s\s'}F_{\s\s'}) \bar\th^{\e\e'}
(\frac 12 F_{\e'\nu}\bar\th^{\nu\nu'}F_{\nu'\e}  - i[A_{\e'},A_{\e}]_\th) \bigg)\! .\nn
\end{align}
The commutator terms provide precisely the missing cubic terms for the 
gauge-invariant completion of \eq{2-field-collect-b}, so that the complete induced potential 
 including all terms with dimension up to 6 is given by 
\begin{fboxalign}
 \Gamma_{\L^4}((A,\vp,p)^{4-6}) &= \frac{\tr\one}{16}\,\frac{\L^4}{\LNC^4}\int\frac{d^4x}{(2\pi)^2} \sqrt{g}
\Big(g^{\a\b}D_\a\varphi^i D_\b\varphi_i \nn\\
& \quad - \frac 12\LNC^4 \big( \bar\th^{\m\n} F_{\nu\m}\bar\th^{\r\s} F_{\s\r}
 + (\bar\th^{\s\s'}F_{\s\s'}) (F\bar\th F\bar\th)\big) \nn\\
& \quad - 2\bar\th^{\nu\mu}F_{\m\a} g^{\a\b}\del_\nu\vph^i\del_\b\vph_i 
 +\frac 12(\bar\th^{\mu\nu}F_{\mu\nu}) g^{\a\b}  \del_\b\vph^i\del_\a\vph_i  
\quad + \textrm{h.o.}\Big), 
 \label{potential-collect}
\end{fboxalign}
which is manifestly gauge invariant. Remarkably, this result 
will be precisely recovered from the simple matrix model effective action 
\eq{vacuuum-energy-matrix}. 
This is a strong confirmation of the $SO(D)$ symmetry,
demonstrating the power of the matrix model point of view.

\subsection{\texorpdfstring{$\cO(\L^2)$}{O(Lambda**2)} curvature terms}

Consider now the dimension 6 terms proportional to $\L^2$. 
They must be gauge-invariant by themselves. 
The terms quadratic in the fields are
{\allowdisplaybreaks
\begin{align}
&\Gamma_{\L^2}((A,\varphi)^2,p^4)  \nn\\*
&= -\frac 14\frac{\tr\one}{24}\,\frac{\L^2}{\LNC^4}\Intt{q} \sqrt{g}
   \Bigg( 6\LNC^8 \Th^{ab}(q)\Th_{ab}(-q)\tilde q\cdot\tilde q \nn\\*
 &\quad  -(\bar G_{\nu\s} q\cdot q - q_\n q_\s)\tilde q\cdot\tilde q   \bG^{\m\n}\bG^{\r\s}A_\m(q)A_\r(-q) 
- 2 q\cdot q \qt_\nu \tilde q_\s  \bG^{\m\n}\bG^{\r\s}A_\m(q)A_\r(-q) \Bigg)\nn\\
&= \frac 14\frac{\tr\one \L^2}{24}\int\!\!\frac{d^4x}{(2\pi)^2} \sqrt{g}
   \Bigg(- 6\LNC^{12}\Th^{ab}\bar\Box_g\Th_{ab}  \nn\\*
& \qquad  + \frac 12 \LNC^{4}  F_{\mu\nu} \bar\Box_g F_{\mu'\nu'} \bG^{\m\m'} \bG^{\n\n'}
+ \frac 12 \LNC^4 (\bar\theta^{\m\n} F_{\m\n}) \bar\Box_G(\bar\theta^{\r\s} F_{\r\s}) 
\quad + \cO(A^3) \Bigg)  \nn\\
&= \frac 14\frac{\tr\one \L^2}{24}\int\!\!\frac{d^4x}{(2\pi)^2} \sqrt{g}
   \Bigg(-\frac{11}2 \LNC^4 F_{\r\e}\bar\Box_g F_{\s\t} \bG^{\r\s}\bG^{\e\t}
 -12 \LNC^{8}\bar\Box_g\varphi^i \bar\Box\varphi_i    \nn\\
& \qquad  + \frac 12 \LNC^4 (\bar\theta^{\m\n} F_{\m\n}) \bar\Box_G(\bar\theta^{\r\s} F_{\r\s}) 
\quad + \cO(A^3) \Bigg)  
. \label{2-field-collect-curv1}
\end{align}
This is indeed gauge-invariant
up to $\cO(A^3)$ terms, which should be recovered later.
Note that there are different Laplacians \eq{Moyal-laplaceops} in this expression
such as $\bar\Box_g$, corresponding to the matrix operators 
but for the {\moyal} background. They contain powers of $\LNC$.
This will facilitate the comparison with the matrix model expressions.
}

\paragraph{Three field contributions.}

The 3-field contributions proportional to $\L^2$ from $\cO(V^2)$ start at $\cO(p^5)$
i.e. dimension 8, which is not considered here.
The contributions from $\cO(V^3)$ due to \eq{3-field-collect} are given by
\begin{align}
&\Gamma^{(3)}((A,\varphi)^3,p^3) \nn\\
&= \frac{i\tr\one}{12}   \int\!\frac{d^4Qd^4L}{(2\pi\LNC^2)^{4}} \Bigg(
  \frac{\L^2\LNC^4}{16}\bG^{\m\n} \left(\tilde{L}_{\nu }
   \tilde{Q}\cdot \tilde{Q}+\tilde{L}\cdot \tilde{L}
   \tilde{Q}_{\nu }+2 \left(\tilde{L}_{\nu }+\tilde{Q}_{\nu}\right) \tilde{L}\cdot \tilde{Q}\right) \times \nn\\
&\quad \times\!\! \left(\! A_\m(Q)\Th^{ab}(L)\Th_{ab}(-Q\!-\!L)
 +\Th^{ab}(Q)A_\m(L)\Th_{ab}(-Q\!-\!L)
 +\Th^{ab}(Q)\Th_{ab}(L)A_\m(-Q\!-\!L)\!\right) \!\!\! \Bigg)   \nn\\
&=  \frac{\tr\one}{4} \frac{\L^2 \LNC^4}{16}\int\!\!\frac{d^4 x}{(2\pi)^{2}} 
 \bar\theta^{\m\n} g^{\a\b} 
\left(\del_\a\del_\b A_\m \del_{\nu} \Th^{ab}\Th_{ab}
 + \del_{\nu} A_\m \del_\a\del_\b \Th^{ab}\Th_{ab} 
 - 2 \del_\a A_\m \del_\b \Th^{ab} \del_{\nu} \Th_{ab}
\right)  \nn\\
&= -\frac{\tr\one}{4} \frac{\L^2 \LNC^4}{16}\int\!\frac{d^4 x}{(2\pi)^{2}} 
 \bar\theta^{\m\n} g^{\a\b}
\left(2 (F_{\a\m} + \del_\mu A_\a)\del_\nu\Th^{ab}\del_\b\Th_{ab}
 + \frac 12  F_{\nu\m} \del_\b\Th^{ab}\del_\a\Th_{ab}
+\textrm{h.o.} \right)  \nn\\
&\sim \frac{\L^2 \LNC^4 \tr\one}{64}\!\int\!\!\frac{d^4 x}{(2\pi)^{2}} 
g^{\a\b} \!\left(\! 2 i [A_\a,\Th^{ab}]_\th \del_\b\Th_{ab} 
- 2 \bar\theta^{\m\n} F_{\a\m}\del_\nu\Th^{ab}\del_\b\Th_{ab} 
- \tinv2 \bar\theta^{\m\n} F_{\nu\m} \del_\b\Th^{ab}\del_\a\Th_{ab}
\!\right)  
\end{align}
using \eq{Avpvp-term}.
The middle term is clearly part of a covariant derivative 
$ D_\a \Th^{ab} D_\b\Th_{ab}$. However 
these are also dimension 8 operators which we will not consider any further in this paper.

\subsection{\texorpdfstring{$\cO(\L^6)$}{O(Lambda**6)} curvature terms}

Finally consider the dimension 6 terms proportional to $\L^6$, quadratic in the fields. 
There are also contributions from the $\cO(V)$ terms \eq{Gamma-eff-1}, due to 
\be
\L^2 -  \Leff^2(l) = \L^2 -\frac{\L^2}{1 + \inv4\frac{\L^2}{\LNC^4}\, l^2}  
= \inv4 \frac{\L^4 l^2}{\LNC^4} - \inv{16} \frac{\L^6 l^4}{\LNC^8} + \ldots 
\,. \label{Leff-expand-2}
\ee
This gives
\begin{align}
&\Gamma_{\L^6}((A,\varphi)^2,p^4) \nn\\
&= \frac{\tr\one}{16}\,\frac{\L^6}{\LNC^{12}}\Intt{l} \sqrt{g} 
\Bigg(-\frac 14 (l^2)^2 \varphi^i(-l)\varphi_i(l)
 + \LNC^4\bG^{\m\n}\bG^{\r\s}A_\m(q)A_\r(-q) q^2\tilde{q}_{\nu } \tilde{q}_{\sigma }  \Bigg)\nn\\
 &= 
\frac{\tr\one}{16}\,\frac{\L^6}{\LNC^{8}}\int\frac{d^4 x}{(2\pi)^2} \sqrt{g} 
\Bigg(-\frac 14 \LNC^{16}\bar\Box_g\varphi^i \bar\Box_g\varphi_i
 + \frac 14\LNC^{12} (\bar\th^{\m\n} F_{\nu\m})\bar\Box_g(\bar\th^{\r\s}F_{\s\r})
\quad + \cO(A^3)  \Bigg)
. \label{2-field-collect-curv2}
\end{align}
Again this is indeed gauge-invariant, which constitutes a non-trivial check of our 1-loop 
computations.
The 3-field contributions at $O(\L^6)$ due to \eq{3-field-collect}
as well as the
missing terms for the gauge-invariant completion of $F$ or $\bar\Box_g$
have dimension 8 or higher (upon expanding $[A,A]$), which we do not consider here. 

This $\L^6$ contribution is expected to arise from $\cO(X^{14})$ terms in the matrix model 
such as $\Box_g X^a \Box_g X^a$, which however will not be worked out in this paper. 
Note that this term is comparable with the $\L^2$ contribution if $\L \approx \LNC$,
but is negligible if $\L \ll \LNC$.
It is also worth pointing out that
this contribution may not be obtained in a semi-classical analysis
along the lines of~\cite{Steinacker:2008a,Klammer:2009dj}, because it involves a higher-order 
contribution in the UV/IR mixing term \eq{Leff-expand-2}.

\subsection{Free contribution}

The induced action resp. vacuum energy for 
the free case  (i.e. the constant term in \eq{engineering-expand}) is given by
\eqn{
\Gamma_L[\bar X]
&= -\inv{2}\Tr\int\limits_0^\infty\frac{d\a}{\a} e^{-\a\sD_0^2-\inv{\a L^2}}  
=  -\inv{2}\tr \one\, \int\limits_0^\infty\frac{d\bar\a}{\bar\a} 
 \int \frac{d^4 p}{(2\pi)^2} e^{-\bar\a p\cdot p -\inv{\a \L^2}} 
\int \frac{d^4 x}{(2\pi)^2}\,\sqrt{g}  \nn\\
&=  -\frac{\L^4 \,\tr \one}{8} \, \int \frac{d^4 x}{(2\pi)^2}\,\sqrt{g} 
\,. \label{vacuum-free}
}
Remarkably, this short computation --- along with general geometrical considerations 
--- suffices to predict all of the above loop computations (and beyond)
for the potential, as explained below.

\section{Effective matrix model action}
\label{sec:effMM}

In this section, we will determine an effective matrix model action  
which reproduces the induced gauge theory action obtained above.
The scaling law \eq{basic-scaling} combined with gauge invariance suggests to 
consider action functionals of the form
\be
\Gamma_L[X] = \Tr \cL\Big(\frac{X^a}{L}\Big) .
\label{eff-action}
\ee
Since we need to reproduce terms which diverge as $\sim \L^n = L^n \LNC^{2n}$ for $n>0$, 
it is clearly not sufficient to assume that $\cL(X)$ is a polynomial or a power series in $X^a$.
For example, the effective matrix model action corresponding to the 
vacuum energy contributions \eq{potential-collect} must be homogeneous functions of degree $-4$ in $X^a$.
This lack of ``analyticity'' in $X^a$ should not be surprising, since
the loop computation is based on the assumption of a
4-dimensional background, described by a deformation of $\R^4_\th$. On 
such 4-dimensional backgrounds,  one can generalize the polynomial functions to a certain
class of holomorphic functions in $X^a$ which are analytic {\em near such backgrounds}. This 
is indeed what happens, which implies the existence of singularities at other 
locations in the moduli space of Hermitian matrices. This should of course be expected,
since the same matrix model accommodates spaces of different dimensions, 
which surely lead to very different quantum effects.

Gauge invariance strongly restricts the possible form of the action. 
We restrict ourselves to single-trace terms, which is
sufficient for the present context without non-Abelian gauge fields.
Furthermore, the $SO(D)$ symmetry of the bare matrix model and the 
$SO(D)$-invariant regularization \eq{TrLog-id} strongly suggest that the effective action should also 
be manifestly invariant under $SO(D)$. However this is a non-trivial statement,
which might be spoiled by anomalies
due to the infinite-dimensional nature of the matrix model. 
Nevertheless, we will indeed find an effective matrix model
action with manifest $SO(D)$ symmetry, which reproduces 
all of the loop results as far as we can verify them. 
This is a highly non-trivial result, which predicts infinitely many higher-order terms
in the induced action of {\nc} gauge theory.  
All terms of dimension 6 are verified in detail, 
providing  very strong 
support for the $SO(D)$ symmetry and the effective action given below.

It should be pointed out that since the effective matrix model 
action holds for generic deformations of 
$\R^4_\th$, it provides a background-independent action for gravity 
(as well as non-Abelian gauge theory) for 
4-dimensional non-commutative branes $\cM^4_\th \subset \R^{10}$. 
This is a key issue of emergent gravity.

\subsection{Effective potential \texorpdfstring{$V(X)$}{V(X)}}
\label{sec:potential}

Due to translational invariance \eq{transl-inv},
the effective action $\Gamma_L = \Tr \cL(X)$ can be written in terms of commutators, and
organized in terms of the order of commutators. 
Commutators correspond to derivative operators for the gauge fields. They
may arise as simple commutators $[X^a,X^b]$, or multiple commutators such as
$[X^a,[X^b,X^c]]$ corresponding to higher derivatives. 
The leading term $\cL(X) = V(X) + \textrm{h.o.}$ in such a ``momentum expansion'' of the effective matrix action
will have only simple commutators, resp. 
first-order derivatives of the gauge theory language. We will denote such simple-commutator terms as 
{\em effective potential} $V(X) \equiv V([X,X])$. From the emergent gravity point of view, 
it depends only on the tensor fields $g_{\mu\nu}(x)$ and $\theta^{\mu\nu}(x)$ rather than their
derivatives (notably curvature). Hence $V(X)$ will govern the vacuum in the flat case, 
justifying the name potential. 
It corresponds to the 
vacuum energy, cf. \cite{Blaschke:2010rg} for a related discussion.

Taking into account (or assuming) also the $SO(D)$ symmetry, 
it follows that $V(X)$  can be written in terms of contractions of 
$\Theta^{ab}$ with $g_{ab} = \d_{ab}$. Since different orderings of products of $\Theta^{ab}$ 
differ by commutators, 
it is enough to consider only contractions of 
neighbouring indices, in the cyclic sense.
This means that $V(X)$ can be written in terms of products of\footnote{Note that $J$ defines an 
almost-complex structure under certain natural conditions \cite{Steinacker:2008ya}.} 
\be
J^a_b := i\Theta^{ac}g_{cb} \, = \, [X^a,X_b] , \qquad  \tr J \equiv J^a_a = 0 \,, 
\ee
where $\tr$ will denote the trace over the $SO(D)$ indices.
Since $\Theta^{ab}$ is anti-symmetric, 
$V(X)$ can be written in terms of products of traces of even powers of $J$.

By analyzing the possible effective potential terms in the semi-classical limit,
we can further narrow down the possible form of $V(X)$ by 
recalling the (semi-classical) characteristic equation which holds for generic 
4-dimensional branes $\cM^4 \subset \R^D$ \cite{Steinacker:2008ya,Blaschke:2010qj}:
\be
(J^4)^a_b -\frac 12 (\tr J^2)\, (J^2)^a_b \, \sim \, - \LNC^{-8}(x) (\cP_T)^a_b \,,  \qquad 
J^5 - \frac 12 (\tr J^2)\, \, J^3 \, \sim\,  - \LNC^{-8}(x)\, J \,, 
\label{char-4D}
\ee
where the semi-classical object $\cP_T^{ab}:=g^{\m\n}\pa_\m x^a\pa_\n x^b$ 
is the projector on the tangential bundle (cf. \eqref{eq:matrix-splitting}), 
and $\LNC^{-4}(x) = \sqrt{|\theta^{\mu\nu}(x)|}$ denotes the scale defined by the full 
NC structure $\theta^{-1}_{\mu\nu}(x) = \bar \theta^{-1}_{\mu\nu} + F_{\mu\nu}$. 
This follows from the fact that $J$ defines a 4-dimensional 
anti-symmetric tensor field in the semi-classical limit.
Furthermore, 
\be
\tr J^2 \, \sim \, \LNC^{-4}(x)\, G^{\mu\nu}(x) g_{\mu\nu}(x) \, \equiv  \LNC^{-4}\, (Gg) .
\ee
Due to these relations, any expression
\bea
\tr J^{2n} \,, \quad n\geq 6
\,, \label{Hab-chain}
\eea
can be reduced semi-classically\footnote{In the fully NC case, this argument strictly speaking
applies only up to higher-order corrections in $\theta^{\mu\nu}$. 
Nevertheless, it appears to work.} to a function of $\tr J^4$ and  $\tr J^2$. 
Therefore the most general (single-trace) form of the effective potential compatible with the 
scaling  has the form 
$V(X) = V\big(\frac{L^4}{\tr J^2},\frac{\tr J^4}{(\tr J^2)^2}\big)$, or equivalently
\be
V(X) = V\Big(-\frac{L^4}{\tr J^2},\frac{-\tr J^4 + \frac 12 (\tr J^2)^2}{(\tr J^2)^2}\Big) \,
\,\sim \,  V\Big(\frac{L^4}{\LNC^{-4}(x)\, (G g)},\frac{4}{(Gg)^2}\Big) .
\ee
Note that both arguments of $V(z_1,z_2)$ are Hermitian matrices
in the adjoint of $U(\infty)$,
which are invertible on 4-dimensional NC branes with $\LNC \neq 0$. Therefore 
the above $V(X)$ is an admissible candidate for the effective potential, and 
$\Tr V(X)$ is well-defined and gauge invariant
provided $V$ is analytic (and real-valued) if both $z_1$ and $z_2$ are on the positive real line. 

It only remains to determine the function $V(z_1,z_2)$.
Remarkably, this can be determined already from the vacuum energy in the free\footnote{Recall that all 
higher-order commutators vanish on $\R^4_\theta$.} 
case \eq{vacuum-free}, 
which is proportional to $\L^4$ (in agreement with \eq{potential-collect}) 
and does not depend on $(Gg)$ (which measures the deviation of $\theta^{\mu\nu}$ from 
the (anti-)selfdual case). It follows that $V \sim z_1/\sqrt{z_2}$, i.e.
\begin{align}
&\Gamma_L[X]  = \,  \Tr V(X) + \textrm{h.o.}, \nn\\
& \fbox{$ \Tr V(X) =  -\frac 14 \Tr \bigg(\frac{L^4}{\sqrt{-\tr J^4 + \frac 12 (\tr J^2)^2}}\bigg)  
\, \sim \,  -\frac 18\frac 1{(2\pi)^2} \int\! d^4 x\, \L^4(x) \sqrt{g} \,. 
$}
\label{vacuuum-energy-matrix}
\end{align}
The normalization factor is obtained from \eq{vacuum-free},
identifying the rhs with the semi-classical vacuum energy. 
``h.o.'' stands for higher-order commutator terms, 
i.e. curvature contributions etc. which will be considered below.
As a quick check observe that $\sqrt{g}\sim\sqrt{\bar g}(1+\inv2 g\del\phi\del \phi+\ldots)$
for a non-trivial background,
in agreement with \eq{potential-collect}.
We therefore conjecture that the potential term in the induced 
effective action for generic 4-dimensional backgrounds is given by \eq{vacuuum-energy-matrix}. 
We will verify below that this reproduces precisely the
above loop computations up to the order computed here. 
This demonstrates the power of the geometric view of the matrix model.

Note that the bare matrix model 
\eq{YM-action} could be viewed as a potential with 
$V(z_1,z_2) = \frac 1{z_1}$. However that term would be proportional to 
$\L^{-4}$ (rather than $\log\L$ as one might suspect), which is 
highly suppressed and not considered here.

It is remarkable that the vacuum energy \eq{vacuuum-energy-matrix} has a very non-trivial
structure from the matrix model point of view. This should have very interesting physical 
consequences in particular for the cosmological constant problem. 
However a more complete picture, notably including the 
contributions from the bosonic sector, is 
required before its physical consequences can be addressed. 

An analogous expression should work in other dimensions, which should be studied elsewhere.
In the presence of non-Abelian gauge fields, the action also will need to be generalized.

In order to verify \eq{vacuuum-energy-matrix}, we simply have to include fluctuations 
to the matrices $X^a$ around $\R^4_\th$, expand it  to any desired order, 
and compare the result with the induced gauge theory action 
computed above. This will be done in detail below. 
However a partial comparison can be made easily using the geometrical point of view.

\paragraph{Semi-classical analysis}

The gauge sector of the above terms can be obtained quickly 
by setting  $\del_\mu \phi^i = 0$, resp. more generally by 
going to  normal embedding coordinates \cite{Steinacker:2010rh}. Then
\eq{char-4D} gives for the semi-classical limit
\be
\tr J^4  - \inv2 (\tr J^2)^2 \,\sim\, -\LNC^{-8}(x)\, \tr \cP_T
\,=\, -4 ({\rm Pfaff}(\th(x)))^2  
\,=\, - \frac 1{16} \(\vare_{\mu\nu\a\b} \th^{\mu\nu}(x) \th^{\a\b}(x)\)^2
\ee
where $\th_{\mu\nu}^{-1}(x) = \bar\th^{-1}_{\mu\nu} + F_{\mu\nu}(x)$ is the full 
symplectic structure, and $\cP_T$ is the projector on the tangential bundle. 
The Pfaffian of an anti-symmetric $4\times 4$ matrix is defined as  
\bea
\textrm{Pfaff}(F_{\m\n}) = \inv8 \vare^{\m\n\r\eta}F_{\m\n}F_{\r\eta} 
\, \label{pfaff-explicit}
\eea
and coincides with $\pm \sqrt{|F|}$.
This yields
\eqn{
\sqrt{-\tr J^4  + \frac 12 (\tr J^2)^2} \,\, &\sim \,  
\frac 14\varepsilon_{\mu\nu\a\b} \th^{\mu\nu} \th^{\a\b} =
\frac 14\det\bar\th\, \varepsilon^{\mu\nu\a\b} (\bar\th^{-1}_{\mu\nu} + F_{\mu\nu})
 (\bar\th^{-1}_{\a\b} + F_{\a\b})\nn\\
&= 2\LNC^{-4} \, \(1 - \frac 12 \bar\th^{\mu\nu} F_{\mu\nu} +  \LNC^{-4}{\rm Pfaff} (F) \)
}
which is actually correct to all orders in $F$.
Here we use
\bea
\frac 18\varepsilon^{\mu\nu\a\b} \bar\th^{-1}_{\mu\nu} = -\frac 14\rm{Pfaff}\bar\th^{-1} \bar \th^{\a\b} 
= -\frac 14\LNC^4 \bar \th^{\a\b} ,
\eea
which can be seen e.g. 
using the standard form \eq{theta-standard-general-E} for $\bar\th$.
Using \eq{FF-theta} this gives
\eqn{
\inv{\sqrt{-\tr J^4  + \frac 12 (\tr J^2)^2}}  &= 
\frac 12\LNC^{4} \Big(1 + \frac 12 \bar\th^{\mu\nu} F_{\mu\nu}  + \frac 14 (\bar\th F)^2 + \frac 18 (F^{\mu\nu}\th_{\mu\nu})^3
-   \bar\th^{\mu\nu} F_{\mu\nu} \LNC^{-4}{\rm Pfaff} (F) \nn\\
&\quad - \LNC^{-4}{\rm Pfaff} (F) + \cO(F^4)   \Big) \nn\\
&= \frac{\LNC^{4}}2 \left(\! 1 + \inv2 \bar\th^{\m\n} F_{\m\n}  + \inv4 (\bar\th F)^2 
+  \inv4  \bar\th^{\mu\nu} F_{\mu\nu} (F\bar\theta F \bar\theta) + \cO(F^4) \!\right) 
\!, \label{invrootaction-o3}
}
which is in agreement with the pure gauge sector of  \eq{potential-collect}.

\subsection{Building blocks: the matrix tensors \texorpdfstring{$H^{ab}$}{H**(ab)}}
\label{sec:blocks}

Now consider the general expansion of the effective matrix model action for fluctuating matrices.
Besides $\Theta^{ab} = -i[X^a,X^b]$, the following 
``matrix tensors'' \cite{Blaschke:2010qj} play an important role
\begin{align}
H^{ab} &= \inv{2}\aco{\co{X^a}{X^c}}{\co{X^b}{X_c}} 
\,, \nn\\
H &=  H^{ab} g_{ab} =  \co{X_c}{X^d}\co{X^{c}}{X_{d}}  = -\tr J^2 
\,, \label{eq:def-H}
\end{align}
where we use coordinates such that $g_{ab} = \d_{ab}$. 
On the {\moyal} vacuum $\R^4_\th \subset \R^D$, they reduce to
$H^{ab} \to \bar H^{ab}$ where
\begin{align}
\bar H^{\mu\nu} &= -\LNC^{-4}\bG^{\mu\nu} \,, && \bar H^{\mu i} = \bar H^{ij} = 0
\,, \nn\\
\bar H &=  \bar H^{ab} g_{ab} = -\LNC^{-4}\bG^{\mu\nu} g_{\mu\nu} 
\,. \label{H-explicit}
\end{align}
Using the following characteristic relation for 4-dimensional Moyal space \cite{Blaschke:2010qj}
\be
(\bG g\bG)^{\m\n} = -\inv2 \LNC^{4}\bar H \bG^{\m\n} - g^{\m\n}
= \inv{2}(\bar Gg)\bar G^{\mu\nu} - g^{\mu\nu} 
\,, \label{4D-id-2}
\ee
(which follows from \eq{char-4D})
as well as
\begin{align}
\Theta^{\mu\nu} &=  \bar\theta^{\mu\nu} +  \cF^{\mu\nu} 
= -\thb^{\m\r} \thb^{\n\s}(\bar\theta^{-1}_{\r\s} + F_{\r\s}) 
\,, \nn\\
\Theta^{\mu i} &=  \bar\theta^{\mu\nu} D_\nu\phi^i
\,, \label{XX-gauge-b}
\end{align}
this gives
\begin{align}
(\bar H^{ab} - \frac 12\bar H g^{ab}) [\bar X_a,\Phi][\bar X_b,\Psi]   
&= (\LNC^{-4}\bG^{\m\n} + \inv2 \bar H g^{\m\n})
g_{\mu\mu'} g_{\nu\nu'} \bar\theta^{\mu'\rho}\bar \theta^{\nu'\eta} \del_\rho \Phi \del_\eta \Psi  \nn\\
&= -\LNC^{-8} g^{\rho\eta}\del_\rho \Phi \del_\eta \Psi
\,. 
\end{align}
We hence define
\eqn{
\Box_g \Phi &:= [X_a,(H^{ab} - \inv2 H g^{ab}) [ X_b,\Phi]]  \,, \nn\\
\bar\Box_g \Phi &:= [\bar X_a,(\bar H^{ab} - \inv2 \bar H g^{ab}) [\bar X_b,\Phi]]   
= -\LNC^{-8} g^{\r\eta}\del_\r\del_\eta \Phi \,, \nn\\
\Box \Phi &:= [X^a,[X^b,\Phi]] g_{ab} \,, \nn\\
\bar\Box \Phi &:= [\bar X^a,[\bar X^b,\Phi]] g_{ab}  
=  -\LNC^{-4} \bar G^{\r\eta}\del_\r\del_\eta \Phi \,.
\label{Moyal-laplaceops} 
}
(It was shown in \cite{Blaschke:2010qj} that they reduce to the appropriate Laplace operators
in the semi-classical limit for general 
4-dimensional branes.) 
In particular, this yields
\be
(\bar H^{ab} - \frac 12\bar H g^{ab}) [\bar X_a,e^{i p x}][\bar X_b,e^{i q x}]   
= \LNC^{-8}g^{\r\eta} p_\r q_\eta e^{i p x}e^{i q x} \,.
\ee
Including fluctuations around $\R^4_\th\subset \R^D$ as in \eq{eq:general-X}, we have
{\allowdisplaybreaks
\begin{align}
H^{\mu\nu} &= -\frac 12 [\Theta^{\mu c},\Theta^{\nu d}]_+ g_{cd} 
= - \frac 12 [\bar\theta^{\mu\a} + \cF^{\mu \a},\bar\theta^{\nu \b} + \cF^{\nu\b}]_+ g_{\a\b} 
 - \frac 12 \bar\theta^{\mu\a}\bar\theta^{\nu \b}[D_\a\phi^i,D_\b\phi^j]_+ g_{ij} \nn\\*
&= -  \frac{\bG^{\mu\nu}}{\LNC^4}
- \frac{\bar\theta^{\mu\mu'}\bar\theta^{\nu\nu'}}{\LNC^4} \(\!
 \bG^{\a\b}(\bar\theta^{-1}_{\mu'\a} F_{\nu' \b} 
+ \bar\theta^{-1}_{\nu' \b} F_{\mu'\a} 
+\frac 12 [F_{\mu'\a},F_{\nu' \b}]_+ )
 + \inv2 \aco{D_{\mu'}\varphi^i}{D_{\nu'}\varphi_i} \!\) \!, \nn\\
H^{\mu i} &= -\frac 12 [\Theta^{\mu c},\Theta^{i d}]_+ g_{cd} 
 =  - \frac 12 [\Theta^{\mu k},\Theta^{i j}]_+ g_{kj} 
 -   \frac 12 [\Theta^{\mu \a},\Theta^{i \b}]_+ g_{\a\b}    \nn\\*
&= \frac 12 \LNC^{-6}\bar\theta^{\mu \mu'} \([D_{\mu'}\varphi^k,i[\varphi^i,\varphi^j]]_+ g_{kj} 
 - [(\bar\theta^{-1} + F)_{\mu'\a'},D_\k\varphi^i]_+\bG^{\a'\k}\)  \,, \nn\\
 H^{ij} &= -\frac 12 [\Theta^{i c},\Theta^{j d}]_+ g_{cd} 
 =  - \frac 12 [\Theta^{i k},\Theta^{jl}]_+ g_{kl} 
 -  \frac 12 [\Theta^{i \a},\Theta^{j \b}]_+ g_{\a\b}    \nn\\*
&=  \frac 12 \LNC^{-8}[[\varphi^i,\varphi^k],[\varphi^j,\varphi^l]]_+ g_{kl} 
 - \frac 12 \LNC^{-8}\bG^{\a\b} [D_{\a} \varphi^i,D_{\b} \varphi^j]_+   \,, \nn\\
H &= H^{\mu\nu} g_{\mu\nu} + H^{ij} g_{ij} \nn\\*
&= -  \LNC^{-4}\bG^{\mu\nu}g_{\mu\nu} - \LNC^{-8}
\( -2 \LNC^{4}\hat\theta^{\nu\b} F_{\nu \b} 
+ \bG^{\mu\nu} \bG^{\a\b}F_{\mu\a} F_{\nu \b} 
 + 2\bG^{\mu\nu} D_{\mu}\varphi^i D_{\nu}\varphi^j g_{ij}\) \nn\\*
&\quad + \LNC^{-8}[\varphi^i,\varphi^k][\varphi^j,\varphi^l] g_{kl} g_{ij} 
\,, \label{H-explicit-fluct}
\end{align}
where we defined
}
\begin{align}
D_\a \phi &= \del_\a\phi + i [A_\a,\phi] \,,  &
\textrm{and} \qquad
\hat\theta^{\mu\nu} &= (\bG g \bar\theta)^{\mu\nu} = - \hat\theta^{\nu\mu} 
\,.
\end{align}
Note that the expressions \eqref{H-explicit-fluct} are exact.

\subsection{Expansion of the effective potential}

Using the above results, we can expand the various terms in the effective potential
$V(X)$ in terms of gauge and scalar fields. 
The contributions quadratic in the fields $(A,\varphi)$ are 
{\allowdisplaybreaks
\begin{align}
H^{ab}H_{ab} &= H^{\mu\nu}H^{\m'\n'}g_{\m\m'}g_{\n\n'} + H^{ij}H^{i'j'}g_{ii'}g_{jj'}  + 2 H^{\mu i}H^{\m' i'}g_{\m\m'}g_{ii'} \nn\\*
&= \LNC^{-8} (\bG g \bG)^{\k\k'} g_{\k\k'} 
+ 4\LNC^{-12}\bG^{\a\b} (\bG g \bG)^{\mu\nu} F_{\mu\a}F_{\nu \b}   \nn\\*
&\quad  +  2\LNC^{-8} \hat\theta^{\mu'\b}\hat\theta^{\nu\a'} F_{\nu \b} F_{\mu'\a'} 
 + 4 \LNC^{-8} F_{\mu\a}(\frac 12 (\bar G g) \hat\th - \bar\th)^{\a\mu}   \nn\\*
&\quad + 4\LNC^{-12}(\bG g \bG)^{\a\b} D_{\a}\varphi^i D_{\b}\varphi^j g_{ij} 
+ \mbox{h.o.} \,, \nn\\
H^2 &=  \LNC^{-8}(\bG^{\mu\nu}g_{\mu\nu})^2 
+ 4 \LNC^{-12}(\bG g) \bG^{\mu\nu}  D_{\mu}\varphi^i D_{\nu}\varphi^j g_{ij}
+ 4 \LNC^{-8}(\hat\th^{\mu\nu} F_{\mu\nu})^2 \nn\\* 
&\quad -4 \LNC^{-8}(\bar G g)(\hat\th^{\mu\nu} F_{\mu\nu})
+ 2 \LNC^{-12}(\bG g) \bG^{\mu\nu} \bG^{\a\b} F_{\mu\a} F_{\nu \b}
+ \mbox{h.o.}  
\end{align}
The $\cO(A\vph\vph)$ terms are found to be
\begin{align}
H^{ab}H_{ab}|_{A\vph\vph} 
&= 8 \LNC^{-12}\, \hat\theta^{\rho\b} F_{\b\nu}\bar G^{\nu\eta}
 D_{\rho}\varphi^i D_{\eta}\varphi_i  \,, \nn\\
H^2_{A\vph\vph} &= -8 \LNC^{-12} \LNC^{4}\hat\theta^{\nu\b} F_{\nu \b} 
  \bG^{\mu\nu} D_{\mu}\varphi^i D_{\nu}\varphi^j g_{ij}  \,, \nn\\
(H^{ab}H_{ab} - \frac 12 H^2)|_{A\vph\vph} 
&= 8 \LNC^{-12}\, \Big(\hat\theta^{\rho\b} F_{\b\nu}\bar G^{\nu\eta} D_{\rho}\varphi^i D_{\eta}\varphi_i 
 + \frac 12 \hat\theta^{\a\b} F_{\a\b}\bar G^{\mu\nu} D_{\mu}\varphi^i D_{\nu}\varphi_i \Big)
, 
\end{align}
noting that
}
\be
(\bG g\bG g\bar\th)^{\m\n} = \inv2 (\bG g) \hat\th^{\m\n} - \bar\th^{\m\n}
\,.
\ee
We can rewrite $\hat\theta^{\mu\b}\hat\theta^{\nu\a}F_{\nu \b} F_{\mu\a}$
using the identities in Lemma \ref{lemma-1}, giving
{\allowdisplaybreaks
\begin{align}
H^{ab}H_{ab} - \frac 12 H^2 &= 
\LNC^{-8}(\bG g \bG - \frac 12 (\bG g) \bG)^{\k\k'} g_{\k\k'} 
+ 4\LNC^{-12}\bG^{\a\b} (\bG g \bG - \frac 14 (\bG g)\bG)^{\mu\nu} F_{\mu\a}F_{\nu \b}  \nn\\*
&\quad  + 4\LNC^{-12}(\bG g \bG - \frac 12  (\bG g) \bG)^{\a\b} D_{\a}\varphi^i D_{\b}\varphi^j g_{ij}
 -\LNC^{-8}(\hat\th^{\mu\nu} F_{\mu\nu})^2 \nn\\*
&\quad - 4 \LNC^{-8} F_{\mu\a}\bar\th^{\a\mu}
- 8\LNC^{-8}\textrm{Pfaff}(F_{\mu\nu}) \textrm{Pfaff}(\hat\theta^{\mu\nu}) \nn\\*
&\quad + 8 \LNC^{-12}\, \Big(\hat\theta^{\rho\b} F_{\b\nu}\bar G^{\nu\eta} D_{\rho}\varphi^i D_{\eta}\varphi_i 
 + \tfrac 12 \hat\theta^{\a\b} F_{\a\b}\bar G^{\mu\nu} D_{\mu}\varphi^i D_{\nu}\varphi_i \Big)
\! + \mbox{h.o.} \nn\\
&= -\LNC^{-8} g^{\k\k'} g_{\k\k'} 
 - 4\LNC^{-12} g^{\a\b} D_{\a}\varphi^i D_{\b}\varphi^j g_{ij}
 -\LNC^{-8}(\bar\th^{\mu\nu} F_{\mu\nu})^2 \nn\\*
&\quad + 4 \LNC^{-8} F_{\mu\a}\bar\th^{\mu\a}
- 8\LNC^{-12}\textrm{Pfaff}(F_{\mu\nu}) \nn\\*
&\quad +8 \LNC^{-12}\, \Big(\bar\theta^{\rho\b} F_{\b\nu} g^{\nu\eta} D_{\rho}\varphi^i D_{\eta}\varphi_i 
 + \frac 12 \bar\theta^{\a\b} F_{\a\b} g^{\mu\nu} D_{\mu}\varphi^i D_{\nu}\varphi_i \Big)
+ \mbox{h.o.} 
\end{align}
Here $\textrm{Pfaff}(F_{\m\n})$ is a purely topological surface term i.e. a total derivative.

\begin{lemma}
For anti-symmetric matrices $F_{\mu\nu}$ and $\hat\theta^{\mu\nu}$
defined as above, the following identities hold:
\label{lemma-1}
\begin{subequations}
\begin{align}
\frac 18(F\hat\theta) (F\hat\theta) - \frac 14 (F\hat\theta F \hat\theta) 
&= \textrm{Pfaff}(F_{\mu\nu}) \textrm{Pfaff}(\hat\theta^{\mu\nu})
\,, \label{FF-theta} \\
\bG^{\a\b} (\bG g \bG - \frac 14 (\bG g)\bG)^{\mu\nu} F_{\m\a}F_{\n\b} 
 &= \frac 14\LNC^4\((\hat\th^{\mu\nu} F_{\mu\nu})^2 - (\bar\th^{\mu\nu} F_{\mu\nu})^2 \)
\,, \label{strange-id} \\
 (\hat\theta^{\mu\a} F_{\a\b}\bar G^{\b\nu} - \bar\theta^{\mu\a} F_{\a\b} g^{\b\nu})
  \del_{\mu} \del_{\nu}  &= 
\frac 12 (\bar\theta^{\a\b} F_{\a\b} g^{\mu\nu} -  \hat\theta^{\a\b} F_{\a\b}\bar G^{\mu\nu}) 
\del_{\mu} \del_{\nu}
\,. \label{strange-id-2}
\end{align}
\end{subequations}
\end{lemma}
\begin{proof}
See \appref{app:proof-of-lemma-1}.
\end{proof}
This gives
\begin{align}
&\quad \frac 1{\sqrt{-H^{ab}H_{ab} + \frac 12 H^2}} \nn\\
&\sim \frac{\LNC^4}2 \Bigg(1+\LNC^{-4} g^{\a\b} D_{\a}\varphi^i D_{\b}\varphi_i
 +\frac 14 (\bar\th^{\mu\nu} F_{\mu\nu})^2  - F_{\mu\a}\bar\th^{\mu\a}
+ 2\LNC^{-4}\textrm{Pfaff}(F_{\mu\nu}) \nn\\
&\qquad -2 \LNC^{-12}\, \Big(\bar\theta^{\rho\b} F_{\b\nu} g^{\nu\eta} D_{\rho}\varphi^i D_{\eta}\varphi_i 
 + \frac 12 \bar\theta^{\a\b} F_{\a\b} g^{\mu\nu} D_{\mu}\varphi^i D_{\nu}\varphi_i \Big)
+ \cO(FFF)  + \textrm{h.o.}\Bigg)^{-1/2} \nn\\
&= \frac 12 \Bigg(\LNC^4-\frac 12 g^{\a\b} D_{\a}\varphi^i D_{\b}\varphi_i
 +\frac 14 \LNC^4(\bar\th^{\mu\nu} F_{\mu\nu})^2  +\frac 12 \LNC^4 F_{\mu\a}\bar\th^{\mu\a}
 - \textrm{Pfaff}(F_{\mu\nu}) \nn\\*
&\qquad  + \LNC^{-8}\, \Big(\bar\theta^{\rho\b} F_{\b\nu} g^{\nu\eta} D_{\rho}\varphi^i D_{\eta}\varphi_i 
 - \frac 14 \bar\theta^{\a\b} F_{\a\b} g^{\mu\nu} D_{\mu}\varphi^i D_{\nu}\varphi_i \Big) 
+ \cO(FFF) + \textrm{h.o.}\Bigg) 
\end{align}
consistent with \eq{invrootaction-o3}, from which we will take  the $\cO(A^3)$ terms for simplicity.
This yields
\begin{align}
\Tr V(X) &= -\frac 14\Tr \frac{L^4}{\sqrt{\frac 12 H^2-H^{ab}H_{ab} }} \nn\\
&= -\frac 1{8}\L^4 \int\!\! \frac{d^4 x}{(2\pi)^2}\sqrt{g} \Bigg(
  1- \frac{1}{2\LNC^4} g^{\a\b} D_{\a}\varphi^i D_{\b}\varphi_i
 +\frac {1}{4} \Big((\bar\th^{\mu\nu} F_{\mu\nu})^2  
 +  \bar\th^{\mu\nu} F_{\mu\nu} (F\bar\theta F \bar\theta)\Big) \nn\\*
&\qquad + \LNC^{-12}\, \Big(\bar\theta^{\rho\b} F_{\b\nu} g^{\nu\eta} D_{\rho}\varphi^i D_{\eta}\varphi_i 
 - \frac 14 \bar\theta^{\a\b} F_{\a\b} g^{\mu\nu} D_{\mu}\varphi^i D_{\nu}\varphi_i \Big) 
\, + \mbox{h.o.} \Bigg)  
\label{pot-matrix-action}
\end{align}
dropping surface terms. 
This agrees precisely with the induced action \eq{potential-collect}.
In particular, the $SO(D)$ symmetry is indeed preserved.
}

Note that the effective matrix model contains much more information than what was put in.
It predicts an infinite series of higher-order terms in $F$ and $\varphi$
proportional to $\L^4$. Since the above expression was uniquely determined
by very simple arguments in \secref{sec:potential}, this represents a strong prediction 
of the matrix-model framework,
based on the $SO(D)$ symmetry and the scaling law \eq{basic-scaling} combined with 
a very basic one-line loop computation.

It is interesting that these induced ``vacuum energy'' 
terms are distinct from the bare matrix model \eq{YM-action}.
This means that the physics of vacuum energy is different from GR, 
which may be very relevant to the cosmological constant problem.

\subsection{Curvature terms}

Having understood the single-commutator matrix model terms, we now 
consider contributions which contain double commutators, more precisely those which 
can be written such that there are two double commutators. 
It turns out that these lead to curvature terms from the gravity point of view, 
which usually diverge as $\L^2$. 
They correspond to gauge theory terms which have at least four derivatives 
such as  $\del\del\phi \del\del\phi$, of dimension $\geq 6$.
We first discuss the structure of such terms in the matrix model, and then compute their 
gauge theory content.

\subsubsection{On the structure of curvature terms}

We first want to understand what kind of terms in the matrix model 
can be written in terms of two double commutators\footnote{It is clear 
from $SO(D)$ invariance that there is no term
which has only one double commutator
$[X,[X,X]]$ except for $D=3$, which is not considered here.}. 
We consider only single-trace terms here, and restrict ourselves to the case of $\cO(X^6)$
and $\cO(X^{10})$ contributions. Eventually, a more systematic classification should be given.

While simple commutators $[X,X]$
correspond to constants $\bar\theta^{\mu\nu}$ or 
first-order derivatives $F_{\mu\nu}$ and $\del_\mu \phi^i$ (resp. $\del\phi\del\phi$), 
double commutators such as $[X,[X,X]]$ or $[[X,X],[X,X]]$ correspond to 
terms with second-order derivatives. 
Note that all terms which contain $[V(X),V(X)]$
can be reduced to terms of the structure $V(X) [[X,X],[X,X]]$ (up to higher commutators),
which using the Jacobi identity can be rewritten as
\be
V(X) [[X,X],[X,X]] = - V(X) [X,[X,[X,X]]]  - V(X) [X,[X,[X,X]]] 
\,.  
\ee
Here $V(X)$ is a potential term, i.e. involves only single commutators.
Under the trace, these can be rewritten as $\Tr [X,V(X)] [X,[X,X]]$.
It remains to characterize the most general term of the structure
$\Tr V_1(X) [X,V_2(X)] [X,V_3(X)]$. Using again the Leibnitz rule, this can be reduced to
\be
\Tr \cL_{\textrm{curv}}[X]
= \Tr V(X) [X,[X,X]] [X,[X,X]] 
\label{curv-MM-form}
\ee
up to additional or higher-order commutators. 
Here the indices can be contracted in any possible way.

We note the following relation
\be
\Theta^{dc} [X^a,\Theta^{bc'}] g_{cc'} = [X^a,\Theta^{dc}\Theta^{bc'}g_{cc'}] 
 - \Theta^{bc} [X^a,\Theta^{dc'}] g_{cc'}\quad + \textrm{h.o.},
\label{matrix-PI}
\ee
which implies 
\eqn{
\,[X_a,H^{ab}] &= \inv{2} \(\aco{\Box X_c}{ [X^b,X^c]}  
+ \frac 12 [X^b,H]\)  \nn\\
&\sim \Box X_c[X^b,X^c] + \frac 14 [X^b,H] \quad + \textrm{h.o.} 
\label{H-cons} }

\paragraph{$\cO(X^{6})$ curvature terms.}
\label{sec:O10curv}

It is easy to see \cite{Blaschke:2010rg} that there are only two independent terms of order $\cO(X^6)$, given by
\begin{align}
S_{6,A} &= \Tr \Box X^a \Box X_a \,, \nn\\
S_{6,B} &= \Tr [X^c,[X^a,X^b]] [X_c,[X_a,X_b]]
\,. 
\end{align}
Notice that there is no potential term at order $\cO(X^6)$.
From the geometrical point of view they contain curvature 
contributions. Although they are not induced in the effective action, they will occur
as part of the $\cO(X^{10})$ terms.

\paragraph{$\cO(X^{10})$ curvature terms.}

At order $X^{10}$, we have to characterize the most general term of structure 
$\Tr J J [X,J] [X,J]$. There are several cases:
\bit
\item 
Assume that the indices of the $X$ are contracted as in
$\Tr J J [X^a,J] [X_a,J]$. 
Now consider the possible contractions of the indices of the $J$'s.
The indices of the $J$'s can form two disjoint loops, or a single loop. 

Consider first the case of {\em two $J$ loops}. 
If the  indices  of the internal $J$ (resp. the external $J$) are contracted among themselves, 
this gives $\Tr H [X^a,J] [X_a,J]$ i.e. $\Tr H S_{6,\,\textrm{curv}}$.
Otherwise, each $J$ loop can be written using \eq{matrix-PI} as 
$[X^a,H]$, so that the action is $\Tr [X^a,H] [X_a,H]$. 

Now consider the case of a single $J$ loop. Using again \eq{matrix-PI}, these can be
reduced to the form $\Tr J [X^a,J] J [X_a,J]$ and  $\Tr [X^a,J^2] [X_a,J^2]$
up to higher-order terms. These turn out to be independent.

\item
Now assume that the indices of the $X$ are not contracted among themselves.
Then consider the two triple commutators $[X,\Theta]$, which altogether 
have 6 indices. Since there are only two external $\Theta$, 
at least two of the 
indices of the two double commutators must be contracted among themselves.
Therefore either two indices of the same $[X,[X,X]]$ are contracted among themselves
which gives $\Box X^a$, or otherwise they must be contracted as in
 $[X^a,\Theta_{cd}][X_b,\Theta_{ae}] \Theta \Theta$. 
The first case will be discussed below. In the second case, 
the double commutator has 4 remaining indices 
$b,c,d,e$. It is antisymmetric in $(cd)$, and 
we can assume that it is symmetric in $(be)$ because the anti-symmetric component in $(be)$
reduces to $[X^a,\Theta_{cd}][X_a,\Theta_{be}]\Theta \Theta$ using the Jacobi identity, 
which has been covered above. Hence there is no (independent) way to contract them with the 
4 free indices of $\Theta\Theta$. Therefore two of these must be contracted with $g^{ab}$. This
leads to
 $[X^a,\Theta_{cd}][X_b,J^c_{a}] H^{db} \sim [X^a,\Theta_{cd}][X^c,\Theta_{ab}] H^{db}$ or to  
 $[X^a,J^b_{c}][X_b,\Theta_{ae}] H^{ce}$ (which is the same) or to 
$[X^a,\Theta_{cd}]\Box X_a H^{cd} \sim 0$
(apart from the case $S_6 H$ which has been covered before).
Hence the only new term is
\eqn{
& 2[X_a,\Theta_{cd}] H^{db} [X_b,\Theta^{ac}]  
=  - [X_d,\Theta_{ac}] H^{db} [X_b,\Theta^{ac}] 
\,. }
It remains to classify the terms of the form $\Box X^a [X_b,\Theta_{cd}] \Theta \Theta$.
Using the Jacobi identity, we can assume that 
the index $a$ is contracted either with an external $\Theta$, or with $\Theta_{cd}$.
In the first case we get  $\Theta_{ae} \Box X^a [X_b,\Theta_{cd}] \Theta$,
which gives either $H_{ac} \Box X^a \Box X^c$ or 
\eqn{
\Theta_{ae} \Box X^a [X_b,\Theta^{ef}] \Theta^{fb} 
&\sim 
\Theta_{ae} \Box X^a [X_b,H^{eb}]  - H_{fa} \Box X^a [X_b,\Theta^{fb}] \nn\\
 &= \Theta_{ae} \Box X^a (\Box X_c[X^e,X^c] + \frac 14 [X^e,H]) - i H_{af} \Box X^a \Box X^f \nn\\
 &= \frac 14 \Theta_{ae} \Box X^a [X^e,H] 
\,. }
In the second case, we can assume (using the Jacobi identity) that the term has the form
$\Box X^a [X_b,\Theta_{ad}] \Theta \Theta$. This leads either to 
$\Box X^a \Box X_a H \in S_6 H$, or using \eq{H-cons} to 
\eqn{
\Box X^a [X_b,\Theta_{ad}] H^{bd}
&=  - \Box X^a [X_d,\Theta_{ba}] H^{bd}  \nn\\
&= - \Box X^a [X_d,\Theta_{ba}H^{bd}] + \Box X^a \Theta_{ba}[X_d,H^{bd}]  \nn\\
&= - \Box X^a [X_d,(J^3)^d_{a}] 
 - \frac 14 \Box X^a \Theta_{ba} [X^b,H] - \Box X^a H_{ae} \Box X^e 
\,. }
\eit
Hence, we have the following complete list of $\cO(X^{10})$ curvature
terms:
\eqn{
S_{10,\,\textrm{curv}} =
&\Tr \Big(H S_{6} + c_1 [X^a,H] [X_a,H]+ c_2\Theta_{ae} \Box X^a [X^e,H]  \nn\\ 
 &\qquad + c_3 \tr J [X^a,J] J [X_a,J] + c_4 [X^a,H^{cd}] [X_a,H_{cd}] \nn\\
& \qquad + c_5 [X_d,\Theta_{ac}] H^{db} [X_b,\Theta^{ac}]
 + c_6 \Box X^a [X_d,(J^3)^d_{a}] 
\Big) , }
which turn out to be independent.
The last four can be replaced by 
\bea
\Tr \Big(c_3 J [X^a,J] J [X_a,J] + c_4 H^{cd} \Box H_{cd}
 + c_5 \Theta_{ac} \Box_g \Theta^{ac}
 + c_6 \Box X^a \Box_g X_a \Big) .
\eea
There are also boundary terms which vanish under the trace, which will be 
supplemented below.

\subsubsection{Expansion of \texorpdfstring{$\cO(X^6)$}{O(X**6)} curvature terms}

As a warm-up, consider the two independent $\cO(X^6)$ terms to quadratic order
\begin{align}
S_{6,A} &= \Tr \Box X^a \Box X_a \nn\\
 &\sim \int\!\!\frac{d^4x}{(2\pi)^2} \sqrt{\bG} \Big(\frac 12 \LNC^{-4} \bar\Box F_{\mu\n}  F_{\a\b} \bG^{\a\mu} \bG^{\b\n} 
+ \LNC^{4}\bar\Box \phi^i\bar\Box \phi_i  + \textrm{h.o.} \Big) \,, \nn\\
S_{6,B} &= \Tr [X^c,[X^a,X^b]] [X_c,[X_a,X_b]] \nn\\
 &\sim \int\!\!\frac{d^4x}{(2\pi)^2} \sqrt{\bG} \Big(\LNC^{-4}  \bar\Box  F_{\m\n}  F_{\a\b} \bG^{\a\m}\bG^{\b\n}
 + 2 \LNC^{4}\bar\Box \phi^i \bar\Box \phi_i + \textrm{h.o.} \Big)\,.
\label{O6-quadratic} 
\end{align} 
Notice that the ``intrinsic combination''  
\begin{align}
\cL_{6,\,\textrm{cubic}}[X] &= \Box X^a \Box X_a - \frac 12 [X^c,[X^a,X^b]] [X_c,[X_a,X_b]]
\label{L6-cubic}
\end{align}
has no quadratic contributions, and turns out to gives a tensorial term involving the 
Riemann tensor \cite{Blaschke:2010rg}.
To obtain the above form, we need the following identities
\begin{lemma}
For an Abelian field strength tensor $F_{\mu\nu}$ and $\hat\theta^{\mu\nu}$ etc.
defined as before, the following identities hold:
\label{lemma-2}
\begin{subequations}
\begin{align}
 g^{\mu\nu} \del_{\mu} F_{\nu\g} \bG^{\rho\b} \del_\rho F_{\b\m'}\bG^{\g\m'}  
 &= - \frac 12 (g^{\mu\nu} \del_\nu \del_{\mu}) F_{\rho\g}  F_{\b\m'} \bG^{\rho\b} \bG^{\g\m'} \quad + \del() 
 \,, \label{FFbox-id-1}\\
\bG^{\mu\nu}\del_{\mu} F_{\nu\g} \bG^{\rho\b}\del_\rho F_{\b\mu'}  g^{\g\mu'}
 &= \tinv2 (g^{\mu\nu} \del_\nu \del_{\mu}) F_{\rho\g}  F_{\b\m'} \bG^{\rho\b} \bG^{\g\m'} 
 - (\bG^{\mu\nu} \del_\nu \del_{\mu}) F_{\rho\g}  F_{\b\m'} g^{\rho\b} \bG^{\g\m'}  + \del()
, \label{FFbox-id-2}
\end{align}
\end{subequations}
\end{lemma}
where $\del()$ denotes surface terms.
The proof is given in \appref{app:proof-of-lemma-2}, based on the Bianchi identity.

For completeness, we give here the full expansion of these terms around $\R^4_\theta$.
Consider first
\begin{align}
\Box \phi^k &=  \co{X^\mu}{\co{X_\mu}{\phi^k}} + \co{\phi^i}{\co{\phi_i}{\phi^k}}  \nn\\
&= -\LNC^{-4}\bG^{\mu\nu} D_\mu D_{\nu'} \phi^k  +\co{\phi^i}{\co{\phi_i}{\phi^k}} 
\,, 
\end{align}
and similarly
\begin{align}
\Box X^\eta &=  \co{X_\mu}{\co{X^\mu}{X^\eta}} + \co{\phi_i}{\co{\phi^i}{X^\eta}}  \nn\\
&= \LNC^{-4}\bar\th^{\eta\b} \bG^{\a\rho} D_\rho F_{\a\b} 
- i\bar\theta^{\eta\b} \co{\phi_i}{D_\b\phi^i} 
\,.
\end{align}
This gives
\begin{align}
S_{6,A} &= \Tr \Box X^a \Box X_a  \nn\\ 
&= \int\!\!\frac{d^4x}{(2\pi)^2} \sqrt{\bG}\Bigg(\LNC^{-8}\bG^{\b\b'} \bG^{\a\rho} 
 \bar G^{\a'\rho'} D_\rho F_{\a\b}  D_{\rho'} F_{\a'\b'} 
+ \LNC^{-4}\bG^{\mu\nu} D_\mu D_{\nu'} \phi^k\bar G^{\mu'\nu'} D_{\mu'} D_{\nu'} \phi_k \nn\\
&\quad \qquad -2i \LNC^{-4}\bG^{\b\b'} \bG^{\a\rho} D_\rho F_{\a\b} \co{\phi_i}{D_{\b'}\phi^i} 
 - 2\bG^{\mu\nu} D_\mu D_{\nu'} \phi^k\co{\phi^i}{\co{\phi_i}{\phi_k}}  \nn\\  
&\quad \qquad -\bG^{\b\b'} \co{\phi_i}{D_\b\phi^i} \co{\phi_j}{D_{\b'}\phi^j} 
+\co{\phi^i}{\co{\phi_i}{\phi^k}}  \co{\phi^j}{\co{\phi_j}{\phi_k}}  
 \Bigg) ,
\end{align}
and a similar computation leads to
\begin{align}
S_{6,B} &= \Tr [X^c,[X^a,X^b]] [X_c,[X_a,X_b]] \nn\\
&= \int\!\!\frac{d^4x}{(2\pi)^2} \frac{\sqrt{\bG}}{\LNC^{4}} \bigg(
\LNC^{-8}\bG^{\a\m}\bG^{\b\n} \bG^{\g\s} D_\s F_{\m\n} D_\g F_{\a\b} 
 + 2 \LNC^{-4} \bG^{\s\l} \bG^{\e\varphi} D_\e D_\s \phi^i D_\varphi D_\l \phi_i  \nn\\
&\quad\qquad -  \bG^{\a\r} D_\a \co{\phi^i}{\phi^j} D_\r \co{\phi_i}{\phi_j}
 - \LNC^{-4} \bG^{\r\l} \bG^{\s\t} \co{\phi^i}{F_{\r\s}}\co{\phi_i}{F_{\l\t}}  \nn\\
&\quad\qquad - 2 \bG^{\s\t} \co{\phi^i}{D_\s\phi^j}\co{\phi_i}{D_\t\phi_j}
\bigg) .
\end{align}
It is remarkable that this contains no explicit $\bar\theta^{\mu\nu}$.

\subsubsection{Expansion of \texorpdfstring{$\cO(X^{10})$}{O(X**10)} terms}
\label{sec:O-10-expansion}

We need to systematically compute the field theory content of 
all the above $\cO(X^{10})$ terms, which  start with $\cO(p^4 (\vp,A)^2)$.
The results are as follows (see \appref{app:S-10-curv}):
{\allowdisplaybreaks
\begin{align}
\cL_{10,A} &= [X_f,H^{ab}][X^f,H_{ab}] \nn\\
&=  \LNC^{-8}\Big(\bar H ( \frac 12 F_{\m\n} \bar\Box F_{\r\s} \bG^{\m\r}\bG^{\n\s}  
  + \LNC^4\bar\Box\varphi^k \bar\Box\varphi_k) \nn\\*
&\qquad\quad  -\frac 12 \big(3(\hat\th F)\bar\Box(\hat\th F) - (\bar\th F)\bar\Box(\bar\th F)\big) 
 + 2\LNC^4\bar\Box\varphi^k \bar\Box_g\varphi_k \Big)  +\textrm{h.o.} \,,  \nn\\
\cL_{10,B} &= - (H^{cd} - \frac 12 H g^{cd}) [X_c, \Theta^{ab}] [X_d, \Theta_{ab}]  \nn\\
&=  \LNC^{-8}  \Big(\bG^{\r\s}\bG^{\e\t} F_{\r\e} \bar\Box_g F_{\s\t} 
 + 2\LNC^{4}\bar\Box_g\varphi^i \bar\Box\varphi_i \Big) + \textrm{h.o.} \,,  \nn\\
\cL_{10,C} &= \Box X^a \Box_g X_a \nn\\*
&=  \LNC^{-8} \Big(\bar\Box_g F_{\rho\g}  F_{\b\m} \bG^{\rho\b} \bG^{\g\m} 
 + \frac 12(\hat\theta F) \bar\Box (\hat\theta F) \nn\\*
 &\quad \qquad -  \frac 14(\bar\theta F)\bar\Box(\bar\theta F) 
 + \LNC^4 \bar\Box \vp^i \bar\Box_g \vp_i \Big)  + \del()  + \textrm{h.o.} \,, \nn\\ 
\cL_{10,D} &=  \tr J [X^a,J] J [X_a,J] 
= -\frac 12 \LNC^{-8}  (\hat\theta F) \bar\Box(\hat\theta F ) 
    + \textrm{h.o.} \,, \nn\\
\cL_{10,E} &= [X_f,H][X^f,H] =  -4\LNC^{-8} (\hat\theta F) \bar\Box (\hat\theta F) 
      + \textrm{h.o.} \,, \nn\\
\cL_{10,F} &= \Box X^a \Theta_{ab} [X^b,H] = -\LNC^{-8}(\hat\theta F) \bar\Box(\hat\theta F) 
 + \textrm{h.o.}
\label{S-10-curv}
\end{align}
Recall that $\bar\Box$ and $\bar \Box_g$ have different powers of $\LNC$, see \eqref{Moyal-laplaceops}.
We also note the following contributions from the $\cO(X^{10})$ ``boundary terms'':
}
\begin{align}
\Box H^2 &= -4\LNC^{-8} (\bar G g) \Box (\hat\theta^{\mu\nu}F_{\mu\nu})  + \Box (\textrm{h.o.})\,,  \nn\\
\Box (H^{ab}H_{ab} - \frac 12 H^2) &= 4\LNC^{-8} \Box (\bar\theta^{\mu\nu}F_{\mu\nu})  +  \Box (\textrm{h.o.})\,,   \nn\\
\Box_g H &= 2 \LNC^{-4}  \Box_g (\hat\theta^{\mu\nu}F_{\mu\nu})  +  \Box_g (\textrm{h.o.})\,.
\label{boundary-linear-terms}
\end{align}

Remarkably, the following terms have no quadratic contributions:
\begin{align}
\cL_{10,\,\textrm{cubic}}[X] &= q_1\Big(\frac 12 [X_f,H^{ab}][X^f,H_{ab}] 
+ (H^{cd} - \frac 12 H g^{cd}) [X_c, \Theta^{ab}] [X_d, \Theta_{ab}]
 + \Box X^a \Box_g X^a \Big) \nn\\
& \qquad + q_2 \Big( 2\tr J [X^a,J] J [X_a,J] -  \Box X^a \Theta_{ab} [X^b,H]  \Big) \nn\\
 & \qquad + q_3  \Big( [X_f,H][X^f,H] -4\Box X^a \Theta_{ab} [X^b,H]  \Big) \,.
\label{L10-cubic}
\end{align}
It would be desirable to include these ``cubic'' terms into the above computations,
which is however beyond the scope of this paper.

\subsubsection{\texorpdfstring{$\cO(X^{14})$}{O(X**14)} term.}

A complete analysis for the $\cO(X^{14})$ curvature terms is beyond the scope of this paper. 
We only consider the following term which is clearly generated in \eq{2-field-collect-curv2}:
\begin{align}
& \Box_g X^a \Box_g X^a = \nn\\
&= \LNC^{-20} \Bigg( \bG^{\a\a'} g^{\eta\b} \del_\eta F_{\b\a} g^{\eta'\b'} \del_{\eta'} F_{\b\a'} 
+ \(\!\big(\tfrac{(\bG g)^2}4  -1\big) \bG^{\a\a'} - \tfrac{(\bG g)}2  g^{\a\a'}\) \bG^{\eta\b} \del_\eta F_{\b\a} \bG^{\eta'\b'} \del_{\eta'} F_{\b'\a'} \nn\\
&\quad -2 \bG^{\a\a'} g^{\eta\b} \del_\eta F_{\b\a}  \bG^{\eta'\b'}\del_{\eta'} F_{\b'\a'} 
 +\frac{\LNC^4}{4} (\bar G g \bG)^{\mu\nu} \del_{\mu} (\hat\th F) \del_{\nu} (\hat\th F) \nn\\
&\quad  + \LNC^{4}  \del_{\nu'} (\hat\th F)
(-g^{\eta\b} \del_\eta F_{\b\a} \hat\th^{\nu'\a} + \bG^{\eta\b} \del_\eta F_{\b\a} (\frac 12 (\bG g) \hat\th^{\n'\a} -\bar\theta^{\n'\a})) \Bigg) 
+ \bar\Box_g \phi^i \bar\Box_g \phi_i 
\end{align}
using 
\begin{align}
\Box_g X^\mu &= [X_c,[(H^{cb} - \frac 12 H g^{cb}) [X_b,X^\mu]] 
= [X_\g,[(H^{\g\b} - \frac 12 H g^{\g\b}) [X_\b,X^\m]] + \cO(\vp,A)^2  \nn\\
& = \LNC^{-8} g^{\rho\b} \del_\rho F_{\b\m'} \bar\th^{\m\m'}
 + \LNC^{-8} \hat\theta^{\a\m} \bG^{\eta\mu'} \del_\eta F_{\mu'\a} 
+\frac 12\LNC^{-8} \bar G^{\m\nu'} \del_{\nu'} (\hat\th^{\a\eta} F_{\a\eta}) + \textrm{h.o.} \,,
\label{Box-X}
\end{align} 
and $\bG g\bG g\bG = (\frac 14 (\bG g)^2 -1) \bG - \frac 12 (\bG g) g$.
Note that for $\bG = g$, the YM-type terms cancel, as in \eq{2-field-collect-curv2}.

\section{Effective matrix model including curvature contributions}

We  now look for a generalization of the matrix model action \eq{vacuuum-energy-matrix} 
involving curvature terms as above, such that 
the dimension 6 operators in \eq{2-field-collect-curv1} proportional to $\L^2 = L^2\LNC^4$
 are reproduced. 
In order to guess the appropriate form, observe first 
that the curvature terms vanish on a flat background $\R^4_\theta$. 
Assuming analyticity in the $X^a$ near $\R^4_\theta$, 
this suggests the following form for the effective matrix model:
\be
\Gamma_L[X] =  -\frac 14 \Tr \Bigg(\frac{L^4}{\sqrt{-\tr J^4 + \frac 12 (\tr J^2)^2 
 + \frac 1{L^2} \cL_{10,\,\textrm{curv}}[X] + \ldots }}\Bigg) 
. \label{full-action-MM}
\ee
This should be expanded on $\R^4_\theta$ up to the required order in $A_\mu, \varphi^i$ and compared with the 
induced gauge theory action. 
Recalling \eq{char-4D} we denote 
\eqn{
4\LNC^{-8}[X] := - \tr J^4 + \frac 12 (\tr J^2)^2  \,\, &\sim \,\, 4\LNC^{-8}(x) 
 = 4\LNC^{-8} - 4\LNC^{-4} \bar\theta^{\mu\nu} F_{\mu\nu} + \ldots 
\,. }
Then
\begin{align}
\Gamma_L[X] 
&= -\frac 14 \Tr L^4 
\Big( 4 \LNC^{-8}[X] + \frac 1{L^2} \cL_{10,\,\textrm{curv}} + \ldots \Big)^{-1/2}   \nn\\
&=  -\frac 18 \Tr
\Big( L^4\LNC^{4}[X] -\frac 18  L^2  \LNC^{12}[X] \cL_{10,\,\textrm{curv}}[X] + \ldots \Big) \nn\\
&=  -\frac 18 \frac{\LNC^4}{(2\pi)^2}\intg
\Big( L^4\LNC^{4}(x) -\frac 18  L^2  \LNC^{12}(x) \cL_{10,\,\textrm{curv}}(x) + \ldots  \Big) 
. \label{curv-action-expand}
\end{align}
Observe that the scaling with $\L^2 = L^2\LNC^4$ necessarily corresponds to a term $\cL_{10,\,\textrm{curv}}[X]$
of order $X^{10}$.
As shown above, the most general order $10$ term in the matrix model has the form
\eqn{
\cL_{10,\,\textrm{curv}} &= c_1 [X_f,H^{ab}][X^f,H_{ab}]
 + c_2  \Box X^a \Box_g X^a 
 + c_3 [X^f,H][X_f,H] \nn\\
 &\quad 
 + \cL_{10,\,\textrm{boundary}} + \cL_{10,\,\textrm{cubic}} 
 + d\, H \cL_{6,A} + H \cL_{6,\,\textrm{cubic}}
\,, }
where $\cL_{10,\,\textrm{cubic}}$ are given in \eq{L10-cubic},
and $\cL_{10,\,\textrm{boundary}}$ is a combination of the total derivative terms listed in \eq{boundary-linear-terms}.
To determine the  coefficients, 
we use the expansion of this action to quadratic order in the fields as given 
in section \ref{sec:O-10-expansion},
and match it with the induced gauge theory action proportional 
to $\L^2$. 
Comparing with \eq{2-field-collect-curv1}, we indeed obtain terms with the desired structure. 
Note that the total derivatives in $\cL_{10,\,\textrm{boundary}}$ may lead to non-trivial 
quadratic contributions, due to the multiplication with $\LNC^{-12}(X)$.
Comparing their contributions \eq{boundary-linear-terms} with
the induced action, we see that $\Box (H^{ab}H_{ab} - \frac 12 H^2)$ leads to $(\bar\theta F)\Box (\bar\theta F) $
which does indeed occur, while the other boundary terms are not induced. Therefore
we set 
\be
\cL_{10,\,\textrm{boundary}} = b\, \Box \LNC^{-8}[X] = - 4b\, \LNC^{-8}\bar\Box (\bar\theta^{\mu\nu}F_{\mu\nu}) 
+ \textrm{h.o.} \,, 
\ee
so that
\eqn{
\Tr \LNC^{12}(X) \cL_{10,\,\textrm{boundary}}(X) &= 
  -4 b\, \Tr (\LNC^{12} + \frac 32 \LNC^{16}(\bar\theta^{\mu\nu}F_{\mu\nu}) + \ldots)
 \LNC^{-8}\bar\Box (\bar\theta^{\mu\nu}F_{\mu\nu})  + \textrm{h.o.} \nn\\
&= -6b\, \LNC^8 \Tr (\bar\theta F)\bar\Box (\bar\theta F)  + \textrm{h.o.}
}
Then the quadratic gauge theory action resulting from \eq{curv-action-expand} is given by
\begin{align}
\Gamma_L[X] 
&= \frac 1{64}\L^2\frac{\LNC^4}{(2\pi)^2}\intg
\Big((\bar\th F)\bar\Box(\bar\th F)(\frac 12 c_1 - \frac 14 c_2 -6b)
 + \LNC^4\bar\Box\varphi^k \bar\Box_g\varphi_k (2 c_1  + c_2) \nn\\
 & \qquad  + c_2 \bG^{\r\s}\bG^{\e\t} F_{\r\e} \bar\Box_g F_{\s\t} 
 +(-\frac 32 c_1 + \frac 12 c_2 -4 c_3) (\hat\th F)\bar\Box(\hat\th F) \nn\\
& \qquad + (\frac{c_1} 2+d)\bar H \big( \frac 12 F_{\m\n} \bar\Box F_{\r\s} \bG^{\m\r}\bG^{\n\s}  
  + \LNC^4\bar\Box\varphi^k \bar\Box\varphi_k\big)\Big) 
 , \label{effective-matrix-gauge}
\end{align}
which must coincide with $\Gamma_{\L^2}((A,\varphi)^2,p^4)$ given in 
\eq{2-field-collect-curv1}.
This yields 
\eqn{
\frac 12 c_1 - \frac 14 c_2 -6b &= \frac 13 \tr\one \,, &
2 c_1  + c_2 &= -8 \tr\one \,,  &
c_2  &= -\frac{11}3 \tr\one \,, \nn\\
-\frac 32 c_1 + \frac 12 c_2 -4 c_3 &= 0 \,, & d &= -\frac{c_1}{2}
\,, }
which has the particular solution 
\be
 c_1  = -\frac{13}6 \tr\one  \,,  \qquad c_2 = -\frac{11}3 \tr\one \,, 
\qquad  b = -\frac 1{12} \tr\one  \,,  \qquad c_3 = \frac{17}{48}\tr\one \,, 
\qquad d= \frac{13}{12} \tr\one \,, 
\ee
plus an undetermined contribution of $\cL_{10,\,\textrm{cubic}}$ and $H \cL_{6,\,\textrm{cubic}}$.

We emphasize again the non-trivial analytic structure of 
the effective matrix model \eq{full-action-MM}, which is {\em not}
adequately described by a power series in $X^a$. 
This is to be expected, since the bare matrix model \eq{YM-action} describes a very rich
spectrum of backgrounds with various dimensions. 

We have thus successfully reproduced the induced gauge theory action 
\eq{2-field-collect-curv1} proportional
to $\L^2$, using the above generalized matrix model \eq{full-action-MM}. Even though the quadratic terms are
insufficient to fully determine the action and there are less non-trivial checks compared with 
the potential term, it is non-trivial that the correct terms are 
indeed reproduced. 
It would of course be desirable to develop more efficient methods to compute the effective generalized
matrix model.
A more detailed comparison and confirmation should eventually be given, 
allowing to determine also the contributions due to $\cL_{\textrm{cubic}}$. These are in fact the 
most interesting terms form the point of view of emergent gravity, as we recall below.

\paragraph{Geometric action.}

Finally, the geometrical meaning of the action \eq{full-action-MM}
is obtained  by expanding it not around $\R^4_\theta$ but around a generic 4-dimensional
NC brane $\cM^4 \subset \R^D$ generated by $X^\mu = \bar X^\mu + \cA^\mu$, cf. \cite{Steinacker:2010rh,Steinacker:2008ri}.
Using \eq{char-4D} and the semi-classical relation
\eqn{
\Tr \cL(X) &\sim \int\!\!\frac{d^4 x}{(2\pi)^2} \sqrt{G} \LNC(x)^4 \cL(x) 
\,, }
which holds on generic 4-dimensional branes $\cM^4$,
we can write  this action in the semi-classical limit as
\begin{align}
\Gamma_L[X] 
&\sim -\frac 14 \frac{1}{(2\pi)^2}\intG L^4 \LNC^4(x)
\Big( 4 \LNC^{-8}(x) + \frac 1{L^2} \cL_{10,\,\textrm{curv}} + \ldots \Big)^{-1/2}   \nn\\
&=  -\frac 18 \frac{1}{(2\pi)^2}\intG
\Big(\L^4(x) -\frac 18  \L^2(x)  \LNC^{12}(x) \cL_{10,\,\textrm{curv}} + \ldots \Big) 
. \label{effective-matrix}
\end{align}
Here we define the effective cutoff on a curved brane as
\be
\L(x) := L \LNC^2(x), 
\ee
in analogy to \eqnref{eq:def-cutoff}.
Now we recall the results of \cite{Blaschke:2010rg,Blaschke:2010qj} 
for the special case $G_{\mu\nu} = g_{\mu\nu}$,
which arises for (anti-)selfdual $\theta^{\mu\nu}$ 
(resp. an almost-K\"ahler structure) on $\cM^4$. 
Generalized to the present case, we note that $\cL_{6,\,\textrm{cubic}}$ and 
the first term in $\cL_{10,\,\textrm{cubic}}$ \eq{L10-cubic} 
\begin{align}
\cL_{10,\,\textrm{cubic}} &\ni \frac 12 [X_f,H^{ab}][X^f,H_{ab}] 
+ (H^{cd} - \frac 12 H g^{cd}) [X_c, \Theta^{ab}] [X_d, \Theta_{ab}]
 + \Box X^a \Box_g X^a \nn\\
&\sim (-\frac 12 H^{ab} \Box H_{ab} - \Box X^a \Box_g X^a)
 + (- \Theta^{ab} \Box_g \Theta_{ab}+2\Box X^a \Box_g X^a )
\end{align}
(up to boundary terms)
  incorporate the Riemannian curvature. Indeed
the first term $\frac 12 H^{ab} \Box H_{ab} + \Box X^a \Box_g X^a$
has been shown in \cite{Blaschke:2010rg,Blaschke:2010qj} to give essentially the curvature scalar $R$,
up to contributions from the dilaton\footnote{The name dilaton is chosen because it couples to the
curvature. 
It has a specific geometrical meaning in terms of $\LNC^4(x)$ here.} $\bar\L_{\rm NC}^4 e^{-\sigma}:= \LNC^4(x)$ and boundary terms.
Similarly, 
\bea
- \Theta^{ab} \Box_g \Theta_{ab}+2\Box X^a \Box_g X^a 
&\stackrel{g=G}{\sim}& 2 \LNC^{-4} \cL_{6,{\rm cubic}} \nn\\
&\sim& \bar\L_{\rm NC}^8 e^{-2\sigma}\theta^{\mu\rho} \theta^{\eta\a} R_{\mu \rho\eta\a} - 4 \bar\L_{\rm NC}^4 e^{-\sigma} R
\quad (+ {\rm dilaton} + \del()) \nn
\eea
has only cubic or higher-order contributions in the gauge theory point of view. 
Extrapolating these results to the present case (which involves different scalar factors),
we can anticipate 
\begin{align}
\Tr\, \LNC^{12}[X]H\cL_{6,\,\textrm{cubic}} 
 &\sim  - \int\!\!\frac{d^4x}{(2\pi)^2}\sqrt{g}\, \L(x)^2 \bigg(
\frac 12\bar\L_{\rm NC}^4 e^{-\s}\theta^{\mu\rho} \theta^{\eta\a} R_{\mu \rho\eta\a} - 2  R 
+ c \del^{\mu}\s \del_\mu\s \!\bigg) , \nn\\
\Tr\, \LNC^{12}[X]\cL_{10,\,\textrm{cubic}} 
&\sim \int\!\!\frac{d^4x}{(2\pi)^2}\sqrt{g}\, \L(x)^2 \bigg(R  
+ (\bar\L_{\rm NC}^4 e^{-\s}\theta^{\mu\rho} \theta^{\eta\a} R_{\mu \rho\eta\a} - 4 R )
+ c' \del^{\mu}\s \del_\mu\s \!\bigg) ,
\label{S6-S10-geom} 
\end{align}
for $G_{\mu\nu} = g_{\mu\nu}$,
where the overall coefficients and 
the scalar field contributions $c$, $c'$ are not determined here.
These are clearly the analogs of the classical Seeley-de Witt coefficients 
corresponding to induced gravity. 
However, there are also terms such as 
$\Tr \cL_{10,C} \sim \frac{1}{(2\pi)^2}\intg \Box X^a \Box_g X_a$ which correspond to extrinsic curvature
on $\cM \subset \R^D$. 
The structure is consistent with the semi-classical results obtained in 
\cite{Klammer:2009dj}, and a more precise result should be given
once the remaining coefficients for $\cL_{\rm cubic}$ are known.

To summarize, we found that
the induced action due to fermions induces not only intrinsic curvature terms on $\cM^4$
including a generalized Einstein-Hilbert action with a dilaton,
but also terms which correspond to extrinsic curvature of the embedding 
$\cM^4 \subset \R^D$.
These might be cancelled by the contributions from the bosonic fields,
but some extrinsic terms are expected to survive, in particular because
the $\cN=4$ SUSY must be broken at low energies. Such terms would definitely 
go beyond general relativity, preferring flat vacuum geometries. 
This might be very interesting e.g. in the context of 
cosmology, however their physical relevance remains to be clarified.

\section{Remarks and conclusion}

In this paper we have computed the effective action in the Yang-Mills type matrix model 
induced upon integrating out the fermions, using a heat-kernel expansion.
There are two important new results: First, we 
point out that one should \emph{not} simply take the {\naiv} UV limit,
in contrast to (most of the) previous work on the heat-kernel expansion on NC spaces.
Rather, a specific IR condition \eq{IR-condition}  on the external momentum scales and the cutoff
should be imposed, cf.~\cite{Steinacker:2008}. Only then a non-trivial 
and robust induced effective action is obtained. The reason is that the induced action is 
entirely due to non-commutativity (for the $U(1)$ sector under consideration here) 
resp. UV/IR mixing, which becomes ill-defined
in the strict UV limit. Such a finite cutoff should be realized notably in 
maximally supersymmetric models (i.e. the IKKT model) and close relatives,
via some SUSY breaking scale. This result makes perfect sense from the 
emergent gravity point of view, which is relevant for the IR.

The second important result is that the effective action can be 
written as a generalized \emph{matrix model}, with manifest $SO(D)$ symmetry. 
We obtain an explicit form for this matrix model, which is conjectured 
to capture the full contribution from simple-commutator (potential) terms, and the 
leading contribution for the double commutator (curvature) terms. 
In particular, the geometrical insights gained from the geometrical point of view of
emergent gravity allow to correctly predict the potential part of this effective action,
thus predicting a series of highly non-trivial loop computations. 
This is very remarkable from the gauge theory point of view. It suggests that 
the effective non-Abelian gauge theory action induced on coinciding branes can also be computed 
efficiently by taking advantage of the full $SO(D)$ symmetry in the matrix model. 
Thus generalized effective matrix models should provide a powerful new tool
in the context of matrix models, gauge theory and gravity. 

Finally, the effective $SO(D)$ symmetric matrix model can be translated into the geometrical 
action for emergent gravity. This is an essential step towards a physical understanding of the
emergent gravity which arises on generic 4D branes in this model. Supplemented with 
the contributions induced by the bosonic modes, a systematic analysis of the resulting 
physics should become possible.

The results presented here clearly are incomplete, and a much more systematic study of the 
induced effective matrix model should be carried out. While the potential 
terms have passed highly non-trivial checks, the curvature terms were only marginally
resp. incompletely determined, 
and should be determined and verified in a more complete computation. 
Many aspects require a more detailed 
understanding. For example,
the analysis of the potential \eq{vacuuum-energy-matrix} in the 
effective matrix model
was based on the 4D relation \eq{char-4D}. This relation was derived only in the 
semi-classical limit, and is not expected to hold in the fully NC case. 
A related point is that the matrix model also admits backgrounds with different dimensions, 
around which the effective action would look very different. Therefore, the 
effective matrix model \eq{full-action-MM} should be seen as ``leading part'' of some more
complete, ``universal'' effective action which holds for any background. These are only some 
issues in an interesting new line of research.

\subsection*{Acknowledgements}

We would like to thank H. Grosse, F. Lizzi,  N. Sasakura, A. Schenkel and R. Szabo for useful discussions. 
This work was supported by the Austrian Science Fund (FWF) under contracts 
P21610-N16, 
P20507-N16 and P20017-N16. 

\startappendix
\Appendix{Perturbative expansion of the heat kernel}
\label{app:regularization}

Consider the one-loop effective action in terms of 
the gauge field $A$ (cf~\cite{Steinacker:2008}, App. B):
\[
\Gamma = 
\frac 12 \Tr \Big(\log\frac 12\Delta_{A}  - \log\frac 12\Delta_0\Big)
\equiv\,\, -\frac 12\Tr\int_{0}^{\infty} \frac{d\a}{\a}\,
\Big(e^{-\a\frac 12\Delta_{A}} - e^{-\a\frac 12\Delta_0 }\Big)\, 
e^{- \frac 1{\a L^2}} 
\]
where the small $\a$ divergence is regularized as in \eq{TrLog-id}
using a UV cutoff $ L$.
To obtain the expansion in $A$ we use the Duhamel formula 
(cf. \cite{Grosse:2007})
\begin{align}
& \left(e^{-\a H}  - e^{-\a H_0} \right) = \nn\\*
&= -  \int_0^\a d t_1 e^{-t_1 H_0} V e^{-(\a-t_1) H_0} 
+ \int_0^\a d t_1 \int_0^{t_1} d t_2 
e^{-t_2 H_0} V e^{-(t_1-t_2) H_0} V e^{-(\a-t_1) H_0} \nn\\
&\quad - \int_0^\a d t_1 \int_0^{t_1} d t_2 \int_0^{t_2}dt_3 e^{-t_3 H_0} V e^{-(t_2-t_3) H_0} V e^{-(t_1-t_2) H_0} V e^{-(\a-t_1) H_0} \nn\\
&\quad + \int\limits_0^\a\!dt_1 \int\limits_0^{t_1}\!dt_2 \int\limits_0^{t_2}\!dt_3\int\limits_0^{t_3}\!dt_4 e^{-t_4 H_0} V e^{-(t_3-t_4) H_0} V e^{-(t_2-t_3) H_0} V e^{-(t_1-t_2) H_0} V e^{-(\a-t_1) H_0} + \dots
\end{align}
where $H = H_0 + V$. 
The second-order term can be written as 
\begin{align}
& \int_0^\infty \frac{d\a}{\a}\,\int_0^\a d t_1 \int_0^{t_1} d t_2 
\Tr\Big(e^{-t_2 H_0} V e^{-(t_1-t_2) H_0} V e^{-(\a-t_1) H_0}\Big)\, 
e^{- \frac 1{\a L^2}} \nn\\
& = \int_0^\infty \frac{d\a}{\a}\, \int_0^{\a} d t'\,\int_{t'}^\a dt_1 
\Tr\Big( V e^{-t' H_0} V e^{-(\a-t') H_0}\Big)\, 
e^{- \frac 1{\a L^2}} \nn\\
& = \int_0^\infty \frac{d\a}{\a}\, \int_0^{\a} d t''\,t''
\Tr\Big(V e^{-t'' H_0} V e^{-(\a-t'') H_0} \Big)\, 
e^{- \inv{\a L^2}} \,,
\end{align}
where $t' = t_1-t_2$ and $t'' = \a-t'$. Combining the two last 
lines we obtain
\begin{align}
& \int_0^\infty \frac{d\a}{\a}\,\int_0^\a d t_1 \int_0^{t_1} d t_2 
\Tr\Big(e^{-t_2 H_0} V e^{-(t_1-t_2) H_0} V e^{-(\a-t_1) H_0}\Big)\, 
e^{- \frac 1{\a L^2}} \nn\\
& = \frac {1}2\,\int_0^\infty d\a\, \int_0^{\a} d t''
\Tr\Big(V e^{-t'' H_0} V e^{-(\a-t'') H_0} \Big)\, 
e^{- \frac 1{\a L^2}} \,.
\end{align}
Hence one finds
\begin{align}
\Gamma  &= 
\inv2\, \int_0^\infty d\a\, \Tr (V e^{-\a H_0}) e^{- \frac 1{\a L^2}}
- \inv4 \int_0^\infty d\a \int_0^\a dt'\,\,
    \Tr \Big(V e^{-t'H_0} V e^{-(\a-t')H_0}\Big) e^{- \frac 1{\a L^2}} \nn\\
 &\quad + \inv2  \int\limits_0^\infty\frac{d\a}{\a}\int\limits_0^\a dt'\int\limits_0^{t'}dt'' t'' \,\Tr \Big(V e^{-t''H_0} V e^{-(t'-t'')H_0} V e^{-(\a-t')H_0}\Big) e^{- \inv{\a L^2}} \nn\\
 &\quad - \inv2 \int\limits_0^\infty\frac{d\a}{\a}\int\limits_0^\a\!dt'\int\limits_0^{t'}\!dt'' \int\limits_0^{t''}\!dt''' t''' \Tr \Big(V e^{-t'''H_0} V e^{-(t''-t''')H_0} V e^{-(t'-t'')H_0} V e^{-(\a-t')H_0}\Big) e^{\frac{-1}{\a L^2}} \nn\\
 &\quad + \ldots  \label{general-eff-action}
\end{align}
While gauge invariance is typically not preserved by this expansion, 
it must be recovered upon collecting all contributions at any given order in  $\L$. 
It may also be convenient to introduce a test function $f$, cf. ~\cite{Vassilevich:2003xt}.

\Appendix{Heat kernel expansion: four field contributions}
\label{app:heatkernel-four-field}
In order to compute the four field contributions, a heat kernel expansion up to fourth order is required. Nonetheless, we would like to state here a partial result coming from the second order in the expansion in order to demonstrate which type of terms are likely to appear. 
Following the lines of \secref{sec:order-by-order} the four field contributions of \eqref{eq:second-order-terms-compact} compute to:
\begin{align}
&\int\limits_0^\infty\! d\a \int\limits_0^\a\! dt\,\,
\Tr\left(V e^{-tH_0}Ve^{-(\a-t)H_0}\right)e^{- \frac 1{\a L^2}}\Big|_{\textrm{4-fields}} \nn\\*
&\approx \frac{\tr\one\sqrt{\bG}}{4\LNC^4}\int\!\frac{d^4Q\,d^4l\,d^4k}{(2\pi\LNC^2)^6}\Bigg\{ 
  \! \Bigg(\! \bG^{\m\n}\bG^{\r\s}A_\m(Q-l)A_\r(-Q-k)A_\n(l)A_\s(k) \nn\\*
&\quad \qquad +\vph^i(Q-l)\vph^j(-Q-k)\vph_i(l)\vph_j(k)+\bG^{\m\n}A_\m(Q-l)A_\n(l)\vph^i(-Q-k)\vph_i(k) \Bigg)\!\times \nn\\*
&\quad \times\! \Bigg(\! \L^2 \bigg( \lt \cdot (\lt +\Qt ) 
  \!\left(\cos \!\left(\tfrac{(k \th  Q)}{2} \right)\!-\!1\!\right)
   -\Qt \!\cdot\! \Qt \sin \!\left(\tfrac{ (2 k-l) \th  Q}{4}\right)\!\sin \!\left(\tfrac{(l \th  Q)}{4} \right) \nn\\*
&\qquad\qquad - \left(2 \kt \cdot (\kt +\Qt )+\Qt \cdot
   \Qt \right) \sin^2\! \left(\tfrac{(l \th Q)}{4} \right)\bigg) \nn\\*
&\quad\qquad +16 \sin^2\!\left(\tfrac{(k \th  Q)}{4} \right) \sin^2\!\left(\tfrac{(l \th  Q)}{4} \right) 
   \left(\ln \!\left(\!\frac{\L^2}{Q\cdot Q}\!\right)-2 \gamma_E +2\right)\!\! \Bigg)  
  + \left(\!\!\begin{array}{c} Q\to-Q \\ l\leftrightarrow k \end{array}\!\!\right) + \cO\!\left(\!\frac{Q^4}{\L^4}\!\right) \Bigg\}
. \label{4-field-collect}
\end{align}

\Appendix{Supplemental computations for the effective matrix model action}
\label{app:supl-calc}
\SubAppendix{Proof of Lemma \ref{lemma-1}}{app:proof-of-lemma-1}

\eqnref{FF-theta} can be shown using the fact that 
$\hat\th^{ij} \hat\th^{kl} - \hat\th^{il} \hat\th^{kj} - \hat\th^{lj} \hat\th^{ki}$ 
is totally anti-symmetric  (cf. \cite{Steinacker:2007dq}).
We verify the tensor identity \eq{strange-id} by elaborating both sides in a
local basis where $\bar\th^{\mu\nu}$ has canonical form 
\eq{theta-standard-general-E} and $g_{\mu\nu} = \d_{\mu\nu}$. Then
\eqn{
(\bG g) & = \diag(\a^2,\a^2,\a^{-2},\a^{-2}) \,, \nn\\
\bG g  - \frac 14 (\bG g)\d &= \frac 12\diag(\e,\e,-\e,-\e)\,, \qquad 
\epsilon = \a^2-\a^{-2}
\,. \label{Gg-1}
}
Now the lhs of \eq{strange-id} is
\begin{align}
(\textrm{lhs}) &= \frac \e2(\a^4 F_{12} F_{12} + \a^4 F_{21} F_{21}) -
\frac \e2 (\a^{-4} F_{34} F_{34} + \a^{-4} F_{43} F_{43})
- \frac \e2( F_{13} F_{13} - F_{31} F_{31}) \nn\\
& \quad
- \frac \e2( F_{14} F_{14} - F_{41} F_{41})
- \frac \e2( F_{23} F_{23} - F_{32} F_{32})
- \frac \e2( F_{24} F_{24} - F_{42} F_{42})
\nn\\
&= \e (\a^4 F_{12} F_{12} - \a^{-4} F_{34} F_{34})
\,. 
\end{align}
This agrees with the rhs of \eq{strange-id}, since
\eqn{
(\bar\th F)^2   &= 4 \LNC^{-4} (\a F_{12} \pm \a^{-1} F_{34})^2  \,,\nn\\
(\hat\th F)^2   &= 4 \LNC^{-4} (\a^3 F_{12} \pm \a^{-3} F_{34})^2  \,,\nn\\
(\hat\th F)^2 - (\bar\th F)^2  &= - 4 \LNC^{-4} (\a F_{12} \pm \a^{-1} F_{34})^2  
 + 4 \LNC^{-4} (\a^3 F_{12} \pm \a^{-3} F_{34})^2  \nn\\
&= 4 \LNC^{-4} \e (\a^4 F_{12} F_{12} - \a^{-4} F_{34} F_{34}) 
\,.
}
Now consider \eq{strange-id-2}. It is easy to see using \eq{Gg-1} that the
lhs vanishes (assuming again that $\bar\theta^{\mu\nu}$ has canonical form)
if $\mu$ and $\nu$ live in different blocks $(1,2)$ resp. $(3,2)$.
Furthermore, $\th^{\mu\mu'} F_{\mu'\nu}$ vanishes if $\mu\neq\nu$ live in the same 
blocks $(1,2)$ resp. $(3,2)$. Hence it is sufficient to check the relation for 
the diagonal elements $\mu=\nu$, which can be checked explicitly using
e.g. 
\eqn{
(G F \hat \th)^{11} &= (G F \hat \th)^{22} =  (\a^5-\a)F_{12} \,,  \nn\\
\qquad (G F \hat \th)^{11} &= (G F \hat \th)^{22} = \pm (\a^{-5}-\a^{-1})F_{34}
\,, 
}
which agrees with the rhs of \eq{strange-id-2}.

\SubAppendix{Proof of Lemma \ref{lemma-2}}{app:proof-of-lemma-2}

Using the Bianchi identity and partial integration, we have
\begin{align}
&  \bG^{\mu\nu}\del_{\mu} F_{\nu\g} g^{\rho\b} \del_\rho F_{\b\m'} H^{\g\m'}  
=  \bG^{\mu\nu} \del_\rho F_{\nu\g} g^{\rho\b} \del_{\mu} F_{\b\m'} H^{\g\m'}  \quad + \del()    \nn\\
&= - \bG^{\mu\nu} (\del_\nu F_{\g\rho} +  \del_\g F_{\rho\nu}) g^{\rho\b} \del_{\mu} F_{\b\m'} H^{\g\m'}   \quad + \del()   \nn\\
&=   \( F_{\g\rho} g^{\rho\b} (\bG^{\mu\nu} \del_\nu\del_{\mu}) F_{\b\m'} H^{\g\m'}   
 - \bG^{\mu\nu} \del_{\mu} F_{\rho\nu} g^{\rho\b} \del_\g F_{\b\m'} H^{\g\m'}   \)  \quad + \del()  \nn\\
&= \( F_{\g\rho} g^{\rho\b} (\bG^{\mu\nu} \del_\nu\del_{\mu}) F_{\b\m'} H^{\g\m'}   
 - \bG^{\mu\nu} \del_{\mu} F_{\nu\g} H^{\rho\b}\del_\rho F_{\b\mu'}  g^{\g\mu'}   \) \quad + \del()  
\end{align}
for any symmetric matrix $H^{\mu\nu}$. 
Hence we get the identities
\eqn{
 \bG^{\mu\nu}\del_{\mu} F_{\nu\g} (g^{\rho\b} \del_\rho F_{\b\m'}\bG^{\g\m'}  + \bG^{\rho\b}\del_\rho F_{\b\mu'}  g^{\g\mu'})
 &= -   (\bG^{\mu\nu} \del_\nu\del_{\mu}) F_{\rho\g}  F_{\b\m'} g^{\rho\b} \bG^{\g\m'} \quad + \del() \,, \nn\\
 2 g^{\mu\nu}\del_{\mu} F_{\nu\g} \bG^{\rho\b} \del_\rho F_{\b\m'}\bG^{\g\m'}  
 &= -  (g^{\mu\nu} \del_\nu\del_{\mu}) F_{\rho\g}  F_{\b\m'} \bG^{\rho\b} \bG^{\g\m'}  \quad + \del()  
 \,. \nn }
Combining the first two gives \eq{FFbox-id-2}.

\SubAppendix{Computation for \eqnref{S-10-curv}}{app:S-10-curv}
{\allowdisplaybreaks
Using \eqref{H-explicit-fluct}, we find
\begin{align}
& \co{X_f}{H^{ab}}\co{X^f}{H_{ab}}
=\LNC^{-4}\left(\co{\varphi^i}{H^{ab}}\co{\varphi_i}{H_{ab}}-\bG^{\m\n}D_\m H^{ab}D_\n H_{ab}\right) \nn\\*
&= \LNC^{-12}g_{\m\r}g_{\n\s}\co{\varphi^i}{\bG^{\m\a}\bar\th^{\n\n'}F_{\n'\a}+\m\leftrightarrow\n}\co{\varphi_i}{\bG^{\r\b}\bar\th^{\s\s'}F_{\s'\b}+\r\leftrightarrow\s} \nn\\
&\quad - \LNC^{-12}g_{\m\r}g_{\n\s}\bG^{\a\b}D_\a \left(\bG^{\m\a}\bar\th^{\n\n'}F_{\n'\a}+\m\leftrightarrow\n\right)
 D_\b\left(\bG^{\r\b}\bar\th^{\s\s'}F_{\s'\b}+\r\leftrightarrow\s\right) \nn\\
&\quad + 2\LNC^{-16}(\bG g\bG)^{\e\t}\left(\co{\varphi^i}{D_\e\varphi^k}\co{\varphi_i}{D_\t\varphi_k}-\bG^{\m\n}D_\m D_\e\varphi^kD_\n D_\t\varphi_k\right) +\textrm{h.o.} \nn\\
&= 2\LNC^{-12}\left(\co{\varphi^i}{F_{\m\n}}\co{\varphi_i}{F_{\r\s}}-\bG^{\a\b}D_\a F_{\m\n}D_\b F_{\r\s}\right)\left(\LNC^{-4}(\bG g\bG)^{\m\r}\bG^{\n\s}+\hat\th^{\m\r}\hat\th^{\n\s}\right) \nn\\*
&\quad + 2\LNC^{-16}(\bG g\bG)^{\e\t}\left(\co{\varphi^i}{D_\e\varphi^k}\co{\varphi_i}{D_\t\varphi_k}
  -\bG^{\m\n}D_\m D_\e\varphi^kD_\n D_\t\varphi_k\right) +\textrm{h.o.} 
\,.
\end{align}
Thus
\begin{align}
& \Tr \co{X_f}{H^{ab}}\co{X^f}{H_{ab}} \nn\\*
&= \frac2{\LNC^{8}}\int\!\!\frac{d^4x}{(2\pi)^2} \sqrt{\bG} 
\left(\bar\Box\varphi^k (\bG g\bG)^{\e\t} D_\e D_\t\varphi_k
 - F_{\m\n} \bar\Box F_{\r\s} (\bG g\bG)^{\m\r}\bG^{\n\s}
 - \LNC^{4}  F_{\m\n} \hat\th^{\n\s}\bar\Box F_{\s\r}\hat\th^{\r\m}\!\right) 
  \nn\\* & \quad\qquad +\textrm{h.o.} \,,  \nn\\
&= 2\LNC^{-8}\int\!\!\frac{d^4x}{(2\pi)^2} \sqrt{\bG} 
\Big(\frac 14 \LNC^4 \bar H  F_{\m\n} \bar\Box F_{\r\s} \bG^{\m\r}\bG^{\n\s}
 - \frac 12 \LNC^{4} (\hat\th F) \bar\Box(\hat\th F) \nn\\*
& \quad   -\frac{\LNC^4}4 \left(\!(\hat\th F)\bar\Box(\hat\th F) 
- (\bar\th F)\bar\Box(\bar\th F)\!\right)
 + \bar\Box\varphi^k \big(\LNC^8\bar\Box_g\varphi_k-\frac 12 \LNC^4 (\bG g) \bar\Box\varphi_k \big)  \Big)  +\textrm{h.o.} ,  
\end{align}
using \eq{FF-theta} and \eq{strange-id} (multiplied by $q\cdot q$ in momentum representation) 
where $\bar H\sim -(\bG g)$. 
}

{\allowdisplaybreaks
Similarly, using \eqref{H-explicit-fluct} and $\bG\bar\th^{-1}\bG=\LNC^4\hat\th$, we find
\begin{align}
& H^{cd}\co{X_c}{J^a_{\ b}}\co{X_d}{J^b_{\ a}}=H^{cd}\co{X_c}{\co{X^a}{X_b}}\co{X_d}{\co{X^b}{X_a}} 
=H^{cd}\co{X_c}{\Th^{ab}}\co{X_d}{\Th_{ab}} \nn\\*
&= H^{\m\n}\co{X_\m}{\Th^{ab}}\co{X_\n}{\Th_{ab}} + 2H^{\m i}\co{X_\m}{\Th^{ab}}\co{\phi_i}{\Th_{ab}} + H^{ij}\co{\phi_i}{\Th^{ab}}\co{\phi_j}{\Th_{ab}} \nn\\
&= \LNC^{-16}(\bG g\bG)^{\m\n}\bG^{\r\s}\left(\bG^{\e\t}D_\m F_{\r\e}D_\n F_{\s\t} 
 + 2D_\m D_\r\varphi^i D_\n D_\s\varphi_i \right) + \textrm{h.o.}
\,. 
\end{align}
In particular, 
}
\begin{align}
& (H^{cd}- \inv2 H g^{cd})\co{X_c}{\Th^{ab}}\co{X_d}{\Th_{ab}} \nn\\
&\sim \LNC^{-16}((\bG g\bG) - \frac 12 (\bG g)G)^{\m\n}\bG^{\r\s}\left(\bG^{\e\t}D_\m F_{\r\e}D_\n F_{\s\t} 
 + 2D_\m D_\r\varphi^i D_\n D_\s\varphi_i \right) + \textrm{h.o.}     \nn\\
&= -\LNC^{-16} g^{\m\n}\bG^{\r\s}\left(\bG^{\e\t}D_\m F_{\r\e}D_\n F_{\s\t} 
 + 2D_\m D_\r\varphi^i D_\n D_\s\varphi_i \right) + \textrm{h.o.}    
\label{thetaboxtheta-explicit}
\end{align}
Thus
\begin{align}
& \Tr (H^{cd}- \inv2 H g^{cd})\co{X_c}{\Th^{ab}}\co{X_d}{\Th_{ab}} \nn\\
& = \LNC^{-12}\int\!\!\frac{d^4x}{(2\pi)^2} \sqrt{\bG} \Big(\bG^{\r\s}\bG^{\e\t} F_{\r\e} (g^{\m\n}D_\m D_\n F_{\s\t}) 
 + 2 \LNC^{4} g^{\m\n} D_\m D_\n\varphi^i  \bar\Box\varphi_i \Big) + \textrm{h.o.}   \nn\\
&= \LNC^{-12}\int\!\!\frac{d^4x}{(2\pi)^2} \sqrt{\bG} 
  \Big(-\LNC^8 \bG^{\r\s}\bG^{\e\t} F_{\r\e} \bar\Box_g F_{\s\t} 
 - 2\LNC^{12}\bar\Box_g\varphi^i \bar\Box\varphi_i \Big) + \textrm{h.o.}  
\end{align}
which is consistent with the $\cO(X^6)$ results \eq{O6-quadratic}.


Furthermore, a straightforward computation gives
\begin{align}
\Box X^a \Box_g X^a 
&=  \LNC^{-16} \bG^{\mu\nu}\del_{\mu} F_{\nu\g}
 \Big( g^{\rho\b} \del_\rho F_{\b\m'} \bG^{\g\m'}
 + (GgG)^{\g\b} \bG^{\eta\mu'} \del_\eta F_{\mu'\b} 
 -\frac{\LNC^4}2 \hat\theta^{\g\nu'} \del_{\nu'} (\hat\th^{\b\eta} F_{\b\eta}) \! \Big)  \nn\\
&\quad + \Box_g \phi^i \bar\Box_g \phi_i 
\,. 
\end{align}
Under the integral resp. trace this can be rewritten as
\begin{align}
\Tr \Box X^a \Box_g X^a 
&= \Tr \LNC^{-16} \bG^{\mu\nu}\del_{\mu} F_{\nu\g}
 \Big( g^{\rho\b} \del_\rho F_{\b\m'} \bG^{\g\m'}
 +  \bG^{\rho\b} \del_\rho F_{\b\mu'} (GgG)^{\g\mu'} \Big) \nn\\
 &\quad  +\frac 14 \LNC^{-12}  \bG^{\mu\nu} \del_{\nu} (\hat\theta F)\del_{\mu}(\hat\th F) 
   + \Box_g \phi^i \bar\Box_g \phi_i 
\,, 
\end{align}
noting that the Bianchi identity gives
\be
\hat\theta^{\g\nu'} \del_{\nu'}F_{\nu\g} = -\frac 12  \del_{\nu} (\hat\theta^{\g\nu'}F_{\g\nu'}) 
\,. 
\ee
Now we can use the second relation \eq{FFbox-id-2} of Lemma \ref{lemma-2}, 
replacing $g$ by $(GgG)$. 
This gives after some (by now standard) manipulations
{\allowdisplaybreaks
\begin{align}
& \Tr \Box X^a \Box_g X^a \nn\\*
&=  \Tr \LNC^{-8} 
 \Big( \bar\Box_g F_{\rho\g}  F_{\b\m'} \bG^{\rho\b} \bG^{\g\m'} 
 + \frac 12(\hat\theta F) \bar\Box (\hat\theta F) 
 -  \frac 14(\bar\theta F)\bar\Box(\bar\theta F)
  + \LNC^4 \Box_g \varphi^i \bar\Box_g \varphi_i  \Big) .
\end{align}
The other terms of \eqref{S-10-curv} can be computed along the same lines.
}


\bibliographystyle{../custom1.bst}
\bibliography{../articles.bib,../books.bib}

\begin{thebibliography}{10}
\expandafter\ifx\csname url\endcsname\relax
  \def\url#1{\texttt{#1}}\fi
\expandafter\ifx\csname urlprefix\endcsname\relax\def\urlprefix{\\URL }\fi
\providecommand{\eprint}[2][]{\url{#2}}
\small\itemsep=3pt
\tolerance 1414
\hbadness 1414
\emergencystretch 1.5em
\hfuzz 0.3pt
\widowpenalty=10000
\vfuzz \hfuzz
\raggedbottom

\bibitem{Aoki:1999vr}
H.~Aoki, N.~Ishibashi, S.~Iso, H.~Kawai, Y.~Kitazawa and T.~Tada,
  \textit{{Noncommutative Yang-Mills in IIB matrix model}}, \textit{Nucl.
  Phys.} \textbf{B565} (2000) 176--192,
  \href{http://www.arxiv.org/abs/hep-th/9908141}{\texttt{[arXiv:hep-th/9908141%
]}}.

\bibitem{Steinacker:2007dq}
H.~Steinacker, \textit{{Emergent Gravity from Noncommutative Gauge Theory}},
  \textit{JHEP} \textbf{12} (2007) 049,
  \href{http://www.arxiv.org/abs/0708.2426}{\texttt{[arXiv:0708.2426]}}.

\bibitem{Steinacker:2008ri}
H.~Steinacker, \textit{{Emergent Gravity and Noncommutative Branes from
  Yang-Mills Matrix Models}}, \textit{Nucl. Phys.} \textbf{B810} (2009) 1--39,
  \href{http://www.arxiv.org/abs/0806.2032}{\texttt{[arXiv:0806.2032]}}.

\bibitem{Steinacker:2010rh}
H.~Steinacker, \textit{{Emergent Geometry and Gravity from Matrix Models: an
  Introduction}}, \textit{Class. Quant. Grav.} \textbf{27} (2010) 133001,
  \href{http://www.arxiv.org/abs/1003.4134}{\texttt{[arXiv:1003.4134]}}.

\bibitem{Sakharov:1967}
A.~D. Sakharov, \textit{Vacuum quantum fluctuations in curved space and the
  theory of gravitation}, \textit{Sov. Phys. Dokl.} \textbf{12} (1968)
  1040--1041.

\bibitem{Ishibashi:1996xs}
N.~Ishibashi, H.~Kawai, Y.~Kitazawa and A.~Tsuchiya, \textit{{A large-$N$
  reduced model as superstring}}, \textit{Nucl. Phys.} \textbf{B498} (1997)
  467--491,
  \href{http://www.arxiv.org/abs/hep-th/9612115}{\texttt{[arXiv:hep-th/9612115%
]}}.

\bibitem{Imai:2003jb}
T.~Imai, Y.~Kitazawa, Y.~Takayama and D.~Tomino, \textit{{Effective actions of
  matrix models on homogeneous spaces}}, \textit{Nucl. Phys.} \textbf{B679}
  (2004) 143--167,
  \href{http://www.arxiv.org/abs/hep-th/0307007}{\texttt{[arXiv:hep-th/0307007%
]}}.

\bibitem{CastroVillarreal:2005uu}
P.~Castro-Villarreal, R.~Delgadillo-Blando and B.~Ydri, \textit{{Quantum
  effective potential for $U(1)$ fields on $S_L^2 \times S_L^2$}},
  \textit{JHEP} \textbf{09} (2005) 066,
  \href{http://www.arxiv.org/abs/hep-th/0506044}{\texttt{[arXiv:hep-th/0506044%
]}}.

\bibitem{Steinacker:2003sd}
H.~Steinacker, \textit{{Quantized gauge theory on the fuzzy sphere as random
  matrix model}}, \textit{Nucl. Phys.} \textbf{B679} (2004) 66--98,
  \href{http://www.arxiv.org/abs/hep-th/0307075}{\texttt{[arXiv:hep-th/0307075%
]}}.

\bibitem{Azuma:2005pm}
T.~Azuma, S.~Bal, K.~Nagao and J.~Nishimura, \textit{{Perturbative versus
  nonperturbative dynamics of the fuzzy $S^2 \times S^2$}}, \textit{JHEP}
  \textbf{09} (2005) 047,
  \href{http://www.arxiv.org/abs/hep-th/0506205}{\texttt{[arXiv:hep-th/0506205%
]}}.

\bibitem{Steinacker:2008}
H.~Grosse, H.~Steinacker and M.~Wohlgenannt, \textit{{Emergent Gravity, Matrix
  Models and UV/IR Mixing}}, \textit{JHEP} \textbf{04} (2008) 023,
  \href{http://www.arxiv.org/abs/0802.0973}{\texttt{[arXiv:0802.0973]}}.

\bibitem{Gayral:2006vd}
V.~Gayral, B.~Iochum and D.~V. Vassilevich, \textit{{Heat kernel and number
  theory on NC-torus}}, \textit{Commun. Math. Phys.} \textbf{273} (2007)
  415--443,
  \href{http://www.arxiv.org/abs/hep-th/0607078}{\texttt{[arXiv:hep-th/0607078%
]}}.

\bibitem{Vassilevich:2005vk}
D.~V. Vassilevich, \textit{{Heat kernel, effective action and anomalies in
  noncommutative theories}}, \textit{JHEP} \textbf{08} (2005) 085,
  \href{http://www.arxiv.org/abs/hep-th/0507123}{\texttt{[arXiv:hep-th/0507123%
]}}.

\bibitem{Armoni:2000xr}
A.~Armoni, \textit{{Comments on perturbative dynamics of non-commutative Yang-
  Mills theory}}, \textit{Nucl. Phys.} \textbf{B593} (2001) 229--242,
  \href{http://www.arxiv.org/abs/hep-th/0005208}{\texttt{[arXiv:hep-th/0005208%
]}}.

\bibitem{Drummond:2009fd}
J.~M. Drummond, J.~M. Henn and J.~Plefka, \textit{{Yangian symmetry of
  scattering amplitudes in $N=4$ super Yang-Mills theory}}, \textit{JHEP}
  \textbf{05} (2009) 046,
  \href{http://www.arxiv.org/abs/0902.2987}{\texttt{[arXiv:0902.2987]}}.

\bibitem{Beisert:2010gn}
N.~Beisert, J.~Henn, T.~McLoughlin and J.~Plefka, \textit{{One-Loop
  Superconformal and Yangian Symmetries of Scattering Amplitudes in $N=4$ Super
  Yang-Mills}}, \textit{JHEP} \textbf{04} (2010) 085,
  \href{http://www.arxiv.org/abs/1002.1733}{\texttt{[arXiv:1002.1733]}}.

\bibitem{Bern:2004bt}
Z.~Bern, L.~J. Dixon and D.~A. Kosower, \textit{{All next-to-maximally
  helicity-violating one-loop gluon amplitudes in $N = 4$ super-Yang-Mills
  theory}}, \textit{Phys. Rev.} \textbf{D72} (2005) 045014,
  \href{http://www.arxiv.org/abs/hep-th/0412210}{\texttt{[arXiv:hep-th/0412210%
]}}.

\bibitem{ArkaniHamed:2010kv}
N.~Arkani-Hamed, J.~L. Bourjaily, F.~Cachazo, S.~Caron-Huot and J.~Trnka,
  \textit{{The All-Loop Integrand For Scattering Amplitudes in Planar $N=4$
  SYM}} \href{http://www.arxiv.org/abs/1008.2958}{\texttt{[arXiv:1008.2958]}}.

\bibitem{Chamseddine:1996zu}
A.~H. Chamseddine and A.~Connes, \textit{{The spectral action principle}},
  \textit{Commun. Math. Phys.} \textbf{186} (1997) 731--750,
  \href{http://www.arxiv.org/abs/hep-th/9606001}{\texttt{[arXiv:hep-th/9606001%
]}}.

\bibitem{Steinacker:2008a}
D.~Klammer and H.~Steinacker, \textit{{Fermions and Emergent Noncommutative
  Gravity}}, \textit{JHEP} \textbf{08} (2008) 074,
  \href{http://www.arxiv.org/abs/0805.1157}{\texttt{[arXiv:0805.1157]}}.

\bibitem{Klammer:2009dj}
D.~Klammer and H.~Steinacker, \textit{{Fermions and noncommutative emergent
  gravity II: Curved branes in extra dimensions}}, \textit{JHEP} \textbf{02}
  (2010) 074,
  \href{http://www.arxiv.org/abs/0909.5298}{\texttt{[arXiv:0909.5298]}}.

\bibitem{Gilkey:1995mj}
P.~B. Gilkey, \textit{{Invariance theory, the heat equation and the
  Atiyah-Singer index theorem}}, Boca Raton: CRC Press Inc., second edition,
  1995.

\bibitem{Vassilevich:2003xt}
D.~V. Vassilevich, \textit{{Heat kernel expansion: User's manual}},
  \textit{Phys. Rept.} \textbf{388} (2003) 279--360,
  \href{http://www.arxiv.org/abs/hep-th/0306138}{\texttt{[arXiv:hep-th/0306138%
]}}.

\bibitem{Szabo:2001}
R.~J. Szabo, \textit{Quantum field theory on noncommutative spaces},
  \textit{Phys. Rept.} \textbf{378} (2003) 207--299,
  \href{http://www.arxiv.org/abs/hep-th/0109162}{\texttt{[arXiv:hep-th/0109162%
]}}.

\bibitem{Blaschke:2010kw}
D.~N. Blaschke, E.~Kronberger, R.~I.~P. Sedmik and M.~Wohlgenannt,
  \textit{{Gauge Theories on Deformed Spaces}}, \textit{SIGMA} \textbf{6}
  (2010) 062,
  \href{http://www.arxiv.org/abs/1004.2127}{\texttt{[arXiv:1004.2127]}}.

\bibitem{Jack:2001cr}
I.~Jack and D.~R.~T. Jones, \textit{{Ultra-violet finiteness in noncommutative
  supersymmetric theories}}, \textit{New J. Phys.} \textbf{3} (2001) 19,
  \href{http://www.arxiv.org/abs/hep-th/0109195}{\texttt{[arXiv:hep-th/0109195%
]}}.

\bibitem{Sasakura:2004dq}
N.~Sasakura, \textit{{Heat kernel coefficients for compact fuzzy spaces}},
  \textit{JHEP} \textbf{12} (2004) 009,
  \href{http://www.arxiv.org/abs/hep-th/0411029}{\texttt{[arXiv:hep-th/0411029%
]}}.

\bibitem{Blaschke:2010rg}
D.~N. Blaschke and H.~Steinacker, \textit{{Curvature and Gravity Actions for
  Matrix Models}}, \textit{Class. Quant. Grav.} \textbf{27} (2010) 165010,
  \href{http://www.arxiv.org/abs/1003.4132}{\texttt{[arXiv:1003.4132]}}.

\bibitem{Blaschke:2010qj}
D.~N. Blaschke and H.~Steinacker, \textit{{Curvature and Gravity Actions for
  Matrix Models II: the case of general Poisson structure}}, \textit{Class.
  Quant. Grav.} \textbf{27} (2010) 235019,
  \href{http://www.arxiv.org/abs/1007.2729}{\texttt{[arXiv:1007.2729]}}.

\bibitem{Steinacker:2008ya}
H.~Steinacker, \textit{{Covariant Field Equations, Gauge Fields and
  Conservation Laws from Yang-Mills Matrix Models}}, \textit{JHEP} \textbf{02}
  (2009) 044,
  \href{http://www.arxiv.org/abs/0812.3761}{\texttt{[arXiv:0812.3761]}}.

\bibitem{Grosse:2007}
H.~Grosse and M.~Wohlgenannt, \textit{Induced gauge theory on a noncommutative
  space}, \textit{Eur. Phys. J.} \textbf{C52} (2007) 435--450,
  \href{http://www.arxiv.org/abs/hep-th/0703169}{\texttt{[arXiv:hep-th/0703169%
]}}.

\end{thebibliography}

\end{document}